\documentclass[a4paper,USenglish]{lipics-v2019}
\pdfoutput=1
\usepackage{bm}
\usepackage{xspace}
\usepackage[sort,compress]{cite}
\usepackage[noend]{algpseudocode}

\newcommand{\multicastadp}{\textsc{MultiCastAdp}\xspace}
\newcommand{\multicastadvadp}{\textsc{MultiCastAdvAdp}\xspace}

\newtheorem{fact}[theorem]{Fact}

\makeatletter
\algnewcommand{\LineComment}[1]{\Statex \hskip\ALG@thistlm $\blacktriangleright$ #1}
\makeatother

\title{Broadcasting Competitively against Adaptive Adversary in Multi-channel Radio Networks}
\author{Haimin Chen}{State Key Laboratory for Novel Software Technology, Nanjing University, China}{haimin.chen@smail.nju.edu.cn}{}{}
\author{Chaodong Zheng}{State Key Laboratory for Novel Software Technology, Nanjing University, China}{chaodong@nju.edu.cn}{}{}
\authorrunning{Haimin Chen \& Chaodong Zheng}
\ccsdesc[500]{Theory of computation~Distributed algorithms.}
\keywords{Broadcast, wireless networks, resource competitive algorithms.}
\hideLIPIcs

\begin{document}
\nolinenumbers
\maketitle

\begin{abstract}
Broadcasting in wireless networks is vulnerable to adversarial jamming. To thwart such behavior, \emph{resource competitive analysis} is proposed. In this framework, sending, listening, or jamming on one channel for one time slot costs one unit of energy. The adversary can employ arbitrary strategy to disrupt communication, but has a limited energy budget $T$. The honest nodes, on the other hand, aim to accomplish broadcast while spending only $o(T)$. Previous work has shown, in a $C$-channels network containing $n$ nodes, for large $T$ values, each node can receive the message in $\tilde{O}(T/C)$ time, while spending only $\tilde{O}(\sqrt{T/n})$ energy. However, these multi-channel algorithms only work for certain values of $n$ and $C$, and can only tolerate an oblivious adversary.

In this work, we provide new upper and lower bounds for broadcasting in multi-channel radio networks, from the perspective of resource competitiveness. Our algorithms work for arbitrary $n,C$ values, require minimal prior knowledge, and can tolerate a powerful adaptive adversary. More specifically, in our algorithms, for large $T$ values, each node's runtime is $O(T/C)$, and each node's energy cost is $\tilde{O}(\sqrt{T/n})$. We also complement algorithmic results with lower bounds, proving both the time complexity and the energy complexity of our algorithms are optimal or near-optimal (within a poly-log factor). Our technical contributions lie in using ``epidemic broadcast'' to achieve time efficiency and resource competitiveness, and employing coupling techniques in the analysis to handle the adaptivity of the adversary. At the lower bound side, we first derive a new energy complexity lower bound for 1-to-1 communication in the multi-channel setting, and then apply simulation and reduction arguments to obtain the desired result.
\end{abstract}
\clearpage

%%% Main body of the paper

\section{Introduction}\label{sec-intro}

Consider a synchronous, time-slotted, single-hop wireless network formed by $n$ devices (or, simply \emph{nodes}). Each node is equipped with a radio transceiver, and these nodes communicate over a shared wireless medium containing $C$ channels. In each time slot, each node can operate on one arbitrary channel, but cannot send and listen simultaneously. In this model, we study a fundamental communication problem---broadcasting---in which a designated \emph{source} node wants to disseminate a message $m$ to all other nodes in the network.

Lots of modern wireless devices are powered by battery and are able to switch between active and sleep states. Often, sending and listening occurring during active state dominate the energy expenditure, while sleeping costs much less~\cite{polastre05}. Therefore, when running an algorithm, each node's energy complexity (or, \emph{energy cost}) is often defined as the number of channel accesses~\cite{kardas13,chang17,chang18,chang20}; while time complexity is the number of slots till it halts.

The open and shared nature of wireless medium makes it vulnerable to jamming~\cite{gummadi07}.
%Such denial-of-service attacks could quickly deplete nodes' energy, putting an end to the normal operation of the network.
To thwart such behavior, one reasonable restriction is to bound the total amount of jamming, as injecting interfering signals also incurs operational cost. Specifically, we assume the existence of a jamming adversary called Eve. She can jam multiple channels in each slot, and jamming one channel for one slot costs one unit of energy. Eve has an energy budget $T$ that is \emph{unknown} to the nodes, and she can employ \emph{arbitrary} strategy to disrupt communication.

This setting motivates the notation of \emph{resource competitive algorithms}~\cite{bender15,king11,king18,gilbert14,chen19,augustine20disc} which focus on optimizing relative cost. Specifically, assume for each node the cost of sending or listening on one channel for one slot is one unit of energy (while idling is free),\footnote{In reality, the cost for sending, listening, and jamming might differ, but they are often in the same order. The assumptions here are mostly for the ease of presentation, and are consistent with existing work. Moreover, allowing different actions to have different constant costs will not affect the results.} can we design broadcast algorithms that ensure each node's cost is only $o(T)$? Such results would imply Eve cannot efficiently stop nodes from accomplishing the distributed computing task in concern. Interestingly enough, the answer is positive. In particular, Gilbert et al.~\cite{gilbert14} present a resource competitive broadcast algorithm in the single-channel radio network setting: with high probability, each node receives the message and terminates within $\tilde{O}(T+n)$ slots, while spending only $\tilde{O}(\sqrt{T/n}+1)$ energy.\footnote{We say an event happens \emph{with high probability (w.h.p.)} if the event occurs with probability at least $1-1/n^c$, for some tunable constant $c\geq 1$. Moreover, we use $\tilde{O}$ to hide poly-log factors in $n$, $C$, and $T$.} This algorithm works even when Eve is adaptive and $n$ is unknown to the nodes. Later, Chen and Zheng~\cite{chen19} consider the multi-channel setting: they show that when Eve is oblivious and $C=O(n)$, having multiple channels allows a linear speedup in time complexity, while the energy cost of each node remains to be $\tilde{O}(\sqrt{T/n}+1)$.

In this paper, we develop two new multi-channel broadcast algorithms that can tolerate a stronger adaptive adversary and work for arbitrary $n,C$ values, without sacrificing time efficiency or resource competitiveness. The first algorithm---called \multicastadp---needs to know $n$; while the other more complicated one---called \multicastadvadp---does not. Both algorithms are randomized, and in the interesting case where $T$ is large compared with $n$ and $C$, each node's runtime is $O(T/C)$, while each node's energy cost is $\tilde{O}(\sqrt{T/n})$.\footnote{The primary goal of resource competitive algorithms is to optimize nodes' cost for large $T$ values, see previous work (e.g., \cite{bender15,king18}) and discussion on resource competitiveness in Section \ref{subsec-model} for more details.}

\begin{theorem}\label{thm-multicastadv}
\multicastadp guarantees the following properties w.h.p.: (a) all nodes receive the message and terminate within $O(T/C+\tau_{time})=\tilde{O}(T/C+\max\{n/C,C/n\})$ slots; and (b) the cost of each node is $O(\sqrt{T/n}\cdot\sqrt{\lg{T}}\cdot\lg{n}+\tau_{cost})=\tilde{O}(\sqrt{T/n}+C/n)$.
\begin{itemize}
	\item When $C=O(n)$, $\tau_{time}=(n/C)\cdot\lg{(n/C)}\cdot\lg^2{n}$, and $\tau_{cost}=\lg{(n/C)}\cdot\lg{n}$.
	\item When $C=\Omega(n)$, $\tau_{time}=(C/n)\cdot\lg{(C/n)}\cdot\lg^2{n}$, and $\tau_{cost}=(C/n)\cdot\lg{(C/n)}\cdot\lg{n}$.
\end{itemize}
\end{theorem}

\begin{theorem}\label{thm-multicastadvadp}
\multicastadvadp guarantees the following properties w.h.p.: (a) all nodes receive the message and terminate within $O(T/C+(nC+C^2)\cdot\lg^4(nC))=\tilde{O}(T/C+nC+C^2)$ slots; and (b) the cost of each node is $O(\sqrt{T/n}\cdot\lg^2{T}+C^2\cdot\lg^5(nCT)+(nC+C^2)\cdot\lg^4(nC))=\tilde{O}(\sqrt{T/n}+nC+C^2)$.
\end{theorem}

We also complement algorithmic results with lower bounds. Specifically, the $O(T/C)$ term in runtime is optimal, as Eve can jam all $C$ channels continuously for $T/C$ slots. Meanwhile, the $\tilde{O}(\sqrt{T/n})$ term in energy cost matches lower bound up to a poly-logarithmic factor. Thus our algorithms achieve (near) optimal time and energy complexity simultaneously.
%An overview on the techniques we used to obtain the results is provided in Section \ref{subsec-technique-overview}.

\begin{theorem}\label{thm-cost-lower-bound}
For an adaptive adversary with budget $T$, any fair multi-channel broadcast algorithm that succeeds with constant probability imposes an expected cost of $\Omega(\sqrt{T/n})$ per node. Notice, an algorithm is \emph{fair} if all participating nodes have the same expected cost; both \multicastadp and \multicastadvadp are fair.
\end{theorem}

%In the reminder of this section, we will give more details on related work and model of computation/communication. We will also give a brief overview on the techniques that are used in obtaining the aforementioned results.

\subsection{Related Work}

Broadcasting in radio networks is non-trivial due to collisions. Classical results often rely on variants of the \textsc{Decay} procedure~\cite{bar-yehuda92}, while recent ones (e.g., \cite{ghaffari15,czumaj17}) tend to employ more advanced techniques (e.g., network decomposition) to improve performance. Besides time complexity, energy cost has also been taken into consideration when building communication primitives (e.g., \cite{gasieniec07,chang17,chang18,chang20}), but usually without assuming the existence of a jamming adversary.

Distributed computing in jamming-prone environment has attracted a lot of attention as well. Researchers from the theory community usually pose certain restrictions on the behavior of the malicious user(s), and then develop corresponding countermeasures (e.g., \cite{awerbuch08,richa10,meier09,dolev07}). Unfortunately, these restrictions somewhat limit the adversary's strategy, and many of the proposed algorithms also require honest nodes to spend a lot of energy. In view of these, \emph{resource competitive analysis}~\cite{bender15} is proposed. This framework allows more flexibility for the adversary, hence potentially better captures reality. However, it also brings new challenges to the design and analysis of algorithms.

In 2011, King, Saia, and Young~\cite{king11} developed the first resource competitive algorithm, in the context of 1-to-1 communication. (That is, Alice wants to send a message to Bob.)
%Specifically, the proposed Las Vegas algorithm ensures the expected cost of Alice and Bob is only $O(T^{0.62}+1)$.
As mentioned earlier, Gilbert et al.~\cite{gilbert14} later devise a single-channel broadcast algorithm that is resource competitive against jamming. They have also proved several lower bounds showing the algorithm's energy cost is near optimal. The work that is most closely related to ours is by Chen and Zheng~\cite{chen19}, in which several multi-channel broadcast algorithms are developed. However, an important drawback of \cite{chen19} is that it only considers an oblivious adversary, while all other previous results can tolerate an adaptive (or even reactive) adversary. In this paper, we close the gap by considering an adaptive adversary, and provide similar or better results than \cite{chen19} that work for arbitrary values of $n$ and $C$. We also prove our results are (near) optimal by deriving new lower bounds.

%Lastly, we note that resource competitive analysis has also been applied in other settings, interested readers can refer to \cite{bender15,king18} for more details.

\subsection{Additional Model Details}\label{subsec-model}

All nodes in the network start execution simultaneously and can independently generate random bits. In each slot, each node either sends a message on a channel, or listens on a channel, or remains idle. Only listening nodes get feedback regarding channel status. The adversary Eve is adaptive and potentially randomized: at the beginning of each slot, she is given all past execution history and can use these information to determine her behavior. However, she does \emph{not} know honest nodes' random bits or behavior of the current slot.

In each slot, for each listening node, the channel feedback is determined by the number of sending nodes on that channel and the behavior of Eve. Specifically, consider a slot and a channel $ch$. If no node sends on $ch$ and Eve does not jam $ch$, then nodes listening on $ch$ hear silence. If exactly one node sends a message on $ch$ and Eve does not jam $ch$, then nodes listening on $ch$ receive the unique message. Finally, if at least two nodes send on $ch$ or Eve jams $ch$, then nodes listening on $ch$ hear noise. Note that we assume nodes cannot tell whether noise is due to jamming or message collision (or both).

We adopt the following definition of resource competitive algorithms introduced in \cite{bender15}:

\begin{definition}\label{def-comp-alg}
Consider an execution $\pi$ in which nodes execute algorithm $\mathcal{A}_N$ and Eve employs strategy $\mathcal{A}_E$. Let $\texttt{cost}_u(\pi)$ denote the energy cost of node $u$, and $T(\pi)$ denote the energy cost of Eve. We say $\mathcal{A}_N$ is \emph{$(\rho,\tau)$-resource competitive} if $\max_{u}\{\texttt{cost}_u(\pi)\}\leq\rho(T(\pi))+\tau$ for any execution $\pi$.
\end{definition}

In above, $\rho$ is a function of $T$ and possibly other parameters (such as $n,C$). It captures the additional cost nodes incur due to jamming. The other function $\tau$ captures the cost of the algorithm when Eve is absent, thus $\tau$ should not depend on $T$. Most resource competitive algorithms aim to minimize $\rho$, while keeping $\tau$ reasonably small.

\subsection{Overview of Techniques}\label{subsec-technique-overview}

\subparagraph*{Fast and competitive broadcast against jamming.} Most resource competitive broadcast algorithms group slots into consecutive \emph{epochs}, and execute a jamming-resistant broadcast scheme within each epoch. In the single-channel setting, often the core idea is to broadcast ``sparsely''~\cite{king11,king18}. Consider 1-to-1 communication as an example. If both nodes send and listen in $\Theta(\sqrt{R})$ random slots in an epoch of length $R$, then by a birthday-paradox argument, successful transmission will occur with constant probability even if Eve jams constant fraction of all $R$ slots. In the multi-channel setting, ``\emph{epidemic broadcast}'' is employed~\cite{chen19}. In the simplest form of this scheme, in each time slot, each node will choose a random channel from $[C]=\{1,2,\cdots,C\}$. Then, each informed node (i.e., the node knows the message $m$) will broadcast $m$ with a constant probability, while each uninformed node will listen with a constant probability. If $C=n/2$, broadcast will complete in $O(\lg{n})$ slots w.h.p., and this claim holds even if Eve jams constant fraction of all channels for constant fraction of all slots.

In designing \multicastadp and \multicastadvadp, one key challenge is to extend the basic epidemic broadcast scheme to guarantee an optimal $O(T/C)$ runtime for arbitrary $n,C$ values, without increasing energy expenditure. To that end, we note that in the single-channel setting, \cite{gilbert14} has shown $\Theta(1/\sqrt{Rn})$ is roughly an optimal working probability (i.e., sending/listening probabilities). When $C$ channels are available, a good way to adjust the probability would be to multiply it by a factor of $\sqrt{C}$ (i.e., $\Theta(\sqrt{C/(Rn)})$). Intuitively, the reason being: if each node works on $\sqrt{C}$ random channels simultaneously in each slot, then again by a birthday-paradox argument, each pair of nodes will meet on at least one channel with at least constant probability, which effectively means the optimal single-channel analysis could be applied again. Of course nodes do not have multiple transceivers and cannot work on multiple channels simultaneously, but over a period of time, multiplying the single-channel working probability by $\sqrt{C}$ achieves similar effect. On the other hand, although the working probability of nodes is increased by a factor of $\sqrt{C}$, the energy expenditure of Eve will increase by a factor of $\Theta(C)$. As a result, compared with single-channel solutions, our algorithms has a $\Theta(C)$ speedup in time, yet the resource competitive ratio is unchanged.

\vspace{-2ex}\subparagraph*{Termination and the coupling technique.} Termination mechanism is another key integrant, it ensures nodes stop execution correctly and timely. For each node $u$, a helpful termination criterion is comparing $N_u$---the number of silent slots it observed during the current epoch---to some pre-defined threshold.
%This is because, a large $N_u$ often indicates jamming from Eve is weak, implying the broadcast scheme is likely to have succeeded; while a small $N_u$ suggests jamming from Eve is strong, implying the adversary has paid a lot of price for stopping message dissemination.
To argue the correctness of our algorithms, we often need to show $N_u$ is close to its expected value. However, this is non-trivial if Eve is adaptive.

To see this, consider an epoch containing $R$ slots. Define $G_i$ as the \emph{behavior} (i.e., channels choices and actions) of all nodes in slot $i$, and define $Q_i$---the set of channels that are \emph{not} jammed by Eve---as the \emph{jamming result} of slot $i$. Note that $N_u$ can be written as the sum of $R$ indicator random variables: $N_u=\sum_{i=1}^{R}N_{u,i}$, where $N_{u,i}=1$ iff $u$ hears silence in the $i$\textsuperscript{th} slot. $N_{u,i}$ is determined by $G_i$ and $Q_i$, but in general $Q_i$ can be arbitrary function of $\{G_1,G_2,\cdots,G_{i-1},Q_1,Q_2,\cdots,Q_{i-1}\}$. Nonetheless, in case Eve is oblivious (i.e., an offline adversary), her optimal strategy would be a fixed vector of jamming results $\langle q_1,q_2,\cdots,q_R\rangle$, thus $\{N_{u,1},N_{u,2},\cdots,N_{u,R}\}$ are mutually independent when $\{G_1,G_2,\cdots,G_R\}$ are mutually independent (this can be easily enforced by the algorithm). Therefore, if Eve is oblivious, we can directly apply powerful concentration inequalities like Chernoff bounds to show $N_u$ is close to its expectation. However, once Eves becomes adaptive, $Q_i$ could depend on $\{G_1,\cdots,G_{i-1}\}$ and above observations no longer hold: $\{N_{u,1},\cdots,N_{u,R}\}$ could be dependent!

In this paper, we leverage the \emph{coupling} technique (see, e.g., \cite{dubhashi09}) extensively to resolve the dependency issue. Specifically, for each vector of jamming results over one epoch, we create a coupled execution and relate $N_u$ to a corresponding random variable in the coupled execution. By carefully crafting the coupling, the random variable in the coupled execution can be interpreted as the sum of a set of independent random variables, allowing us to bound the probability that $N_u$ deviates a lot from its expectation. However, there is a catch in this approach: bounding the probability that $N_u$ deviates a lot from its expectation requires us to sum the failure probability over all jamming results vectors, but there are $O(2^{CR})$ such vectors! Our solution to this new problem is to group all vectors into fewer categories, such that vectors within one category have same or similar effects on the metric we concern.

\textbf{Remark.} Techniques like ``principle of deferred decision'', or the ones used in previous work, cannot resolve the dependency issue directly in our setting. See Appendix \ref{sec-app-adaptive-discussion}.

\vspace{-2ex}\subparagraph*{Lower bound.} Existing result~\cite{gilbert14} indicates broadcast in the single-channel settings requires each node spending $\Omega(\sqrt{T/n})$ energy, but could it be the case that having multiple channels also reduces the energy complexity of the problem? We show the answer is negative.

Specifically, for any multi-channel broadcast algorithm $\mathcal{A}_n$, we devise a corresponding multi-channel 1-to-1 communication algorithm $\mathcal{A}_2$ that simulates $\mathcal{A}_n$ internally. We also devise a jamming strategy $\mathcal{S}$ for disrupting $\mathcal{A}_n$ and $\mathcal{A}_2$: in each slot, for each channel, Eve jams that channel iff a successful transmission will occur on that channel with a probability more than $1/T$. $\mathcal{A}_2$ and $\mathcal{S}$ are carefully constructed so that algorithms' success probabilities and participating nodes' energy expenditure in the two executions (i.e., $(\mathcal{A}_n,\mathcal{S})$ and $(\mathcal{A}_2,\mathcal{S})$) are closely connected. Then, we derive an energy complexity lower bound for multi-channel 1-to-1 communication assuming Eve uses $\mathcal{S}$. (This result, Theorem \ref{thm-cost-lower-bound-1-to-1} in Section \ref{sec-lower-bound}, could be of independent interest and is strong in two aspects: (a) the bound holds even if the two nodes has multiple transceivers; (b) the proof uses an explicit jamming strategy $\mathcal{S}$.) At this point, an energy complexity lower bound for $\mathcal{A}_n$ can be obtained via reduction.

\section{Notations}\label{sec-notation}

Let $V$ be the set of all nodes. Since all algorithms developed in this paper proceed in epochs, consider a slot $i$ in an epoch of length $R$, where $1\leq i\leq R$. Denote $Q_i\in 2^{[C]}$ as the jamming result of the $i$\textsuperscript{th} slot: $Q_i$ is the set of channels that are \emph{not} jammed by Eve in the $i$\textsuperscript{th} slot. Denote $G_i=\langle(G_{i,v}^{ch})_{v\in V},(G_{i,v}^{act})_{v\in V}\rangle$ as the behavior (i.e., channel choices and actions) of the $n$ nodes in the $i$\textsuperscript{th} slot: $G_i\in\Omega=[C]^n\times\{\text{send},\text{listen},\text{idle}\}^n$.\footnote{There is a technical subtlety worth clarifying. The ``behavior'' here does not care about the exact content to be broadcast if some node(s) choose to send message(s) in a slot. That is, for each slot, the ``behavior'' here is \emph{not} some element in $[C]^n\times(M\cup\{\text{listen},\text{idle}\})^n$, where $M$ is the set of all possible messages. This is for the ease of presentation and will not affect the correctness of our results.\label{footnote-behavior}} We use $F_i$ to denote the randomness Eve used in the $i$\textsuperscript{th} slot. Since Eve is adaptive, it is easy to see $Q_i$ is determined by $\bm{F}_{\leq i}=(F_1,\cdots,F_i)$ and $\bm{G}_{<i}=(G_1,\cdots,G_{i-1})$.
%\footnote{Strictly speaking, $Q_i$ could also depend on the execution history of previous epochs. Again, this is a technical subtlety that will not affect our results, and we ignore it to simplify presentation.}
Lastly, define $\bm{Q}_{\leq i}=(Q_1,\cdots,Q_i)$.

To quantify the severity of jamming from Eve, for a given slot, we use $\mathcal{E}(>x)$ (respectively, $\mathcal{E}(\geq x)$, $\mathcal{E}(<x)$, $\mathcal{E}(\leq x)$) to denote that in a slot, more than (respectively, at least, less than, at most) $x$ fraction of the $C$ channels are \emph{not} jammed by Eve. In the following, we use $\mathcal{E}(\cdot x)$ to represent one of the above four forms. (I.e., ``$\cdot$'' denotes ``$>$'', ``$\geq$'', ``$<$'', or ``$\leq$''.)

For an epoch, we use $\mathcal{E}^{(>y)}(\cdot x)$ (respectively, $\mathcal{E}^{(\geq y)}(\cdot x)$, $\mathcal{E}^{(<y)}(\cdot x)$, $\mathcal{E}^{(\leq y)}(\cdot x)$) to denote the event that for more than (respectively, at least, less than, at most) $y$ fraction of the $R$ slots, $\mathcal{E}(\cdot x)$ happen. For example, $\mathcal{E}^{(>0.1)}(>0.2)$ means in an epoch, for more than $0.1$ fraction of all slots, Eve leaves more than $0.2$ fraction of all channels unjammed.

Define negation operation in the following manner: $(\overline{>x})=(\leq x)$ and vice versa; $(\overline{<x})=(\geq x)$ and vice versa. Further define complement operation in the following manner: $\complement(>x)=(<1-x)$ and vice versa; $\complement(\geq x)=(\leq 1-x)$ and vice versa. It is easy to verify $\mathcal{E}^{(\cdot y)}(\cdot x)=\mathcal{E}^{(\complement(\cdot y))}(\overline{\cdot x})$ and $\overline{\mathcal{E}^{(\cdot y)}(\cdot x)}=\mathcal{E}^{(\overline{\cdot y})}(\cdot x)$. Therefore:

\vspace{-1ex}$$\overline{\mathcal{E}^{(\geq y)}(\geq x)} = \overline{\mathcal{E}^{(\leq 1-y)}(\overline{\geq x})} = \mathcal{E}^{(>1-y)}(\overline{\geq x}) = \mathcal{E}^{(>1-y)}(<x)$$

Again, as a simple example, the above equality implies ``if in an epoch, it is not the case that in at least $0.1$ fraction of all slots Eve leaves at least $0.2$ fraction of all channels unjammed, then it must be the case that in more than $0.9$ fraction of all slots, Eve leaves less than $0.2$ fraction of all channels unjammed; and vice versa''. (I.e., $\overline{\mathcal{E}^{(\geq 0.1)}(\geq 0.2)}= \mathcal{E}^{(>0.9)}(<0.2)$.)

\section{The \multicastadp Algorithm}\label{sec-multicastadp}

Each node $u$ maintains a Boolean variable $M_u$ to indicate whether it knows the message $m$ (in which case $M_u=true$ and $u$ is informed) or not (in which case $M_u$ is $false$ and $u$ is uninformed). Initially, only the source node sets $M_u=true$. The algorithm proceeds in epochs and the $i$\textsuperscript{th} epoch contains $R_i=a\cdot 4^i\cdot i\cdot\lg^2{n}$ slots, where $a$ is some large constant. In each slot in epoch $i$, for each active node $u$, it will hop to a uniformly chosen random channel. Then, $u$ will choose to broadcast or listen each with probability $p_i=(\sqrt{C/n})/2^i$. If $u$ decides to broadcast and $M_u=true$, it sends $m$; otherwise, $u$ sends a special beacon message $\pm$. On the other hand, if in a slot $u$ decides to listen, it will record the channel feedback. Finally, by the end of an epoch $i$, for a node $u$, if among the slots it listened within this epoch, at least $(p_iR_i)/2$ are silent slots, then $u$ will halt. One point worth noting is, the first epoch number is not necessarily one; instead, it is chosen as a sufficiently large integer to ensure $p_i\leq 1/2$ and $p_i\leq C/(4n)$. Hence, the first epoch number is $I_b=2+\lceil\max\{\lg{(\sqrt{n/C})},\lg{(\sqrt{C/n})}\}\rceil$. See Figure \ref{fig-alg-multicastadp} in Appendix \ref{sec-app-code} for the pseudocode of \multicastadp.

\section{Analysis of \multicastadp}\label{sec-multicastadp-analysis}

\subparagraph*{Effectiveness of epidemic broadcast.} The first technical lemma states if in an epoch jamming from Eve is not strong and every node is active, then all nodes will be informed by the end of the epoch.%It highlights the effectiveness of the epidemic broadcast scheme: when less than $n/2$ nodes know message $m$, the number of informed nodes will increase by some constant factor every so often; once at least $n/2$ nodes know $m$, remaining uninformed nodes will quickly learn the message too. Appendix \ref{sec-app-multicastadp-proof} includes omitted proofs of this section.

\begin{lemma}\label{lemma-multicastadp-fast-bcst}
Assume all nodes are active during epoch $i$. By the end of epoch, with probability at most $n^{-\Theta(i)}$, the following two events happen simultaneously: (a) $\mathcal{E}_{\geq}$ occurs in the epoch; and (b) some node is still uninformed. Here, $\mathcal{E}_{\geq}$ is event $\mathcal{E}^{\geq 0.1}(\geq 0.1)$, defined in Section \ref{sec-notation}.
\end{lemma}

Proof of this lemma highlights how coupling can help us handle the adaptivity of Eve. To construct the coupling, we first specify how nodes' behavior is generated. Fix an epoch and let constant $x_1=0.1$, and $y_1=0.1$. Imagine two sufficiently long bit strings $\bm{T}_{high}$ and $\bm{T}_{low}$, in which each bit is generated independently and uniformly at random. Divide $\bm{T}_{high}$ and $\bm{T}_{low}$ into consecutive \emph{chunks} of equal size, such that each chunk provides enough random bits for $n$ nodes to determine their behavior in a slot. More formally, $\bm{T}_{high}=(T_{hi}^{(1)},T_{hi}^{(2)},\cdots,T_{hi}^{(R)})$ and $\bm{T}_{low}=(T_{lo}^{(1)},T_{lo}^{(2)},\cdots,T_{lo}^{(R)})$, where each $T_{hi}^{(*)}$ or $T_{lo}^{(*)}$ is a chunk. Next, we introduce three processes that are used during the coupling: $\beta$, $\beta'$, and $\gamma$.

We begin with $\beta$, which is an execution of \multicastadp with adversary Eve. The tricky part about $\beta$ is: in the $i$\textsuperscript{th} slot, nodes' behavior $G_i$ is \emph{not} determined by $T_{hi}^{(i)}$ or $T_{lo}^{(i)}$ directly. Instead, the chunk is chosen in a more complicated way. Specifically, at the beginning of slot $i$, Eve first computes its jamming result $Q_i$ (i.e. the set of unjammed channels) based on $\bm{F}_{\leq i}$ and $\bm{G}_{<i}$. If $|Q_i|\geq x_1C$ and the number of previously used chunks from $\bm{T}_{high}$ is less than $y_1R$, then we pick the next unused chunk from $\bm{T}_{high}$; otherwise, we pick the next unused chunk from $\bm{T}_{low}$. Assume $T^{(j)}$ is the chosen chunk, and it computes to nodes' behavior $\langle(\hat{G}^{ch}_{v})_{v\in V},(\hat{G}^{act}_{v})_{v\in V}\rangle$. Still, we do not use $\langle(\hat{G}^{ch}_{v})_{v\in V},(\hat{G}^{act}_{v})_{v\in V}\rangle$ as nodes' behavior. Instead, we permute the channel choices according to the jamming result. Specifically, for each $q\in 2^{[C]}$, define permutation $\pi_q$ on $[C]$ as follows: for $1\leq k\leq |q|$, $\pi_q(k)$ is the $k$\textsuperscript{th} smallest element in $q$; and for $|q|+1\leq k\leq C$, $\pi_q(k)$ is the $(k-|q|)$\textsuperscript{th} smallest element in $[C]\backslash q$. (For example, if $C=5$ and $q=\{2,4\}$, then $\pi_q$ permutes $\langle 1,2,3,4,5\rangle$ to $\langle 2,4,1,3,5\rangle$.) Further define bijection $\Psi_q:\Omega \to \Omega$ using $\pi_{q}$:\label{def-psi_q-pi_q}

\vspace{-1ex}$$\Psi_q\left(\left\langle\left(\hat{G}^{ch}_{v}\right)_{v\in V},\left(\hat{G}^{act}_{v}\right)_{v\in V}\right\rangle\right)=\left\langle\left(\pi_q\left(\hat{G}^{ch}_{v}\right)\right)_{v\in V},\left(\hat{G}^{act}_{v}\right)_{v\in V}\right\rangle$$

\noindent Now, we use $\langle(\pi_q(\hat{G}^{ch}_{v}))_{v\in V},(\hat{G}^{act}_{v})_{v\in V}\rangle$ as nodes' behavior $G_i$ in slot $i$. Formally, let $K(\bm{Q}_{\leq i})=\sum_{j=1}^{i}{\mathbb{I}[|Q_j|\geq x_1C]}$ count weakly jammed slots (i.e., $|Q_j|\geq x_1C$) among the first $i$ slots, where $\mathbb{I}[|Q_j|\geq x_1C]$ is an indicator random variable. Then, $G_i$ can be defined as:

\vspace{-1ex}$$G_i=
\left\{\begin{array}{ll}
	\Psi_{Q_i}\left(T_{hi}^{\left(K\left(\bm{Q}_{\leq i}\right)\right)}\right), & |Q_i|\geq x_1C  \textrm{ and } K(\bm{Q}_{\leq i})\leq y_1R \\
	\Psi_{Q_i}\left(T_{lo}^{\left(i-K\left(\bm{Q}_{\leq i}\right)\right)}\right), &|Q_i|< x_1C  \textrm{ and } K(\bm{Q}_{\leq i})\leq y_1R \\
	\Psi_{Q_i}\left(T_{lo}^{\left(i-y_1R\right)}\right), &\textrm{ otherwise} \\
\end{array}
\right.$$

Careful readers might suspect does $\bm{G}=(G_1,G_2,\cdots,G_R)$ in process $\beta$ really has the correct distribution $\bm{\mathcal{G}}$ we want. (That is, $\bm{\mathcal{G}}$ is the distribution in which the behavior of the nodes are determined by, say $\bm{T}_{low}$, directly.)
%More specifically, in each slot, for each node, is it true that the node chooses a channel from $[C]$ uniformly at random, and then decides its action (i.e., send, listen, or idle) based on the working probabilities specified by \multicastadp? Furthermore, is it true that nodes' behavior in different slots is independent?
After all, just by looking at the definition, it seems $G_i$ depends on $Q_i$, which is controlled by Eve. Interestingly enough, indeed $\bm{G}\sim \bm{\mathcal{G}}$. To understand this intuitively, consider the following simple game played between Alice and Eve. In each round, Alice tosses a fair coin but does not reveal it to Eve (this coin plays similar role as $T^{(j)}$). However, Eve can decide whether to flip the coin or not (this is like permuting channel assignments according to $q$). Finally, the coin is revealed and the game continues into the next round. Now, a simple but important observation is: the coin is still a fair coin in each round, although Eve can decide whether to flip it or not.
Similarly, back to our setting, we can show $\bm{G}\sim \bm{\mathcal{G}}$.
Detailed proof of this claim is provided in Appendix \ref{sec-app-multicastadp-proof}.

%Similarly, back to our setting, since $\Psi_q$ is a bijection, and since $Q_i$ depends on $\bm{F}_{\leq i}$ and $\bm{G}_{<i}$, we can show $\bm{G}\sim \bm{\mathcal{G}}$.

%Due to space constraint, proof of this claim is provided in Appendix \ref{sec-app-multicastadp-proof}.

%Now that we have shown in $\beta$ our more complicated way of generating random choices faithfully simulates distribution $\bm{\mathcal{G}}$, we continue to introduce process $\beta'$.

We continue to introduce process $\beta'$. In $\beta'$, still there are $n$ nodes executing \multicastadp, along with a jamming adversary Carlo. However, for each slot $i$, if in $\beta$ nodes use $\Psi_{Q_i}(T_{hi}^{(j)})$ (resp., $\Psi_{Q_i}(T_{lo}^{(j)})$) to determine their behavior, then in $\beta'$ nodes directly use $T_{hi}^{(j)}$ (resp., $T_{lo}^{(j)}$), and Carlo leaves channels $\{1,2,\cdots,x_1C\}$ unjammed (resp., jams all channels).

Finally, in $\gamma$, again there are $n$ nodes executing \multicastadp, yet adversary %\sout{jams obliviously using} 
adopts a fixed strategy $\bm{\hat{q}}$: in the first $y_1R$ slots, channels $\{1,2,\cdots,x_1C\}$ are unjammed; and in the remaining $(1-y_1)R$ slots, all channels are jammed. Besides, in the $i^{th}$ slot, node directly use chunk $T^{(i)}_{hi}$ if $i\leq y_1R$, and otherwise $T^{(i-y_1R)}_{lo}$ to compute their behavior.

Now we can define and couple some random variables in these process.
%fully determined by $\bm{T}_{high}$, $\bm{T}_{low}$, and adversary's strategy and randomness. 
We sketch the proof of Lemma~\ref{lemma-multicastadp-fast-bcst} here. %and detailed proof is provided in Appendix \ref{sec-app-multicastadp-proof}.
Appendix \ref{sec-app-multicastadp-proof} includes omitted proofs of this section.
\begin{proof}[Proof sketch of Lemma~\ref{lemma-multicastadp-fast-bcst}] Define $\mathcal{E}_{X}$ (resp., $\mathcal{E}_{X'}$ or $\mathcal{E}_{Y}$) be the event that some node is still uninformed by the end of process $\beta$ (resp., $\beta'$ or $\gamma$). Recall $\mathcal{E}_{\geq}$ is related to $\bm{Q}$ in process $\beta$. The following claims capture the relationship between them:
  
\emph{Claim I: $\mathcal{E}_{X}$ implies $\mathcal{E}_{X'}$}. Notice that permuting nodes' behaviors and unjammed channels together will not change nodes' feedback. Moreover, $(\pi(Q'_i)\subseteq Q_i)\wedge (M'_{i-1}\subseteq M_{i-1})$ implies $M'_i\subseteq M_i$, where $Q_i$ (resp., $Q'_i$) denotes the set of unjammed channels in $i^{th}$ slot, and $M_i$ (resp., $M'_i$) denotes the set of informed nodes by the end of slot $i$, in process $\beta$ (resp., $\beta'$).

\emph{Claim II: $(\mathcal{E}_{X'}\wedge \mathcal{E}_{\geq})$ implies $(\mathcal{E}_{Y}\wedge \mathcal{E}_{\geq})$}. If $\mathcal{E}_{\geq}$ holds, in both $\beta'$ and $\gamma$, Eve leaves exactly $y_1R$ slots in which channel $\{1,2,\cdots,x_1C\}$ are unjammed, and jams all channels in $(1-y_1)R$ slots. Since we can simply ignore the slots in which all channels are jammed, then in each left slot, behavior and further feedback of nodes in the two processes are same correspondingly.

Therefore, $\Pr[\mathcal{E}_{X} \wedge \mathcal{E}_{\geq}]\leq \Pr[\mathcal{E}_{Y} \wedge \mathcal{E}_{\geq}]\leq \Pr[\mathcal{E}_{Y}]\leq\exp(-\Theta(i\cdot\lg{n}))$, the last inequality is due to $|M|$ will increase by some constant factor every so often when $|M|\leq n/2$; once $|M|>n/2$, remaining uninformed nodes will quickly learn the message too.
\end{proof}

\subparagraph*{Competitiveness and Correctness.} We prove two other key lemmas in this part. The first one shows Eve cannot stop nodes from halting without spending a lot of energy, thus guaranteeing the resource competitiveness of the termination mechanism.

\begin{lemma}\label{lemma-multicastadp-fast-halt}
Fix an epoch $i$ and a node $u$, assume $u$ is alive at the beginning of this epoch. By the end of this epoch, with probability at most $\exp(-\Theta(i\cdot\lg {n}))$, the following two events happen simultaneously: (a) $\mathcal{E}_{\geq}$ occurs in the epoch; and (b) node $u$ does not halt. Here, $\mathcal{E}_{\geq}$ is event $\mathcal{E}^{\geq 0.99}(\geq 0.99)$, which is defined in Section \ref{sec-notation}. %Here, $x_2=y_2=0.99$, and $\mathcal{E}^{\geq y_2}(\geq x_2)$ is defined in Section \ref{sec-notation}.
\end{lemma}
\begin{proof}
Arrange the randomness of nodes as what we do in the proof of Lemma \ref{lemma-multicastadp-fast-bcst}, except that we use parameter $x_2=0.99$ to replace $x_1$, and $y_2=0.99$ to replace $y_1$. Let  $R$ be the length of the epoch, and $p$ be nodes' working probability. Now, fix an arbitrary node $u$, define $X_i$ (resp., $X'_i$ or $Y_i$) to be an indicator random variable taking value one iff $u$ hears silence in the $i$\textsuperscript{th} slot in $\beta$ (resp., $\beta'$ or $\gamma$). Following random variables are what we intend to couple: $X=\sum_{i=1}^{R}X_i$, $X'=\sum_{i=1}^{R}X'_i$, and $Y=\sum_{i=1}^{R}{Y_i}$. Specifically:

\emph{Claim I: For any integer $t\geq 0: \Pr[X\leq t]\leq\Pr[X'\leq t]$.} This is because in each slot $i$, when each node chooses its operating channel uniformly at random, both $X_i$ and $X'_i$ are proportional to $|Q_i|$. Moreover, since $\pi(Q'_i)\subseteq Q_i$, with a stochastic dominance argument, then $X'$ is a ``underestimate'' of $X$. Detailed proof of this claim is provided in Appendix \ref{sec-app-multicastadp-proof}.

\emph{Claim II: For any integer $t\geq 0: \Pr[(X'\leq t)\wedge \mathcal{E}_{\geq}]\leq \Pr[(Y\leq t)\wedge \mathcal{E}_{\geq}]$.} In both $\beta'$ and $\gamma$, Eve leaves channels $\{1,2,\cdots,x_2C\}$ unjammed for $y_2R$ slots. Notice that in each of the $R$ slots, nodes' behavior are sampled from a same the distribution, so the index set of the $y_2R$ slots does not matter.

Therefore, $\Pr[(Y\leq t)\wedge \mathcal{E}_{\geq}]\leq \Pr[Y\leq t]\leq \exp(-\Theta(i\cdot\lg^2{n}))$. Since $\{Y_1,Y_2,\cdots,Y_R\}$ is a set of independent random variables, bounding $\Pr[Y<Rp/2]$ is easy. Specifically, $\mathbb{E}[Y]=y_2\cdot x_2\cdot p\cdot (1-p/C)^{n-1}\geq 0.99^2\cdot Rp\cdot(1-p/C)^{n}\geq 0.99^2\cdot Rp\cdot e^{-2np/C}\geq 0.99^2\cdot Rp\cdot e^{-0.5}>0.59Rp$. Apply a Chernoff bound, we know $\Pr[Y<Rp/2]\leq\exp(-\Theta(Rp))\leq\exp(-\Theta(i\cdot\lg {n}))$. 
\end{proof}

The second lemma states that all nodes must have been informed before any node decides to halt, thus message dissemination must have completed before any node stops execution.
\begin{lemma}\label{lemma-multicastadp-halt-imply-informed}
Fix an epoch $i$ in which all nodes are active, fix an alive node $u$. By the end of this epoch, with probability at most $\exp(-\Theta(i\cdot\lg{n}))$, the following two events happen simultaneously: (a) node $u$ halts; and (b) some node is still uninformed.
\end{lemma}

\begin{proof}[Proof sketch of Lemma \ref{lemma-multicastadp-halt-imply-informed}]
Consider two complement cases: either Eve jams a lot in the epoch, or she does not. If jamming is not strong, Lemma \ref{lemma-multicastadp-fast-bcst} implies no node remains uninformed. Otherwise, then $u$ should not hear a lot of silent slots and will not halt.

%Let $x_1=y_1=0.1$, $R$ be the length of the epoch, and $p$ be nodes' working probability. 
Let $\mathcal{E}_1$ be event (a) in lemma statement; 
%(i.e., ``node $u$ hears silence at least $Rp/2$ times'');
and $\mathcal{E}_2$ be event (b) in lemma statement. 
%(i.e., ``some node never hears the message'').
Let $\mathcal{E}_3$ be the event that $\mathcal{E}^{\geq 0.1}(\geq 0.1)$ happens in the epoch. By Lemma \ref{lemma-multicastadp-fast-bcst}, $\Pr[\mathcal{E}_2 \wedge \mathcal{E}_3 ]\leq\exp(-\Theta(i\cdot\lg{n}))$. 
Bounding $\Pr[\mathcal{E}_1\wedge \overline{\mathcal{E}_3}]\leq \exp(-\Theta(i\cdot\lg{n}))$ is similar to that of Lemma \ref{lemma-multicastadp-fast-halt}'s. %with detailed proof provided in Appendix \ref{sec-app-multicastadp-proof}.
As a result, $\Pr[\mathcal{E}_1\wedge\mathcal{E}_2]=\Pr[\mathcal{E}_1\wedge\mathcal{E}_2\wedge\mathcal{E}_3]+\Pr[\mathcal{E}_1\wedge\mathcal{E}_2\wedge\overline{\mathcal{E}_3}]\leq\Pr[\mathcal{E}_1\wedge \overline{\mathcal{E}_3}]+\Pr[\mathcal{E}_2 \wedge \mathcal{E}_3 ]\leq \exp(-\Theta(i\cdot\lg{n}))$. 
\end{proof}

\subparagraph*{Main theorem.} We sketch the proof of Theorem \ref{thm-multicastadv} here, and full version is in Appendix \ref{sec-app-multicastadp-proof}.

Fix a node $u$, we begin by computing how long $u$ remains active. Let $L$ be the total runtime of $u$. Since epoch length increases geometrically, we only need to focus on the last epoch $u$ in which is active. Also, notice that Lemma \ref{lemma-multicastadp-fast-halt} suggests Eve must jam a lot in an epoch---the amount of which can be described as some function of epoch length---to stop $u$ from halting. Putting these pieces together, we show $\Pr(L>\Theta(1)\cdot T/C)\leq n^{-\Omega(1)}$. By a union bound, we know when $T=\Omega(C)$ w.h.p.\ all nodes halt within $O(T/C)$ slots.

Next, we analyze the cost of nodes. Again fix a node $u$, let $F$ denote its total cost. By an argument similar to above, we are able to prove $\Pr(F>\Theta(\lg{n})\cdot\sqrt{\lg{T}\cdot(T/n)})\leq n^{-\Omega(1)}$. By a union bound, we know when $T=\Omega(C)$ w.h.p.\ the cost of each node is $O(\sqrt{T/n}\cdot\sqrt{\lg{T}}\cdot\lg{n})$.

The last step is to show with high probability each node must have been informed when it halts, and this can be proved via an application of Lemma \ref{lemma-multicastadp-halt-imply-informed}.

Finally, we note that when $T=o(C)$, all nodes will halt by the end of the first epoch, with high probability. This results in the $\tau_{time}$ and $\tau_{cost}$ terms in the theorem statement.

\section{The \multicastadvadp Algorithm}\label{sec-multicastadvadp}

Our second algorithm---called \multicastadvadp---works even if knowledge of $n$ is absent. However, its design and analysis are much more involved than that of \multicastadp.
%Networks organized in ad-hoc mode without infrastructure often cannot provide $n$ a priori. Our second algorithm---called \multicastadvadp---deals with such scenario. However, its design and analysis are much more involved than that of \multicastadp.

\vspace{-2ex}\subparagraph*{Building \multicastadvadp.} When the value of $n$ is unknown, the principal obstacle lies in properly setting nodes' working probabilities. In view of this, we let \multicastadvadp contain multiple \emph{super-epochs}, each of which contains multiple \emph{phases}, and nodes may use different working probabilities in different phases. Notice, for each super-epoch, we need to ensure it contains sufficiently many ``good'' phases, in the sense that within each such good phase broadcast will succeed if Eve does not heavily jam it. Another challenge posed by the unknown $n$ value is that the simple termination criterion---large fraction of silent slots---no longer works, as this can happen when the working probability is too low.

Gilbert et al.~\cite{gilbert14} provide a solution to the above two challenges in the single-channel setting. Specifically, at the beginning of a super-epoch $i$, nodes set their initial working probability to a pre-defined small value. After each phase, each node $u$ increases its working probability $p_u$ by a factor of $2^{\max\{0,\eta_u-0.5\}/i}$, where $\eta_u$ denotes the fraction of silent slots $u$ observed within the phase. This mechanism provides two important advantages: (a) Eve has to keep jamming heavily to prevent $p_u$ from reaching the ideal value; and (b) $p_u$ and $p_v$ might be different for two nodes $u$ and $v$, but the difference is bounded. As for termination, the number of messages nodes heard could be a good metric. However, a simple threshold would not work. Instead, Gilbert et al.\ develop a two-stage termination mechanism: when a node $u$ hears the message sufficiently many times, it becomes a \texttt{helper} and obtains an estimate of $n$; Later, when $u$ is sure that all nodes have become \texttt{helper}, it will stop execution.

In \multicastadvadp, we extend the above approach to the multi-channel setting. Specifically, we observe that the single-channel message dissemination scheme used in \cite{gilbert14} is relatively slow in that it needs $\Theta(\lg{n})$ phases to accomplish broadcast. By contrast, in \multicastadvadp, the application of epidemic broadcast reduces this time period to a single weakly-jammed phase. This replacement is not a simple cut-and-paste. Instead, we also adjust the phase structure accordingly. In particular, each phase now contains two \emph{steps}. This adjustment further demands us to change the way nodes' update their working probabilities after each phase: $p_u\gets p_u\cdot 2^{\max\{0,\eta^{step1}_u+\eta^{step2}_u-1.5\}}$. In the end, \multicastadvadp provides a slightly better resource competitive ratio than \cite{gilbert14}.

Handing adaptivity via coupling also becomes more challenging.
In more detail, in each phase we need the number of silent slots $u$ heard $N_u$ to be close to its expectation for \emph{any} jamming result (instead of, say, only when jamming is strong, as in the proof of Lemma \ref{lemma-multicastadp-halt-imply-informed}).
To acquire the desired results, we have to consider jamming results vectors at a finer level (rather than the unique extremal category in \multicastadp), which in turn requires the failure probability for each category to be much lower (otherwise a union bound over the increased number of categories would not work).
%\highlight{\sout{Potentially} 
Larger deviation of $N_u$ from expectation solves the issue, but it further demands the initial working probability nodes used at the beginning of each epoch to be sufficiently high.
%This further demands the initial working probability nodes used at the beginning of each super-epoch to be sufficiently high (otherwise the failure probability would be too large even with a Chernoff bound).
Unfortunately, this modification could result in nodes becoming \texttt{helper} with incorrect estimates of $n$, violating the correctness of the termination mechanism. We fix this problem by adding step three to each phase.
%Observing the fraction of silent slots in step three allows nodes to determine the reliability of their estimates.

\vspace{-2ex}\subparagraph*{Algorithm description.} \multicastadvadp contains multiple super-epochs, and the first super-epoch number is $I_b=2\lg{C}+20$. In super-epoch $i$, there are $bi$ phases numbered from $0$ to $bi-1$, where $b$ is some large constant. Each phase contains three steps. For any super-epoch $i$, the length of each step is always $R_i=a\cdot 2^i\cdot i^3$, where $a$ is some large constant. Prior to execution, all nodes are in \texttt{init} status.
Similar to \multicastadp, each node $u$ maintains $M_u$ to indicate whether it knows the message $m$ or not.
%Similar to \multicastadp, each node $u$ maintains a Boolean variable $M_u$ to indicate whether it knows the message $m$ or not. Thus, initially, only the source node sets $M_u$ to $true$, and every other node sets $M_u$ to $false$.

We now describe nodes' behavior in each $(i,j)$-phase---i.e., phase $j$ of super-epoch $i$---in detail. For each slot in an $(i,j)$-phase, each node will go to a channel chosen uniformly at random. Then, for each node $u$, it will broadcast or listen on the chosen channel, each with a certain probability. In step one and two, this probability is $p_u^{i,j}$; in step three, this probability is $p_{step3}^{i}=C^2/2^i$. We often call $p_u^{i,j}$ as the working probability of node $u$. Notice, at the beginning of an super-epoch $i$, the probability $p_u^{i,j}$, which is actually $p_u^{i,0}$, is set to $C/2^i$. In a slot, if $u$ chooses to send, then the broadcast content depends on the value of $M_u$: if $M_u$ is $true$ then $u$ will broadcast $m$, otherwise $u$ will broadcast a beacon message $\pm$. On the other hand, if $u$ chooses to listen in a slot, then it will record the channel feedback. One point worth noting is, a node $u$ will only change $M_u$ from $false$ to $true$ if it hears message $m$ in step one. (The purpose of this somewhat strange design is to facilitate analysis.)

At the end of each phase $j$, nodes will compute $p_u^{i,j+1}$ (i.e., the working probability of the next phase). Specifically, for each node $u$, define $\Delta^{step1}_u=\Delta^{step2}_u=R_ip^{i,j}_u/(1-p^{i,j}_u/C)$ and $\Delta^{step3}_u=R_ip^{i}_{step3}/(1-p^{i}_{step3}/C)$. Let $N_u^{step1,c}$, $N_u^{step2,c}$, and $N_u^{step3,c}$ denote the number of silent slots $u$ observed in step one, step two, and step three in phase $j$, respectively. Then, $\eta_u^{i,j}=N_u^{step1,c}/\Delta^{step1}_u+N_u^{step2,c}/\Delta^{step2}_u+N_u^{step3,c}/\Delta^{step3}_u$, and $p_u^{i,j+1}=p_u^{i,j}\cdot 2^{\max\{0,\eta_u^{i,j}-2.5\}}$.

At the end of each phase $j$, nodes will also potentially change their status. Specifically, if a node $u$ is in \texttt{init} status and finds: (a) $\eta_u^{i,j}\geq 2.4$; and (b) it has heard the message $m$ at least $ai^3$ times during step two of phase $j$. Then, node $u$ will become \texttt{helper} and compute an estimate of $n$ as $n_u=C/((p_u^{i,j})^2\cdot 2^i)$. On the other hand, if $u$ is already a \texttt{helper} and finds $p_u^{i,j+1}\geq 64\sqrt{C/(2^i\cdot n_u)}$, then $u$ will change its status to \texttt{halt} and stop execution. Pseudocode of \multicastadvadp is provided in Figure \ref{fig-alg-multicastadvadp} in Appendix \ref{sec-app-code}.

\section{Analysis of \multicastadvadp}

Throughout the analysis, when considering an $(i,j)$-phase, we often omit the indices $i$ and/or $j$ if they are clear from the context. For any node $u$, we often use $p_u$ to denote its working probability in a step. We always use $V$ to denote active nodes, and use $M$ to denote active nodes with $M_u=true$. Omitted proofs and some auxiliary materials are in Appendix \ref{sec-app-multicastadvadp-proof}.

\vspace{-2ex}\subparagraph*{The ``bounded difference'' property.} The main goal of this part is to show nodes' working probabilities can never differ too much. This ``bounded difference'' property is used extensively in remaining analysis, either explicitly or implicitly.

\begin{lemma}\label{lemma-multicastadvadp-pu-pv-bounded}
Consider a super-epoch $i>\lg{n}$. With probability at least $1-\exp(-\Theta(iC))$, we have $1/2\leq p_u/p_v\leq 2$ for any two nodes $u$ and $v$ at any phase of the super-epoch.
\end{lemma}

At a high level, the above lemma holds because the fraction of silent slots nodes observed during a phase cannot differ too much. To prove it formally, we show the following claim via a coupling argument. However, details of the coupling differ from the ones we saw in Section \ref{sec-multicastadp-analysis}. Specifically, we divide jamming results vectors into $\binom{R+C}{C}$ categories.

\begin{claim}\label{claim-multicastadvadp-pu-pv-bounded}
Consider a step of length $R$ and two active nodes $u$ and $v$. Let $p_u$ (resp., $p_v$) be the sending/listening probabilities of $u$ (resp., $v$); and let $X_u$ (resp., $X_v$) be the number of silent slots $u$ (resp., $v$) observed. Define $\Delta_u=Rp_u/(1-p_u/C)$ and $\Delta_v=Rp_v/(1-p_v/C)$. For constant $g\leq a/20$, define $\chi_u=\sqrt{{giC}/{(Rp_u)}}$ and $\chi_v=\sqrt{{giC}/{(Rp_v)}}$. Then:
\begin{enumerate}
	\item $\Pr[X_u/\Delta_u>1]\leq\exp(-\Theta(i^3C))$.
	\item $\Pr[(X_u/\Delta_u>0.2)\wedge(X_v/\Delta_v<0.1)]\leq\exp(-\Theta(i^3C))$.
	\item $\Pr[(|X_u/\Delta_u-X_v/\Delta_v|\geq \chi_u+\chi_v)\wedge((X_u/\Delta_u\geq 0.1)\wedge(X_v/\Delta_v\geq 0.1))]\leq\exp(-\Theta(iC))$.
\end{enumerate}
\end{claim}

\begin{proof}[Proof sketch]
We begin with part (1). Define $\alpha=\prod_{w\in V}(1-p_w/C)$. To make $X_u$ as large as possible, assume Eve does no jamming, thus whether $u$ hears silence are independent among different slots. Notice that $\mathbb{E}[X_u]=p_u\cdot(\prod_{v\in V\setminus\{u\}}(1-p_v/C))\cdot R=\alpha\cdot\Delta_u<\Delta_u$. Therefore, by a Chernoff bound, the probability that $X_u>\Delta_u$ is at most $\exp(-\Theta(\Delta_u))=\exp(-\Omega(i^3C))$.

Proofs for part (2) and (3) both rely on coupling, and we only focus on part (2) here.

We first setup the coupling. Assume the randomnesses of nodes come from $C$ lists $(\bm{T}_{0},\cdots,\bm{T}_{C})$.
%, where \highlight{$\bm{T}_{l}=(\bm{T}_{l}^{(1)},\bm{T}_{l}^{(2)},\cdots,\bm{T}_{l}^{(R)})$ for each $0\leq l\leq C$.} 
Specifically, for each slot $i$ in the step, if the jamming result is $Q_i\subseteq[C]$, then nodes' behavior in this slot is determined by $\Psi_{Q_i}\left({T}_{|Q_i|}^{(\sum_{j\leq i}\mathbb{I}[|Q_j|=|Q_i|])}\right)$ using permutation $\pi_{Q_i}$ and bijection $\Psi_{Q_i}$.
Notice, $\pi_{Q_i}$ and $\Psi_{Q_i}$ are defined in Section \ref{sec-multicastadp-analysis} on page \pageref{def-psi_q-pi_q}, and ${T}_{|Q_i|}^{\sum_{j\leq i}\mathbb{I}[|Q_j|=|Q_i|]}$ is the $(\sum_{j\leq i}\mathbb{I}[|Q_j|=|Q_i|])$-th chunk in list $\bm{T}_{|Q_i|}$.
Let $X_{u,i}$ be an indicator random variable taking value 1 iff $u$ hears silence in the $i$\textsuperscript{th} slot, define $X_u=\sum_{i=1}^{R}X_{u,i}$.

Define $\mathcal{Z}=\{\bm{z}=\langle z_1,z_2,\cdots,z_C \rangle\in\mathbb{N}^C:\sum_{l=1}^{C}z_l\leq R\}$, thus $|\mathcal{Z}|=\binom{R+C}{C}\leq(R+1)^C\leq(2R)^C$. (Intuitively, for every $l\in[C]$, $z_l$ in $\bm{z}$ is the number of slots in which Eve leaves $l$ channels unjammed.) Denote the jamming results of this step as $\bm{Q}=(Q_1,\cdots,Q_R)\in \mathcal{Q}$, and define $|\bm{Q}|=\sum_{i=1}^{R}|Q_i|$. Further define function $K:\mathcal{Q}\to \mathcal{Z}$ such that $K(\bm{Q})=\langle K_1(\bm{Q}),\cdots,K_C(\bm{Q})\rangle$, where $K_l(\bm{Q})=\sum_{i=1}^R \mathbb{I}[|Q_i|=l]$. (That is, $K_l(\bm{Q})$ counts the number of slots in which Eve leaves $l$ channels unjammed.) Hence, given $K(\bm{Q})$, we can use a function $L:\mathbb{N}^C\to\mathbb{N}$ to compute $|\bm{Q}|$. In particular, $L(\bm{z})=\sum_{l=1}^{C}z_l\cdot l$ and $L(K(\bm{Q}))=|\bm{Q}|$.

Now, consider another execution, for any $j\geq 1$ and $l\in[C]$, let $Y_{u,l}^{(j)}$ be an indicator random variable taking value 1 iff $u$ hears silence in a slot in which the jamming result is $[l]$ and the behavior of nodes is determined by the $j$\textsuperscript{th} chunk of $\bm{T}_l$ directly.
%(Notice, $Y_{u,l}^{(j)}$ is completely determined by the jamming result $[l]$ and $\bm{T}_l^{(j)}$.)
Define $Y_u(\bm{z})=\sum_{l=1}^{C}\sum_{j=1}^{z_l}Y_{u,l}^{(j)}$ for any $\bm{z}\in \mathcal{Z}$. By definition, it is easy to verify $X_u(\bm{Q})=Y_u(K(\bm{Q}))$ for any $\bm{Q}$. That is, for any $\bm{Q}$, values of $X_u$ and $Y_u$ are identical. The significance of this observation is that it relates $X_u$---which counts the number of silent slots $u$ heard---to $Y_u$, and $Y_u$ can be interpreted as the sum of independent random variables once $\bm{z}$ is fixed.

At this point, we are ready to prove part (2). Notice $\mathbb{E}[X_u]/\Delta_u=\mathbb{E}[X_v]/\Delta_v=\alpha\cdot|\bm{Q}|/(RC)$ where $\alpha=\prod_{w\in V}(1-p_w/C)$. Also, it is easy to verify $\mathbb{E}[Y_u(\bm{z})]/\Delta_u=\mathbb{E}[Y_v(\bm{z})]/\Delta_v=\alpha\cdot L(\bm{z})/(RC)$. Let $\mathcal{Z}_1=\{\bm{z}\in\mathcal{Z}:L(\bm{z})\leq0.15RC/\alpha\}$. Then for $\bm{z}\in \mathcal{Z}_1$,  $\mathbb{E}[Y_u(\bm{z})]\leq 0.15\Delta_u$, further by a Chernoff bound, $\Pr(Y_u(\bm{z})>0.2\Delta_u)\leq \exp(-\Theta(i^3C))$. Similarly, for $\bm{z}\in\mathcal{Z}\setminus \mathcal{Z}_1$, $\Pr(Y_v(\bm{z})<0.1\Delta_v)\leq \exp(-\Theta(i^3C))$. Therefore, we can conclude $\Pr\left(X_u(\bm{Q})>0.2\Delta_u \wedge X_v(\bm{Q})<0.1\Delta_v \right) \leq \sum_{\bm{z}\in\mathcal{Z}_1}\Pr(Y_u(\bm{z})>0.2\Delta_u)+\sum_{\bm{z}\in\mathcal{Z}\setminus \mathcal{Z}_1}\Pr(Y_v(\bm{z})<0.1\Delta_v) \leq |\mathcal{Z}|\cdot \exp(-\Theta(i^3C))=\exp(-\Theta(i^3C))$.
\end{proof}

We now sketch the proof of Lemma \ref{lemma-multicastadvadp-pu-pv-bounded}. Denote the working probabilities of the current phase and the next phase as $p$ and $p'$. If $\eta_u\leq 2.5$ and $\eta_v\leq 2.5$, then $p'_u/p'_v=p_u/p_v$ and we are done. So assume $\eta_u> 2.5$. In such case, Claim \ref{claim-multicastadvadp-pu-pv-bounded} imply $|N^{c,step*}_u/\Delta^{step*}_u-N^{c,step*}_v/\Delta^{step*}_v|\leq \sqrt{{giC}/{(Rp_u)}}+\sqrt{{giC}/{(Rp_v)}}$ for any step $*$ in $\{1,2\}$, and $|N^{c,step3}_u/\Delta^{step3}_u-N^{c,step3}_v/\Delta^{step3}_v|\leq 2\sqrt{{giC}/{(Rp_{step3})}}$. This further suggests ${p'_u}/{p'_v}\leq({p_u}/{p_v})\cdot 2^{1/bi}$, thus the lemma is proved.

\vspace{-2ex}\subparagraph*{Correctness.} This part shows \multicastadvadp enforces two nice properties. First, when some node halts, all nodes must have become \texttt{helper}. This property can be seen as a stronger version of Lemma \ref{lemma-multicastadp-halt-imply-informed}, which implies all nodes are informed before any node halts, since a node must have heard the message $m$ when becoming a \texttt{helper}. The second property, on the other hand, states when a node becomes \texttt{helper}, it also obtains a good estimate of $n$. This property helps to ensure nodes can stop execution at the right time.

%This part shows \multicastadvadp enforces two nice properties. First, when some node halts, all nodes must have become \texttt{helper}. Since a node must have learned the message $m$ before becoming a \texttt{helper}, this property implies all nodes are informed before any node halts. This property can be seen as a stronger version of Lemma \ref{lemma-multicastadp-halt-imply-informed}, as \multicastadvadp uses a two-stage termination mechanism. The second property, on the other hand, states when a node becomes \texttt{helper}, it also obtains a good estimate of $n$. This property helps ensure nodes can correctly determine when to stop execution.

\begin{lemma}[``halt-imply-helper'' property]\label{lemma-multicastadvadp-halt-imply-helper}
The probability that some node has stopped execution while some other node has not become \texttt{helper} is at most $n^{-\Omega(1)}$.
\end{lemma}

\begin{lemma}[``good-estimate'' property]\label{lemma-multicastadvadp-good-estimate}
For each node $u$, the probability that $u$ becomes \texttt{helper} with $n_u<n/256$ or $n_u>4n$ is at most $n^{-\Omega(1)}$.
\end{lemma}

The following lemma is helpful for proving both of the above two properties. Roughly speaking, this lemma states that if in an $(i,j)$-phase some node $u$ has working probability $p_u=\Theta(\sqrt{C/(2^in)})$ and decides to raise $p_u$ at the end of the phase, then all nodes must have heard the message many times in step two of the phase.

\begin{lemma}\label{lemma-multicastadvadp-pu-raise-and-msg}
Consider an $(i,j)$-phase where $i>\lg{n}$. Assume at the beginning of the phase: $(\sum_{u\in V}{p_u})/C\leq 1/2$, all nodes are active and their working probabilities are within a factor of two, and the working probability of each node is at least $8\sqrt{{C}/{(2^in)}}$. Then, with probability at most $\exp(-\Theta(i^2))$, these two events both occur: (a) some node raises its working probability at the end of the phase; and (b) some node hears message $m$ less than $ai^3$ times in step two.
\end{lemma}

\begin{proof}[Proof sketch]
Let $\mathcal{E}_{R}$ be the event that some node raises its working probability by the end of the phase, $\mathcal{E}_{M}$ be the event that some node hears $m$ less than  $ai^3$ times during step two, $\mathcal{E}_{un}$ be the event that some node is still uninformed by the end of step one. Moreover, let $\mathcal{E}_1$ (respectively, $\mathcal{E}_2$) be the event that $\mathcal{E}_{step1}^{\geq 0.25}(\geq 0.25)$ (respectively, $\mathcal{E}_{step2}^{\geq 0.25}(\geq 0.25)$) occurs during step one (respectively, step two) of the phase. We know:

\begin{align*}
\Pr(\mathcal{E}_M\mathcal{E}_R)\leq &\Pr(\mathcal{E}_M\wedge(\mathcal{E}_1\wedge\mathcal{E}_2))+\Pr(\mathcal{E}_R\wedge\overline{(\mathcal{E}_1\wedge\mathcal{E}_2)})\\
\leq &\Pr(\mathcal{E}_{un}\mathcal{E}_1)+\Pr(\overline{\mathcal{E}_{un}}\mathcal{E}_M\mathcal{E}_2)+\Pr(\mathcal{E}_R \wedge (\overline{\mathcal{E}_1}\vee \overline{\mathcal{E}_2}))
\end{align*}

The reminder of the proof bounds the three probabilities in the last line.

\emph{Claim I: $\Pr(\mathcal{E}_{un}\mathcal{E}_1)\leq\exp(-\Theta(i^2))$. Proof sketch:} If $\mathcal{E}_1$ happens, then step one is not heavily jammed. Thus every node will be informed at the end of step one due to the effectiveness of the epidemic broadcast scheme, much like the proof of Lemma \ref{lemma-multicastadp-fast-bcst}.

\emph{Claim II: $\Pr(\overline{\mathcal{E}_{un}}\mathcal{E}_{M}\mathcal{E}_2)\leq\Pr(\mathcal{E}_{M}\mathcal{E}_2|\overline{\mathcal{E}_{un}})\leq\exp(-\Theta(i^3))$. Proof sketch:} Fix a node $u$, and assume all nodes know $m$ at the beginning of step two. Similar to the proof of Lemma \ref{lemma-multicastadp-fast-halt} (except that we focus on message slots and apply the coupling argument accordingly), the probability that $u$ hears $m$ less than $ai^3$ times during a step two when $\mathcal{E}_2$ occurs is at most $\exp(-\Theta(i^3))$. Take a union over all nodes and the claim is proved.

\emph{Claim III: $\Pr(\mathcal{E}_{R}\wedge(\overline{\mathcal{E}_1}\vee \overline{\mathcal{E}_2}))\leq\exp(-\Theta(i^3C))$. Proof sketch:} Notice that $\Pr(\mathcal{E}_{R}\wedge(\overline{\mathcal{E}_1}\vee\overline{\mathcal{E}_2})) \leq \Pr(\mathcal{E}_{R}\overline{\mathcal{E}_1})+\Pr(\mathcal{E}_{R}\overline{\mathcal{E}_2}) \leq \sum_{u\in V}\Pr(\mathcal{E}_{u,1}\overline{\mathcal{E}_1})+\sum_{u\in V}\Pr(\mathcal{E}_{u,2}\overline{\mathcal{E}_{2}})+4\sum_{u\in V}\exp(-\Theta(i^3C))$. Here, $\mathcal{E}_{u,1}$ (resp., $\mathcal{E}_{u,2}$) is the event that node $u$ hears silence more than $\Delta_{u}^{step1}/2$ (resp., $\Delta_{u}^{step2}/2$) times in step one (resp., step two) of the phase, and the last inequality is due to part (1) of Claim \ref{claim-multicastadvadp-pu-pv-bounded}. 
When $\overline{\mathcal{E}_1}$ occurs, the expected number of silent slots heard by $u$ in step one is at most $7/16\Delta_{u}^{step1}$. Again via a coupling argument, we know $\Pr(\mathcal{E}_{u,1}\overline{\mathcal{E}_1})\leq \exp(-\Theta(i^3))$, and bounding $\Pr(\mathcal{E}_{u,2}\overline{\mathcal{E}_2})$ is similar.

%Recall some notations defined in the proof of Claim \ref{claim-multicastadvadp-pu-pv-bounded}. When $\overline{\mathcal{E}_1}$ occurs, the jamming results $\bm{Q}$ of step one satisfy $L(K(\bm{Q}))\leq(1/4\cdot 3/4 + 1\cdot 1/4)R_iC=7R_iC/16$. Moreover, we can show $\mathbb{E}[Y_u(\bm{z})]<\Delta^{step1}_u/2$ when $L(\bm{z})\leq 7R_iC/16$. Hence, via a coupling argument, $\Pr(\mathcal{E}_{u,1}\overline{\mathcal{E}_1})\leq (2R)^C \cdot\exp(-\Theta(\Delta^{step1}_u))\leq \exp(-\Theta(i^3C))$.
\end{proof}

At this point, to prove the ``halt-imply-helper'' property, we only need to combine the above lemma with the following two observations. First, nodes are unlikely to become \texttt{helper} in early super-epochs, as the sending probabilities in these super-epochs are too high and nodes cannot hear enough silent slots. (See Lemma \ref{lemma-multicastadvadp-no-helper-in-early-epoch} in Appendix \ref{sec-app-multicastadvadp-proof} on page \pageref{lemma-multicastadvadp-no-helper-in-early-epoch}.) Second, when nodes' working probabilities in step two are too small, they will also not become \texttt{helper} as the number of messages heard is not enough. (See Lemma \ref{lemma-multicastadvadp-estimate-upper-bound} in Appendix \ref{sec-app-multicastadvadp-proof} on page \pageref{lemma-multicastadvadp-estimate-upper-bound}.) Notice, this second observation also leads to an upper bound on the estimates of $n$.

To prove the ``good-estimate'' property, what remains is to show a lower bound for $n_u$. To that end, we show if all nodes are alive and $u$'s working probability is close to the ideal value $\Theta(\sqrt{C/(2^in)})$, then $u$ must have become \texttt{helper} already. (See Claim \ref{claim-multicastadvadp-pu-high-implies-helper} in Appendix \ref{sec-app-multicastadvadp-proof} on page \pageref{claim-multicastadvadp-pu-high-implies-helper}.) By then, a lower bound of $n_u$ can be derived as a simple corollary of this claim.

\vspace{-2ex}\subparagraph*{Termination.} This part shows nodes will quickly become \texttt{helper} and then halt once jamming from Eve becomes weak. (In other words, Eve cannot delay nodes unless she spends a lot of energy.) We begin by classifying phases and super-epochs into \emph{weakly jammed} ones and \emph{strongly jammed} ones. Specifically, call a phase weakly jammed if event $\mathcal{E}^{\geq 0.95}(\geq 0.95)$ occurs for all three steps of the phase. On the other hand, if event $\mathcal{E}^{>0.05}(<0.95)$ occurs for any of the three steps, then the phase is strongly jammed. Call a super-epoch weakly jammed if at least half of the phases in the super-epoch are weakly jammed, otherwise the super-epoch is strongly jammed.

We first show, if a node's working probability has not reached the ideal value, then this probability will increase by some constant factor in a weakly jammed phase.

\begin{lemma}\label{lemma-multicastadvadp-pu-inc-in-weak-jam-phase}
Fix an $(i,j)$-phase where $i\geq\lg(nC)+6$, and fix an active node $u$ satisfying $p^{i,j}_u<C/(128n)$. By the end of the phase, the following two events happen simultaneously with probability at most $\exp(-\Omega(iC))$: (a) the phase is weakly jammed; and (b) $p^{i,j+1}_u< p_u^{i,j}\cdot 2^{(1/10)}$.
\end{lemma}

Building upon Lemma \ref{lemma-multicastadvadp-pu-inc-in-weak-jam-phase}, we can prove nodes' working probabilities will reach $\tilde{p_i}=1024\sqrt{{C}/{(2^in)}}$ in a weakly jammed super-epoch, as there are enough weakly jammed phases.

\begin{lemma}\label{lemma-multicastadvadp-pu-inc-in-weak-jam-epoch}
Fix a super-epoch $i\geq 34+\lg(nC)$ and a node $u$ that is active at the beginning of the super-epoch. The following two events happen simultaneously with probability at most $\exp(-\Omega(iC))$: (a) the super-epoch is weakly jammed; and (b) by the end of the super-epoch $u$ is still alive with a working probability less than $\tilde{p_i}$.
\end{lemma}

Lastly, we show that when a node's working probability reaches $\tilde{p_i}$, it will halt.

\begin{lemma}\label{lemma-multicastadvadp-large-pu-halt}
Fix a super-epoch $i\geq\lg{(nC)}-7$ and a node $u$. Assume when $u$ becomes \texttt{helper}, it is true that $n_u\geq n/256$ and all nodes are active at the beginning of that phase. Then, the probability that $u$ is active at the end of super-epoch $i$ with a working probability exceeding $\tilde{p_i}$ is at most $\exp(-\Theta(i))$.
\end{lemma}

\vspace{-3ex}\subparagraph*{Main theorem.} In this last part we sketch the proof of Theorem \ref{thm-multicastadvadp}.

Fix an arbitrary node $u$. The first step is to analyze how long $u$ remains active. Since super-epoch length increases geometrically, we only need to focus on the last super-epoch in which $u$ is active. Specifically, let $\hat{I}=34+\lg C+\max\{\lg C,\lg n\}$, let $r_i$ be the number of slots in super-epoch $i$, and let $sr_i=\sum_{k=\hat{I}+1}^{i}r_k$ be the total number of slots from super-epoch $\hat{I}+1$ to super-epoch $i$. It is easy to verify, for $i\geq \hat{I}+1$, $sr_{i}\leq 5r_{i-1}$. Define constant $\beta=2400$, and let random variable $L$ denote node $u$'s actual runtime starting from super-epoch $\hat{I}+1$. Combine Lemma \ref{lemma-multicastadvadp-halt-imply-helper}, \ref{lemma-multicastadvadp-good-estimate}, \ref{lemma-multicastadvadp-pu-inc-in-weak-jam-epoch}, \ref{lemma-multicastadvadp-large-pu-halt}, along with the fact that Eve spends less than $r_iC/\beta=bi/2\cdot 0.05^2R_iC$ in super-epoch $i$ implies super-epoch $i$ is weakly jammed, we can prove $L\leq 5\beta T/C$ holds w.h.p. Take a union bound over all nodes, we know every node will terminate within $(\sum_{k=I_b}^{\hat{I}}bk\cdot 3R_k)+5\beta T/C=\tilde{O}(T/C+nC+C^2)$ slots, w.h.p.

Next, we analyze the cost of node $u$. Let $F_{step1,2}$ (respectively, $F_{step3}$) be node $u$'s total actual cost during step one and step two (respectively, step three) starting from super-epoch $\hat{I}+1$. By an analysis similar to above, we show $F_{step1,2}\leq \Theta(\sqrt{T/n}\cdot\lg^2{T})$ and $F_{step3}\leq \Theta(C^2\cdot(\lg{(T)}+\hat{I})^5)$, w.h.p. As a result, we can conclude w.h.p.\ the energy cost of each node is bounded by $F_{step1,2}+F_{step3}+\sum_{k=I_b}^{\hat{I}}(bk\cdot 3R_k)=\tilde{O}(\sqrt{T/n}+nC+C^2)$.

Finally, notice the algorithm itself ensures a node must be informed when it halts.

\section{Lower Bounds}\label{sec-lower-bound}

In this section, we show our algorithms achieve (near) optimal time and energy complexity simultaneously against an adaptive adversary with budget $T$. The time complexity part is obvious: Eve can jam all channels during the first $T/C$ slots, so the $O(T/C)$ term in the runtime of \multicastadp and \multicastadvadp is asymptotically optimal.

Obtaining an energy complexity lower bound is more involved. To do so, the first step is a simulation argument. Specifically, given any fair multi-channel broadcast algorithm $\mathcal{A}_n$, we can devise a multi-channel 1-to-1 communication algorithm $\mathcal{A}_2$ (in which the goal is to let one node called Alice to send a message $m$ to another node Bob) that simulates $\mathcal{A}_n$ internally. To make the simulation feasible, we allow Alice and Bob to have multiple transceivers, so that in each slot they can operate on multiple channels, as well as send and listen simultaneously. In more detail, Alice in $\mathcal{A}_2$ mimics the source node in $\mathcal{A}_n$. As for Bob, he simulates the $n-1$ non-source nodes in $\mathcal{A}_n$. Particularly, in each slot, for each channel, if at least one non-source node listens, then Bob uses a transceiver to listen; if exactly one non-source node broadcasts, then Bob uses a transceiver to broadcast the unique message; and if at least two non-source nodes broadcast, then Bob uses a transceiver to broadcast noise. (Notice Bob can simultaneously listen and broadcast on a channel: he uses two transceivers and incurs two units of cost.) On the other hand, Eve's strategy for disrupting $\mathcal{A}_n$ and $\mathcal{A}_2$ is called $\mathcal{S}$: in each slot, for each channel, if the probability that the source node (resp., Alice) successfully transmits $m$ to some non-source node (resp., Bob) exceeds $1/T$, then Eve jams that channel.

Clearly, the above simulation is ``perfect'': an execution of $\mathcal{A}_2$ is identical to an execution of $\mathcal{A}_n$, assuming nodes and Eve use identical random bits in the two executions. To simplify presentation, we further assume $\mathcal{A}_n$ automatically stops once all nodes are informed, and $\mathcal{A}_2$ automatically stops once Bob is informed. This modification will not increase nodes' energy cost, thus will not affect the correctness of our lower bound. Now, observe that the success of $\mathcal{A}_2$ is a necessary condition for the success of $\mathcal{A}_n$, and Bob's energy cost will not exceed the sum of all non-source nodes' cost, hence the following lemma is immediate.

\begin{lemma}\label{lemma-lower-bound-1-to-1-sim-broadcast}
For any fair multi-channel broadcast algorithm $\mathcal{A}_n$, there exists a multi-channel 1-to-1 communication algorithm $\mathcal{A}_2$. If in $\mathcal{A}_n$ each node incurs an expected cost of $f(T)$ and $\mathcal{A}_n$ succeeds with probability $p$, then: (a) in $\mathcal{A}_2$ Alice and Bob incur an expected cost of at most $f(T)$ and $n\cdot f(T)$, respectively; (b) $\mathcal{A}_2$ succeeds with probability at least $p$.
\end{lemma}

What remains is an energy complexity lower bound for $\mathcal{A}_2$: with such a result, Theorem \ref{thm-cost-lower-bound} is immediate via simple reduction. Indeed, we are able to prove Theorem \ref{thm-cost-lower-bound-1-to-1}, an energy complexity lower bound for 1-to-1 communication in the multi-channel setting. This result could be of independent interest, and at a high-level its proof is organized in the following way. First, we note that in a rough sense, any multi-channel 1-to-1 communication algorithm $\mathcal{A}$ can be viewed as a decision tree, and each path from the root to a leaf in the tree corresponds to an oblivious algorithm. Then, we argue that $\mathcal{A}$ can be used to generate another algorithm $\mathcal{A}'$ which is a ``convex combination'' (or, a distribution) of all such oblivious algorithms, without changing the success probability or the product of Alice's and Bob's expected cost. Moreover, an important observation is that among all the oblivious algorithm used in the ``convex combination'', at least one---say $\mathcal{A}_{\hat{w}}$---is (roughly) as good as $\mathcal{A}'$ in terms of both success probability and energy efficiency. Finally, depending on whether Eve uses all her budget during execution, we consider two potential scenarios for $\mathcal{A}_{\hat{w}}$, and for both we show $\mathbb{E}_{\mathcal{A}_{\hat{w}}}[A]\cdot\mathbb{E}_{\mathcal{A}_{\hat{w}}}[B]\in\Omega(T)$, which in turn implies $\mathbb{E}_\mathcal{A}[A]\cdot\mathbb{E}_\mathcal{A}[B]\in\Omega(T)$. Complete proof of Theorem \ref{thm-cost-lower-bound-1-to-1} is deferred to Appendix \ref{sec-app-lower-bound-proof} on page \pageref{sec-app-lower-bound-proof} due to space constraint.

\begin{theorem}\label{thm-cost-lower-bound-1-to-1}
Consider any multi-channel 1-to-1 communication algorithm that succeeds with constant probability against an adaptive adversary Eve with budget $T$. Let $A$ and $B$ denote Alice's and Bob's expected cost respectively, then Eve can force $\mathbb{E}[A]\cdot\mathbb{E}[B]\in\Omega(T)$.
\end{theorem}

% references
\bibliographystyle{plainurl}
\bibliography{./opodis20-ref}

% appendix
\clearpage
\appendix

\section{More Discussion on Handling the Adaptivity of the Adversary}\label{sec-app-adaptive-discussion}

It is worth noting techniques like ``the principle of deferred decision'' cannot resolve the dependency issue directly. Specifically, assume we first reveal the jamming results of Eve over one epoch, and only then reveal the choices made by node $u$ to determine its behavior. It is tempting to think, if in each slot our algorithm enforces nodes to determine their behavior without considering execution history, then the choices made by $u$ would be independent of the jamming results of Eve. But this is not true! Consider a simple system containing only two nodes $u$ and $v$, further assume there are only two channels and two time slots. Imagine Eve employs the following jamming strategy. In slot one, she does no jamming. In slot two, she jams both channels iff $u$ heard a message from $v$ in slot one. Now, if the jamming results turn out to be no jamming during slot two, then we can be assured $u$ did not hear from $v$ during slot one. On the other hand, if the jamming results indicate jamming during slot two, then $u$ must have received a message from $v$ during slot one. Fundamentally, this is because the adaptivity of Eve allows $Q_i$ to depend on $\{G_1,G_2,\cdots,G_{i-1}\}$. Hence, once the jamming results are fixed, the distribution for honest nodes' behavior is potentially distorted.

Some previous work also consider an adaptive adversary, but either they do not have the problem we face, or their solution is not directly applicable to our setting. In particular, the Las Vegas algorithm King et al.~\cite{king11} devised for the 1-to-1 communication problem does not introduce the dependency issue. For the broadcast problem, in the single-channel setting, Gilbert et al.~\cite{gilbert14} observed that jammed slots within one epoch can be postponed to the end of that epoch without loss of generality (see Lemma 1 of \cite{gilbert14}). However, it is unclear how to correctly extend this observation to the multi-channel setting. 

\section{Pseudocode of Algorithms}\label{sec-app-code}

\subsection{Pseudocode of \multicastadp}

\begin{figure}[!h]
\hrule
\vspace{1ex}\textbf{\multicastadp executed at node $u$:}\vspace{1ex}
\hrule
\begin{small}
\begin{algorithmic}[1]
\State $M_u\gets false$.
\If {($u$ is the source node)} $M_u\gets true$. \EndIf
\For {(each epoch $i\geq I_b$)}
	\State $N_c\gets 0$, $p\gets (\sqrt{C/n})/2^i$.
	\For {(each slot from $1$ to $R=a\cdot 4^i\cdot i\cdot\lg^2{n}$)}
		\State $ch\gets\texttt{random}(1,C)$, $rnd\gets\texttt{random}(1,1/p)$.
		\If {($rnd==1$)}
			\State $feedback\gets\texttt{listen}(ch)$.
			\If {($feedback$ is silence)}
				%\State 
				$N_c\gets N_c+1$.
			\ElsIf {($feedback$ includes message $m$)}
				%\State 
				$M_u\gets true$.
			\EndIf
		\ElsIf {($rnd==2$)}
			\If {($M_u==true$)}
				%\State 
				$\texttt{broadcast}(ch,m)$.
			\Else
				%\State 
				\ $\texttt{broadcast}(ch,\pm)$.
			\EndIf
		\EndIf
	\EndFor
	\If {($N_c\geq Rp/2$)} \textbf{halt}. \EndIf
\EndFor
\end{algorithmic}
\end{small}
\hrule
\vspace{1ex}
\caption{Pseudocode of the \multicastadp algorithm. (In the pseudocode: (a) $\texttt{random}(x,y)$ returns a uniformly chosen random integer from the interval $[x,y]$; (b) $\texttt{listen}(ch)$ instructs the node to listen on channel $ch$ and returns channel status; and (c) $\texttt{broadcast}(ch,m)$ instructs the node to broadcast $m$ on $ch$. We also assume $1/p$ is an integer for the ease of presentation.)}\label{fig-alg-multicastadp}
\vspace{-3ex}
\end{figure}

\clearpage

\subsection{Pseudocode of \multicastadvadp}

\begin{figure}[!h]
\hrule
\vspace{1ex}\textbf{\multicastadvadp executed at node $u$:}\vspace{1ex}
\hrule
\begin{small}
\begin{algorithmic}[1]
\State $status_u\gets init, M_u \gets false, n_u \gets -1, b\gets 20$.
	\If {($u$ is the source node)} $M_u \gets true$.\EndIf
	\State $i\gets 2\lg{C}+20$.
\While {($true$)}
	\State $p^{i,0}_u\gets C/2^i, p^i_{step3}\gets C^2/2^i, R_i\gets a\cdot 2^i\cdot i^3$.
	\For {(each phase $j$ from $0$ to $bi-1$)}
		\State $N^{step1,c}_u\gets 0$, $N^{step2,c}_u\gets 0$, $N^{step3,c}_u\gets 0$, $N^{step2,m}_u\gets 0$.
		\LineComment{\textsc{Step 1}.}
		\For {(each of the $R_i$ slots)} %\Comment{\textsc{Step 1}.}
			\State $ch\gets\texttt{random}(1,C)$, $rnd\gets\texttt{random}(1,1/p^{i,j}_u)$.
			\If {($rnd==1$)}
				\State $feedback\gets\texttt{listen}(ch)$.
				\If {($feedback$ is silence)} $N^{step1,c}_u\gets N^{step1,c}_u +1$.
				\ElsIf {($feedback$ includes message $m$)} $M_u \gets true$.
				\EndIf
			\ElsIf {($rnd==2$)}
		 		\If {($M_u == true$)} $\texttt{broadcast}(ch,m)$.
				\Else\ $\texttt{broadcast}(ch,\pm)$.
				\EndIf
			\EndIf
		\EndFor
		\LineComment{\textsc{Step 2}.}
		\For {(each of the $R_i$ slots)} %\Comment{\textsc{Step 2}.}
			\State $ch\gets\texttt{random}(1,C)$, $rnd\gets\texttt{random}(1,1/p^{i,j}_u)$.
			\If {($rnd==1$)}
				\State $feedback\gets\texttt{listen}(ch)$.
				\If {($feedback$ is silence)} $N^{step2,c}_u\gets N^{step2,c}_u+1$.
				\ElsIf {($feedback$ includes message $m$)} $N^{step2,m}_u\gets N^{step2,m}_u+1$.
				\EndIf
			\ElsIf {($rnd==2$)}
			 	\If {($M_u == true$)} $\texttt{broadcast}(ch,m)$.
				\Else\ $\texttt{broadcast}(ch,\pm)$.
				\EndIf
			\EndIf
		\EndFor
		\LineComment{\textsc{Step 3}.}
		\For {(each of the $R_i$ slots)} %\Comment{\textsc{Step 3}.}
			\State $ch\gets\texttt{random}(1,C)$, $rnd\gets\texttt{random}(1,1/p^i_{step3})$.
			\If {($rnd==1$)}
				\State $feedback\gets\texttt{listen}(ch)$.
				\If {($feedback$ is silence)} $N^{step3,c}_u\gets N^{step3,c}_u+1$.
				\EndIf
			\ElsIf {($rnd==2$)}
			 	\If {($M_u == true$)} $\texttt{broadcast}(ch,m)$.
				\Else\ $\texttt{broadcast}(ch,\pm)$.
				\EndIf
			\EndIf
		\EndFor
		\LineComment{\textsc{Postprocessing}.}
		\State $\Delta^{step1}_u=\Delta^{step2}_u\gets R_ip^{i,j}_u/(1-p^{i,j}_u/C)$, $\Delta^{step3}_u\gets R_ip^{i}_{step3}/(1-p^{i}_{step3}/C)$.
		\State $\eta\gets \frac{N^{step1,c}_u}{\Delta^{step1}_u}+\frac{N^{step2,c}_u}{\Delta^{step2}_u}+\frac{N^{step3,c}_u}{\Delta^{step3}_u}$.
		\State $p^{i,j+1}_u\gets p_u^{i,j} \cdot 2^{\max\{0,\eta-2.5\}}$.
		\If {($status_u==init$ \textbf{and}  $N^{step2,m}_u\geq ai^3$ \textbf{and} $\eta\geq 2.4$)}
			\State $status_u\gets helper$, $n_u\gets{C}/{((p^{i,j}_u)^2\cdot 2^i)}$.
		\ElsIf {($status_u==helper$ \textbf{and} $p^{i,j+1}_u\geq 64\sqrt{C/(2^in_u)}$)}
		%\ElsIf {$\left(status_u==helper \text{\textbf{ and }} p^{i,j+1}_u\geq 64\sqrt{\frac{C}{2^in_u}}\right)$}
			\State $status_u\gets halt$.
			\State \textbf{return}.
		\EndIf
	\EndFor
	\State $i\gets i+1$.
\EndWhile
\end{algorithmic}
\end{small}
\hrule
\vspace{1ex}
\caption{Pseudocode of the \multicastadvadp algorithm.}\label{fig-alg-multicastadvadp}
\vspace{-3ex}
\end{figure}

\clearpage

\section{Omitted Parts in the Analysis of \multicastadp}\label{sec-app-multicastadp-proof}

\begin{proof}[\underline{Proof of Claim $\bm{G}\sim\bm{\mathcal{G}}$}]
\begin{figure*}[!t]
\hrule
\vspace{-1ex}
\begin{small}
\begin{align*}
& \Pr[G_i=\sigma_i~|~\bm{G}_{<i}=\bm{\sigma}_{<i}] \\
=\quad \ \ & \sum_{\bm{q}_{\leq i}\in\left(2^{[C]}\right)^i}{\Pr\left[\bm{Q}_{\leq i}=\bm{q}_{\leq i}~|~\bm{G}_{<i}=\bm{\sigma}_{<i}\right]\cdot\Pr\left[G_i=\sigma_i~|~\bm{G}_{<i}=\bm{\sigma}_{<i}\wedge \bm{Q}_{\leq i}=\bm{q}_{\leq i}\right]} \\
%=& \sum_{\bm{q}_{\leq i}\in\left(2^{[C]}\right)^i}{\Pr\left[\bm{Q}_{\leq i}=\bm{q}_{\leq i}~|~\bm{G}_{<i}=\bm{\sigma}_{<i}\right]\cdot\Pr\left[G_i=\sigma_i~|~\bm{G}_{<i}=\bm{\sigma}_{<i}\wedge \bm{Q}_{<i}=\bm{q}_{< i}\wedge Q_i=q_i\right]} \\
=\quad \ \ & \sum_{\bm{q}_{\leq i}\in \Omega^{*}}{\left(\Pr\left[\bm{Q}_{\leq i}=\bm{q}_{\leq i}~|~\bm{G}_{<i}=\bm{\sigma}_{<i}\right]\cdot\Pr\left[G_i=\sigma_i~|~\bm{G}_{<i}=\bm{\sigma}_{<i}\wedge \bm{Q}_{\leq i}=\bm{q}_{\leq i}\right]\right)} \quad + \\
& \sum_{\bm{q}_{\leq i}\in \Omega^{**}}{\left(\Pr\left[\bm{Q}_{\leq i}=\bm{q}_{\leq i}~|~\bm{G}_{<i}=\bm{\sigma}_{<i}\right]\cdot\Pr\left[G_i=\sigma_i~|~\bm{G}_{<i}=\bm{\sigma}_{<i}\wedge \bm{Q}_{\leq i}=\bm{q}_{\leq i}\right]\right)} \quad + \\
& \sum_{\bm{q}_{\leq i}\in \Omega^{***}}{\left(\Pr\left[\bm{Q}_{\leq i}=\bm{q}_{\leq i}~|~\bm{G}_{<i}=\bm{\sigma}_{<i}\right]\cdot\Pr\left[G_i=\sigma_i~|~\bm{G}_{<i}=\bm{\sigma}_{<i}\wedge \bm{Q}_{\leq i}=\bm{q}_{\leq i}\right]\right)}  \\
\overset{\text{(equality 1)}}{=}& \sum_{\bm{q}_{\leq i}\in \Omega^{*}}{\left(\Pr\left[\bm{Q}_{\leq i}=\bm{q}_{\leq i}~|~\bm{G}_{<i}=\bm{\sigma}_{<i}\right]\cdot\Pr\left[\Psi_{q_i}\left(T_{hi}^{\left(K\left(\bm{q}_{\leq i}\right)\right)}\right)=\sigma_i\right]\right)} \quad + \\
& \sum_{\bm{q}_{\leq i}\in \Omega^{**}}{\left(\Pr\left[\bm{Q}_{\leq i}=\bm{q}_{\leq i}~|~\bm{G}_{<i}=\bm{\sigma}_{<i}\right]\cdot\Pr\left[\Psi_{q_i}\left(T_{lo}^{\left(i-K\left(\bm{q}_{\leq i}\right)\right)}\right)=\sigma_i\right]\right)} \quad +\\
& \sum_{\bm{q}_{\leq i}\in \Omega^{***}}{\left(\Pr\left[\bm{Q}_{\leq i}=\bm{q}_{\leq i}~|~\bm{G}_{<i}=\bm{\sigma}_{<i}\right]\cdot\Pr\left[\Psi_{q_i}\left(T_{lo}^{\left(i-y_1R\right)}\right)=\sigma_i\right]\right)} \\
\overset{\text{(equality 2)}}{=}& \sum_{\bm{q}_{\leq i}\in\left(2^{[C]}\right)^i}{\left(\Pr\left[\bm{Q}_{\leq i}=\bm{q}_{\leq i}~|~\bm{G}_{<i}=\bm{\sigma}_{<i}\right]\cdot\Pr(\sigma_i)\right)} = \Pr(\sigma_i)
\end{align*}
\end{small}
\vspace{-1ex}
\hrule
\vspace{1ex}
\caption{}\label{fig-nodes-rnd-uniform}
\vspace{-3ex}
\end{figure*}
Consider an arbitrary slot $i$. Consider an arbitrary $\sigma_i\in\Omega$ and an $\bm{\sigma}_{<i}=(\sigma_1,\cdots,\sigma_{i-1})\in\Omega^{i-1}$. Let $\Omega^{*}=\{ \bm{q}_{\leq i}\in \Omega^{i}: \left(K(\bm{q}_{\leq i})\leq y_1R\right) \wedge (|q_i|\geq x_1C)\}$,
$\Omega^{**}=\{ \bm{q}_{\leq i}\in \Omega^{i}: \left(K(\bm{q}_{\leq i})\leq y_1R\right) \wedge (|q_i|< x_1C)\}$, and $\Omega^{***}=\{ \bm{q}_{\leq i}\in \Omega^{i}: K(\bm{q}_{\leq i})>y_1R \}$.

Let $\mathcal{G}$ be the correct distribution of nodes' behavior in any a slot. (I.e., the behavior are determined by, say a chunk in $\bm{T}_{low}$ or $\bm{T}_{high}$, directly.)
For simplicity, we use $\Pr(\sigma_i)\triangleq \Pr_{G_i\sim\mathcal{G}}[G_i=\sigma_i]$ to denote the probability that $G_i=\sigma_i$ when $G_i$ is sampled from $\mathcal{G}$.
In Figure \ref{fig-nodes-rnd-uniform}, we show $\Pr[G_i=\sigma_i~|~\bm{G}_{<i}=\bm{\sigma}_{<i}]=\Pr(\sigma_i)$ when the random bits used by Eve is fixed\footnote{More accurately,  $\Pr[G_i=\sigma_i ~|~\bm{G}_{<i}=\bm{\sigma}_{<i} \wedge \bm{F}_{\leq i}=\bm{f}_{\leq i} ]=\Pr(\sigma_i)$ for any $\bm{f}_{\leq i}$ used by Eve.}.  Notice, equality 1 in Figure \ref{fig-nodes-rnd-uniform} holds because the way we generate $G_i$ (i.e., nodes' behavior). On the other hand, equality 2 holds because $\Pr[\Psi_{q_i}(T_{hi}^{(K(\bm{q}_{\leq i}))})=\sigma_i]=\Pr[T_{hi}^{(K(\bm{q}_{\leq i}))}=\Psi^{-1}_{q_i}(\sigma_i)]=\Pr(\sigma_i)$ for $\bm{q}_{\leq i}\in \Omega^{*}$; and similarly, $\Pr[\Psi_{q_i}(T_{lo}^{(i-K(\bm{q}_{\leq i}))})=\sigma_i]=\Pr(\sigma_i)$ for $\bm{q}_{\leq i}\in \Omega^{**}$, and $\Pr[\Psi_{q_i}(T_{lo}^{(i-y_1R)})=\sigma_i]=\Pr(\sigma_i)$ for $\bm{q}_{\leq i}\in \Omega^{***}$.

Finally, due to the chain rule, we know for any $\bm{\sigma}\in\Omega^R$:

$\Pr[\bm{G}=\bm{\sigma}]=\prod_{i=1}^{R}{\Pr[G_i=\sigma_i~|~\bm{G}_{<i}=\bm{\sigma}_{<i}]}=\Pr_{\bm{G}\sim \bm{\mathcal{G}}}[\bm{G}=\bm{\sigma}]$

This completes the proof of the claim.
\end{proof}
\begin{proof}[\underline{Proof of Lemma \ref{lemma-multicastadp-fast-bcst}}]
 Throughout the proof, we use $R$ to denote the length of the $i$\textsuperscript{th} epoch (i.e., $R_i$), and use $p$ to denote the working probability of a node in the $i$\textsuperscript{th} epoch (i.e., $p_i$).

Define $\mathcal{E}_{X}$ (resp., $\mathcal{E}_{X'}$ or $\mathcal{E}_{Y}$) be the event that some node is still uninformed by the end of process $\beta$ (resp., $\beta'$ or $\gamma$). Recall $\mathcal{E}_{\geq}$ is related to $\bm{Q}$ in process $\beta$. The following two claims capture the relationship between them:
\begin{claim*} 
$\mathcal{E}_{X}$ implies $\mathcal{E}_{X'}$.
\end{claim*}

\begin{claimproof}
	
Notice that permuting nodes' behavior and unjammed channels together will not change nodes' feedback.
So for each slot $i$, $(\pi(Q'_i)= Q_i)\wedge (M'_{i-1}= M_{i-1})$ will imply $M'_i=M_i$, where $Q_i$ (resp., $Q'_i$) denotes the set of unjammed channels in $i^{th}$ slot, and $M_i$ (resp., $M'_i$) denotes the set of informed nodes by the end of slot $i$, in process $\beta$ (resp., $\beta'$).

Moreover, $(\pi(Q'_i)\subseteq Q_i)\wedge (M'_{i-1}\subseteq M_{i-1})$ implies $M'_i\subseteq M_i$. We prove via its contrapositive: if exists node $u$ that $u\in M'_i$ and $u\notin M_i$, then either $\pi(Q'_i)\not\subseteq Q_i$ or $M'_{i-1}\not\subseteq M_{i-1}$. Since $u\in M'_i$, then in $i^{th}$ slot of process $\beta'$,  $u$ listens on the channel $h\in Q'_i$, and there is a unique node $v\in M'_{i-1}$ that broadcast on $h$. If $v\notin M_{i-1}$ then the proof is complete. Otherwise, due to the way we generate the behavior in the two processes, $u$ will be successfully informed by $v$ on channel $\pi(h)$ in $\beta$, unless $\pi(h)\notin Q_i$, i.e. the channel is jammed in $\beta$.
\end{claimproof}

\begin{claim*} 
$\Pr[(\mathcal{E}_{X'}\wedge \mathcal{E}_{\geq})]\leq  \Pr[(\mathcal{E}_{Y}\wedge \mathcal{E}_{\geq})]$.	
\end{claim*}
\begin{claimproof}
If $\mathcal{E}_{\geq}$ holds, both in process $\beta'$ and $\gamma$, Eve leaves exactly $y_1R$ slots in which channel $\{1,2,\cdots,x_1C\}$ are unjammed, and jams all channels in $(1-y_1)R$ slots. We can simply ignore the slots in which all channels are jammed, since the set of informed nodes will not change in these slots. Then in each left slot, behavior and further feedback of nodes in the two processes are same correspondingly.
\end{claimproof}

Therefore, $\Pr[\mathcal{E}_{X} \wedge \mathcal{E}_{\geq}]\leq  \Pr[\mathcal{E}_{X'} \wedge \mathcal{E}_{\geq}]\leq \Pr[\mathcal{E}_{Y} \wedge \mathcal{E}_{\geq}]\leq \Pr[\mathcal{E}_{Y}]$. 
Now whether event $\mathcal{E}_{Y}$ occurs only depends on nodes' randomness, since the jamming vector is fixed as $\bm{\hat{q}}$: in the first $y_1R=2b\cdot i\cdot 4^i\cdot\lg^2{n}$ slots, Eve leaves $\{1,2,\cdots,x_1C\}$ unjammed; and all channels are jammed in the remaining slots. Here, $b$ is some sufficiently large constant. Let $\mathcal{R}_1$ denote the collection of the first half of these $y_1R$ slots, and let $\mathcal{R}_2$ denote the second half. To show effectiveness of the epidemic broadcast scheme, we rely on the following two claims: 
\begin{claim*} [$\mathcal{R}_1$-claim]
With probability at least $1-\exp(-\Theta(i\cdot\lg{n}))$, at the end of $\mathcal{R}_1$, the number of informed nodes is at least $n/2$.
\end{claim*}
\begin{claimproof} %[Proof of $\mathcal{R}_1$-claim]
Divide $\mathcal{R}_1$ into $\lg{n}$ \emph{segments}, each containing $b\cdot i\cdot 4^i\cdot\lg{n}$ slots. To prove the first claim, we show after each segment, the number of informed nodes will double, with sufficiently high probability.

Fix a segment in $\mathcal{R}_1$, fix a time slot in this segment, let $t\in[1,n/2]$ denote the number of informed nodes at the beginning of this slot. Consider a node $u$ that is informed at the beginning of the segment. if $u$ wants to inform a previously uninformed node in this slot, the following conditions must hold: (a) $u$ broadcasts; (b) all other $t-1$ informed nodes do not broadcast on the channel chosen by $u$; (c) all uninformed nodes do not broadcast on the channel chosen by $u$ and at least one uninformed node listens on the channel chosen by $u$; and (d) the channel chosen by $u$ is not jammed by Eve. Therefore, the probability that $u$ informs an uninformed node in this slot is at least:
\begin{align*}
&\quad p\cdot \left(1-\frac{p}{C}\right)^{t-1}\cdot \left[\left(1-\frac{p}{C}\right)^{n-t}-\left(1-\frac{2p}{C}\right)^{n-t}\right] \cdot x_1\\
\geq &\quad p\cdot \left(1-\frac{p}{C}\right)^{n-1}\cdot \left[1-\left(\frac{C-2p}{C-p}\right)^{n/2}\right] \cdot x_1\\
\geq &\quad p\cdot e^{-2}\cdot \left(1-e^{-np/(2C)}\right )\cdot x_1 && \text{since } p\leq C/n\\
\geq &\quad \frac{x_1p^2n}{4e^2C}
\end{align*}
Thus after one segment, the probability that $u$ does not inform an uninformed node is at most $\left(1-(x_1p^2n)/(4e^2C)\right)^{b\cdot i\cdot 4^i\cdot\lg{n}}\leq\exp(-\Theta(i\cdot\lg{n}))$. Take a union bound over all the $O(n)$ nodes that are informed at the beginning of this segment, we know the number of informed nodes will at least double with probability at least $\exp(-\Theta(i\cdot\lg{n}))$ by the end of the segment. Take another union bound over the $\lg{n}$ segments, the claim holds.
\end{claimproof}

\begin{claim*} [$\mathcal{R}_2$-claim]
Assume by the end of $\mathcal{R}_1$ indeed there are at least $n/2$ informed nodes, then by the end of the epoch, all nodes will be informed, with probability at least $1-\exp(-\Theta(i\cdot\lg^2{n}))$.
\end{claim*}
\begin{claimproof} %[Proof of $\mathcal{R}_2$-claim]
Fix a node $u$ that is still uninformed at the beginning of $\mathcal{R}_2$. Consider a slot in $\mathcal{R}_2$, assume there are $t\geq n/2$ informed nodes at the beginning of this slot. For $u$ to become informed in this slot, the following conditions must hold: (a) $u$ decides to listen; (b) some informed node $v$ broadcasts on the channel chosen by $u$; (c) all other $n-2$ nodes do not broadcast on the channel chosen by $u$; and (d) the channel chosen by $u$ is not jammed by Eve. Therefore, the probability that $u$ will be informed in this slot is at least:

\vspace{-1ex}$$p\cdot t\cdot \frac{p}{C} \cdot \left(1-\frac{p}{C}\right)^{n-2} \cdot x_1 \geq p\cdot \frac{n}{2}\cdot\frac{p}{C} \cdot \left(1-\frac{p}{C}\right)^{n}\cdot x_1 \geq \frac{x_1p^2n}{2e^2C}$$

\noindent where the last inequality is due to $p\leq C/n$.

Therefore, by the end of $\mathcal{R}_2$, the probability that $u$ is uninformed is at most $\big(1-(x_1p^2n)/(2e^2C)\big)^{|\mathcal{R}_2|}\leq\exp(-\Theta(i\cdot\lg^2{n}))$. Take a union bound over the $O(n)$ nodes that are uninformed at the beginning of $\mathcal{R}_2$ immediately leads to this claim.
\end{claimproof}
\end{proof}

\begin{proof}[\underline{Proof of Claim $\Pr[X\leq t]\leq\Pr[X'\leq t]$ in Lemma~\ref{lemma-multicastadp-fast-halt}}]
\begin{figure*}[!t]
\hrule
\vspace{-1ex}
\begin{small}
\begin{align*}
X_i =& \sum_{l\in q}\Bigg(\mathbb{I}\left[\left(\Psi_{q}\left(T^{(j)}\right)\right)^{ch}_{u}=l\right]\cdot \mathbb{I}\left[\left(\Psi_{q}\left(T^{(j)}\right)\right)^{act}_{u}=\text{listen}\right]\\
& \quad \quad \cdot \left(\prod_{v\in V\setminus \{u\}}\left(1-\mathbb{I}\left[\left(\Psi_{q}\left(T^{(j)}\right)\right)^{ch}_{v}=l\right]\cdot \mathbb{I}\left[\left(\Psi_{q}\left(T^{(j)}\right)\right)^{act}_{v}=\text{send}\right]\right)\right)\Bigg) \\
=& \sum_{l\in[|q|]}\Bigg(\mathbb{I}\left[\left(T^{(j)}\right)^{ch}_{u}=l\right]\cdot \mathbb{I}\left[\left(T^{(j)}\right)^{act}_{u}=\text{listen}\right]\\
& \quad \quad \cdot\left(\prod_{v\in V\setminus \{u\}}\left(1-\mathbb{I}\left[\left(T^{(j)}\right)^{ch}_{v}=l\right]\cdot \mathbb{I}\left[\left(T^{(j)}\right)^{act}_{v}=\text{send}\right]\right)\right)\Bigg) \\
\geq& \sum_{l\in[x_2C]}\Bigg(\mathbb{I}\left[\left(T^{(j)}\right)^{ch}_{u}=l\right]\cdot \mathbb{I}\left[\left(T^{(j)}\right)^{act}_{u}=\text{listen}\right]\\
& \quad \quad\cdot\left(\prod_{v\in V\setminus \{u\}}\left(1-\mathbb{I}\left[\left(T^{(j)}\right)^{ch}_{v}=l\right]\cdot \mathbb{I}\left[\left(T^{(j)}\right)^{act}_{v}=\text{send}\right]\right)\right)\Bigg) \\
\geq & X'_i
\end{align*}
\end{small}
\vspace{-3.5ex}
\hrule
\vspace{1ex}
\caption{}\label{fig-xi-leq-yi}
\vspace{-3ex}
\end{figure*}

To prove the claim, showing $X\geq X'$ is sufficient, which in turn is an immediate corollary if $X_i\geq X'_i$. Thus, we now prove $X_i\geq X'_i$.

Assume the jamming result of $i^{th}$ slot is $q$. We only consider the case $|q|\geq x_2C$ and $K(\bm{q}_{\leq i})\leq y_2R$ here, as the other two cases are very similar. Assume randomness of nodes in slot $i$ comes from chunk $T^{(j)}$ in $\bm{T}_{high}$. The proof for $X_i\geq X'_i$ is shown in Figure \ref{fig-xi-leq-yi}, where the equality in the second line is due to the definition of $\Psi_q$ (or more precisely, the definition of the permutation $\pi_q$).
\end{proof}

\begin{proof}[Proof of Lemma \ref{lemma-multicastadp-halt-imply-informed}]
Consider two complement cases: either Eve jams a lot in the epoch, or she does not. If jamming is not strong, Lemma \ref{lemma-multicastadp-fast-bcst} implies no node remains uninformed. Otherwise, then $u$ should not hear a lot of silent slots and will not halt.
Let $R$ be the length of the epoch, and $p$ be nodes' working probability. 
Let $\mathcal{E}_1$ be event (a) in lemma statement; (i.e., ``node $u$ hears silence at least $Rp/2$ times'');
and $\mathcal{E}_2$ be event (b) in lemma statement. (i.e., ``some node never hears the message'').
Let $\mathcal{E}_3$ be the event that $\mathcal{E}^{\geq 0.1}(\geq 0.1)$ happens in the epoch. 
Bounding $\Pr[\mathcal{E}_1\wedge \overline{\mathcal{E}_3}]\leq \exp(-\Theta(i\cdot\lg{n}))$ is similar to that of Lemma \ref{lemma-multicastadp-fast-halt}'s, but with several modifications. Let $x_1=0.1$, and $y_1=0.1$. As shown in Section \ref{sec-notation}, $\overline{\mathcal{E}_3}=\mathcal{E}^{>0.9}(<0.1)$, while implies $\mathcal{E}^{\geq (1-y_1)}(\leq x_1)$.

First in process $\beta$, the definition of $K$ is changed: let $K(\bm{Q}_{\leq i})=\sum_{j=1}^{i}{\mathbb{I}[|Q_j|\leq x_1C]}$ count heavily jammed slots (i.e., $|Q_j|\leq x_1C$) among the first $i$ slots. Also, $G_i$ is defined as:
\begin{displaymath}
	G_i=
\left\{\begin{array}{ll}
	\Psi_{Q_i}\left(T_{lo}^{\left(K\left(\bm{Q}_{\leq i}\right)\right)}\right), & |Q_i|\leq x_1C  \textrm{ and } K(\bm{Q}_{\leq i})\leq (1-y_1)R \\
	\Psi_{Q_i}\left(T_{hi}^{\left(i-K\left(\bm{Q}_{\leq i}\right)\right)}\right), &|Q_i|> x_1C  \textrm{ and } K(\bm{Q}_{\leq i})\leq (1-y_1)R \\
	\Psi_{Q_i}\left(T_{hi}^{\left(i-(1-y_1)R\right)}\right), &\textrm{ otherwise} \\
\end{array}
\right.
\end{displaymath}

We continue to introduce process $\beta'$. For each slot $i$, if in $\beta$ nodes use $\Psi_{Q_i}(T_{hi}^{(j)})$ (resp., $\Psi_{Q_i}(T_{lo}^{(j)})$) to determine their behavior, then in $\beta'$ nodes directly use $T_{hi}^{(j)}$ (resp., $T_{lo}^{(j)}$), and Carlo leaves all channels unjammed (resp., leaves $\{1,2,\cdots,x_1C\}$ unjammed).

Finally, in $\gamma$, adversary jams obliviously using the following strategy: in the first $(1-y_1)R$ slots, channels $\{1,2,\cdots,x_1C\}$ are unjammed; and in the remaining $y_1R$ slots, all channels are unjammed. Besides, in the $i^{th}$ slot, node directly use chunk $T^{(i)}_{lo}$ if $i\leq (1-y_1)R$, and otherwise $T^{(i-(1-y_1)R)}_{hi}$ to compute their behavior.

The definitions of $X_i,X'_i,Y_i$ and $X,X',Y$ remain unchanged. 
Similarly to the proof of Lemma \ref{lemma-multicastadp-fast-halt}, we can derive the following result: for any integer $t\geq 0$, $\Pr[X\geq t]\leq\Pr[X'\geq t]$, and $\Pr[(X'\geq t)\wedge \mathcal{E}^{\geq 1-y_1}(\leq x_1)]\leq \Pr[(Y\geq t)\wedge \mathcal{E}^{\geq 1-y_1}(\leq x_1)]$.
Therefore, 
\begin{align*}
	\Pr[\mathcal{E}_1\wedge \overline{\mathcal{E}_3}]&\leq \Pr[X\geq Rp/2 \wedge \mathcal{E}^{\geq 1-y_1}(\leq x_1)]\\
	&\leq \Pr[Y\geq Rp/2 \wedge \mathcal{E}^{\geq 1-y_1}(\leq x_1)]\\
	&\leq \Pr[Y\geq Rp/2]\\
	&\leq\exp(-\Theta(Rp))=\exp(-\Theta(i\cdot\lg{n}))
\end{align*}
The last inequality is due to a Chernoff bound and $\mathbb{E}[Y]=((1-y_1)R\cdot x_1+y_1R\cdot 1)\leq 0.19Rp(1-p/C)^{n-1}\leq 0.19Rp$.

Besides, by Lemma \ref{lemma-multicastadp-fast-bcst}, $\Pr[\mathcal{E}_2 \wedge \mathcal{E}_3 ]\leq\exp(-\Theta(i\cdot\lg{n}))$. As a result, $\Pr[\mathcal{E}_1\wedge\mathcal{E}_2]=\Pr[\mathcal{E}_1\wedge\mathcal{E}_2\wedge\mathcal{E}_3]+\Pr[\mathcal{E}_1\wedge\mathcal{E}_2\wedge\overline{\mathcal{E}_3}]\leq\Pr[\mathcal{E}_1\wedge \overline{\mathcal{E}_3}]+\Pr[\mathcal{E}_2 \wedge \mathcal{E}_3 ]\leq \exp(-\Theta(i\cdot\lg{n}))$. 
\end{proof}

\begin{proof}[\underline{Proof of Theorem \ref{thm-multicastadv}}]
Fix a node $u$, we first compute how long $u$ remains active. Let $L$ be the total runtime of $u$. Define constants $\alpha=10$ and $\beta=1/((1-x_2)(1-y_2))$. Recall $I_b$ denotes the number of the first epoch, and we set $R_{I_b-1}=0$ for the ease of presentation. As a result, for any $x\in\mathbb{R}^+$, there is a unique integer $i\geq I_b$ satisfying $R_{i-1}\leq x<R_i$. Let $sr_i=\sum_{k=\hat{I_b}+1}^{i}R_k$ be the total number of slots from epoch $\hat{I_b}+1$ to epoch $i$.

We now bound the probability that $L>\alpha\beta T/C$. Specifically,
\begin{align*}
\Pr(L>\alpha \beta T/C) \leq & \sum_{i=I_{b}}^{\infty}\Pr\Big( (R_{i-1}\leq \beta T/C<R_{i})\wedge (L>\alpha\beta T/C) \Big)\\
\leq & \sum_{i=I_{b}}^{\infty}\Pr\Big(  (\beta T/C<R_{i})\wedge (L>sr_{i}) \Big)\\
\leq & \sum_{i=I_{b}}^{\infty}\Pr\Big(  (\beta T/C<R_{i})\wedge (L>sr_{i}) \ | \  (L>sr_{i-1}) \Big)\\
\leq & \sum_{i=I_{b}}^{\infty}\exp\big(-\Theta(i\cdot\lg{n})\big)=n^{-\Omega(1)}
\end{align*}
In above, the second inequality is due to the fact that the length of epochs increases geometrically; the last inequality is due to Lemma \ref{lemma-multicastadp-fast-halt}, and the fact that Eve spends less than $(1-y_2)(1-x_2)R_{i}C$ in epoch $i$ implies event $\mathcal{E}^{\geq y_2}(\geq x_2)$ occurs in epoch $i$.

Take a union bound over all $n$ nodes, we know when $T=\Omega(C)$ all nodes will halt within $O(T/C)$ slots, with high probability.

Next, we analyze the cost of nodes. Again fix a node $u$. For any epoch $i$ in which $u$ is alive, its expected cost in this epoch is $2R_ip_i$. Apply a Chernoff bound, we know $u$'s cost will be at most $3R_ip_i$, with probability at least $1-\exp(-\Theta(R_ip_i))\geq 1-\exp(-\Theta(i\cdot \lg^2 n))$. Let $F_i$ denote $u$'s actual cost during epoch $i$, and let $F$ denote the actual total cost. Let $\mathcal{E}$ be the event that ``for all $k\geq I_b$, $F_k\leq 3R_kp_k$'', then $\Pr(\overline{\mathcal{E}})<\sum_{k=I_b}^{\infty}\exp(-\Theta(R_kp_k))=n^{-\Omega(1)}$. Finally, define $\gamma=169a\lg^2{n}$. We can now bound $\Pr(F>\sqrt{\gamma\beta\cdot\lg{T}\cdot(T/n)})$:
\begin{align*}
& \Pr\Big(F>\sqrt{\gamma\beta\cdot\lg{T}\cdot(T/n)}\Big) \\
\leq & \Pr\Big(\overline{\mathcal{E}}\Big) +
\sum_{i=I_{b}}^{\infty}\Pr\Big(\left(L>sr_i\right)\wedge(R_{i-1}\leq\beta T/C<R_{i})\Big) \\
& + \sum_{i=I_{b}}^{\infty}\Pr\Big(\mathcal{E}\wedge\left(L\leq sr_i\right)\wedge\left(R_{i-1}\leq\beta T/C<R_{i}\right)\wedge\left(F^2>\gamma\beta\cdot\lg{T}\cdot (T/n)\right)\Big) \\
\leq & n^{-\Omega(1)} + n^{-\Omega(1)} + \sum_{i=I_{b}}^{\infty}\Pr\Big(\mathcal{E}\wedge\left(L\leq sr_i\right)\wedge\left(F^2>\gamma\beta\cdot i\cdot\frac{R_{i-1}C}{\beta n}\right)\Big) \\
\leq & n^{-\Omega(1)} + n^{-\Omega(1)} + \sum_{i=I_{b}}^{\infty}\Pr\left(\mathcal{E}\wedge\left(F\leq \sum_{k=I_{b}}^{i}F_k\right)\wedge\left(F>13R_{i-1}p_{i-1}\right)\right) \\
= & n^{-\Omega(1)} + 0 = n^{-\Omega(1)}
\end{align*}
Notice, the last inequality is due to: (1) ``$L\leq sr_i$'' means $u$ halts by the end of epoch $i$, hence the total cost of $u$ is its cost up to the end of epoch $i$; and (2) ``$F^2>\gamma\beta\cdot i\cdot(R_{i-1}C)/(\beta n)$'' together with the definition $\gamma=169a\lg^2{n}$ implies $F>13R_{i-1}p_{i-1}$. In addition, the equality in the last line is due to: if $\mathcal{E}$ happens, then $F\leq\sum_{k=I_{b}}^{i}F_k\leq\sum_{k=I_{b}}^{i}{3R_kp_k}\leq 13R_{i-1}p_{i-1}$.

Take a union bound over all $n$ nodes, we know when $T=\Omega(C)$ the cost of each node is $O(\sqrt{T/n}\cdot\sqrt{\lg{T}}\cdot\lg{n})$, with high probability.

We continues to show w.h.p.\ each node must have been informed when it halts:
\begin{align*}
& \Pr(\textrm{some node halts while uninformed})\\
\leq\quad & \sum_u\Pr(\textrm{some node halts while uninformed}\wedge u\textrm{ is the first node that halts})\\
\leq\quad & \sum_u\sum_{i=I_b}^{\infty}\Pr(\textrm{some node is uninformed when }u\textrm{ halts firstly in epoch }i)\\
\leq\quad & n\cdot\sum_{i=I_b}^{\infty}\exp(-\Theta(i\cdot\lg{n}))=n^{-\Omega(1)}
\end{align*}
Notice, the last inequality is due to Lemma \ref{lemma-multicastadp-halt-imply-informed}.

Finally, we note that when $T=o(C)$, all nodes will halt by the end of the first epoch, with high probability. In such case, the time and energy cost for each node is $O(R_{I_b})$ and $O(R_{I_b}p_{I_b})$, respectively. This results in the $\tau_{time}$ and $\tau_{cost}$ terms in the theorem statement.
\end{proof}

\section{Omitted Parts in the Analysis of \multicastadvadp}\label{sec-app-multicastadvadp-proof}

This section contains two parts: (a) some auxiliary material (such as additional notations and lemmas); and (b) proofs that are omitted in the main body of the paper.

\subsection{Auxiliary Material}\label{subsec-app-multicastadvadp-aux-material}

Consider a step in an $(i,j)$-phase, recall $p_u$ denotes $u$'s working probability in this step. We introduce the following definitions to facilitate presentation and analysis: (a) $P_V=(\sum_{u\in V}{p_u})/C$, and $P_M=(\sum_{u\in M}{p_u})/C$. (b) Consider a slot in the step in which $u$ listens on a channel not jammed by Eve, define $p_u^c$ be the probability that $u$ hears silence, and $p_u^m$ be the probability that $u$ hears message $m$. Clearly, $p_u^c=\prod_{v\in V\setminus\{u\}}(1-p_v/C)$, and $p_u^m=\sum_{v\in M\setminus\{u\}}((p_v/C)\cdot \prod_{w\in V\setminus\{u,v\}}(1-(p_w/C)))$.

\smallskip The following simple fact will be helpful in various places.

\begin{fact}\label{fact-pcu-pmu-bounds}
$e^{-2P_V}\leq p_u^c\leq e^{-P_V}/(1-p_u/C)$ and $(P_M-p_u/C)e^{-2P_V}\leq p_u^m\leq P_M$.
\end{fact}

\begin{proof}
On the one hand:

$$e^{-2P_V}\leq\prod_{v\in V}\left(1-\frac{p_v}{C}\right)\leq p_u^c\leq\frac{\prod_{v\in V}\left(1-\frac{p_v}{C}\right)}{1-p_u/C}\leq\frac{e^{-P_V}}{1-p_u/C}$$

On the other hand:

$$\left(P_M-\frac{p_u}{C}\right)e^{-2P_V}\leq \sum_{v\in M\setminus\{u\}}\left(\frac{p_v}{C}\cdot\prod_{w\in V}\left(1-\frac{p_w}{C}\right)\right)\leq p^m_u\leq P_M$$
\end{proof}

The following corollary, which states the working probability of a node cannot increase by a factor more than $\sqrt{2}$ after a phase, is a direct consequence of part (1) of Claim \ref{claim-multicastadvadp-pu-pv-bounded}:

\begin{corollary}\label{cor-multicastadvadp-pu-inc-bounded}
Consider an $(i,j)$-phase and an active node $u$, $p^{i,j+1}_u\leq p_u^{i,j}\cdot\sqrt{2}$ with probability at least $1-\exp(-\Theta(i^3C))$.
\end{corollary}

The following lemma shows, for two nodes $u$ and $v$, $\eta_u$ and $\eta_v$ cannot differ much after a phase. Its analysis uses a coupling argument which is similar to that of Claim \ref{claim-multicastadvadp-pu-pv-bounded}.

\begin{lemma}\label{lemma-multicastadvadp-eta-diff-bounded}
Consider an $(i,j)$-phase where $i>\lg n$. The probability that both (a) some node $u$ has $\eta_u>2.5$, and (b) some node $v$ has $\eta_v<2.4$, is at most $\exp(-\Theta(i^3C))$.
\end{lemma}

\begin{proof}
Let $R_i$ be the number of slots in each step of the phase. Similar to the proof of Claim \ref{claim-multicastadvadp-pu-pv-bounded}, we use a coupling argument. Specifically, for each step $*$ in $\{1,2,3\}$, nodes' behavior is determined by $R_i$ chunks from lists $(\bm{T}_{step*,1},\bm{T}_{step*,2},\cdots,\bm{T}_{step*,C})$. Let $\vec{\mathcal{Q}}=\mathcal{Q}\times\mathcal{Q}\times\mathcal{Q}$, and $\vec{\mathcal{Z}}=\mathcal{Z}\times\mathcal{Z}\times\mathcal{Z}$. The function $\vec{K}:\vec{\mathcal{Q}}\to \vec{\mathcal{Z}}$ is defined using function $K$: $\vec{K}(\vec{\bm{Q}})=\vec{K}(\langle \bm{Q}^{step1},\bm{Q}^{step2},\bm{Q}^{step3} \rangle)=\langle K(\bm{Q}^{step1}),K(\bm{Q}^{step2}),K(\bm{Q}^{step3}) \rangle$.

%\langle \langle K_1(\bm{Q}^{step1}),\cdots,K_C(\bm{Q}^{step1})\rangle,\langle K_1(\bm{Q}^{step2}),\cdots,K_C(\bm{Q}^{step2})\rangle,\langle K_1(\bm{Q}^{step3}),\cdots,K_C(\bm{Q}^{step3})\rangle \rangle

For each step $*$ in $\{1,2,3\}$, define $X^{step*}_{u},Y^{step*}_{u}$ by borrowing definition of $X_{u},Y_{u}$ from the proof of Claim \ref{claim-multicastadvadp-pu-pv-bounded} respectively, so $X_u(\bm{Q}^{step*})$ still equals $Y_u^{step*}(K(\bm{Q}^{step*}))$. %Define $Y_u^{step*}$ be the sum of $R_i$ indicator random variables so that $Y_u^{step*}$ equals the number of silent slots $u$ heard in step $*$ (i.e., $X_u^{step*}$).
Notice that for each step $*$ in $\{1,2,3\}$, the ratio $\mathbb{E}[X_u^{step*}]/\Delta_u^{step*}$ are identical for all nodes. Let $\alpha_{1}=\alpha_{2}=\prod_{w\in V}(1-p_w/C)$, and $\alpha_{3}=\prod_{w\in V}(1-p_{step3}/C)$. Let $\vec{\mathcal{Z}}'=\{\vec{\bm{z}}=\langle\bm{z}^{step1},\bm{z}^{step2},\bm{z}^{step3}\rangle:\alpha_{1}\cdot L(\bm{z}^{step1})+\alpha_{2}\cdot L(\bm{z}^{step2})+\alpha_{3}\cdot L(\bm{z}^{step3})\leq 2.45R_iC\}$. At this point, we conclude:
\begin{align*}
& \Pr\left(\exists u,v~:~(\eta_u>2.5)\wedge(\eta_v<2.4)\right)\\
\leq & \sum_{u\in V}\sum_{v\in V}
\left(
\sum_{\vec{\bm{z}}\in\vec{\mathcal{Z}}'}\Pr\left(\sum_{*=1}^{3}\frac{Y^{step*}_u(\bm{z}^{step*})}{\Delta^{step*}_u}>2.5\right)+
\sum_{\vec{\bm{z}}\in\vec{\mathcal{Z}}\setminus\vec{\mathcal{Z}}'}\Pr\left(\sum_{*=1}^{3}\frac{Y^{step*}_v(\bm{z}^{step*})}{\Delta^{step*}_v}<2.4\right)
\right)\\
\leq & n^2\cdot(2R)^{3C}\cdot\exp(-\Theta(i^3C))=\exp(-\Theta(i^3C))
\end{align*}

This completes the proof of the lemma.
\end{proof}

The following lemma shows $P_V$ has a natural upper bound. This is because when $P_V$ is large, nodes cannot hear too many silent slots due to contention among themselves.

\begin{lemma}\label{lemma-multicastadvadp-Pv-bounded}
Consider a super-epoch $i>\lg{n}$. With probability at least $1-\exp(-\Theta(i^3C))$, $P_V=(\sum_{u\in V}{p_u})/C\leq 1/2$ holds at any phase in the super-epoch.
\end{lemma}

\begin{proof}
It is easy to verify $P^{i,0}_V\leq n/2^i\leq 1/2$. Recall Corollary \ref{cor-multicastadvadp-pu-inc-bounded} shows after each phase each node will increase its working probability by a factor of at most $\sqrt{2}$. Hence, due to the fact that $0.35\cdot \sqrt{2}<1/2$, to prove the lemma, it suffices to show when $P_V\geq 0.35$, $P_V$ will not increase in subsequent phases, with probability at least $1-\exp(-\Theta(i^3C))$.

Fix a node $u$ and a phase in which $P_V\geq 0.35$ at the beginning. To make $\eta_u$ as large as possible, assume Eve does no jamming during the phase. Hence, $\mathbb{E}[N_u^{c,step1}]=\mathbb{E}[N_u^{c,step2}]=R_ip_u\cdot p_u^c\leq R_ip_u\cdot e^{-P_V}/(1-p_u/C)\leq e^{-P_V}\Delta^{step1}_u\leq 0.71\Delta^{step1}_u=0.71\Delta^{step2}_u$, and $\mathbb{E}[N_u^{c,step3}]\leq \Delta^{step3}_u$. By a Chernoff bound, the probability that both $N^{c,step1}_u/\Delta^{step1}_u$ and $N^{c,step2}_u/\Delta^{step2}_u$ does not exceed $0.725$, and $N^{c,step3}_u/\Delta^{step3}_u$ does not exceed $1.05$ (these three events lead to $\eta_u\leq 2.5$) is at least $1-2\cdot \exp(-\Theta(\Delta^{step1}_u))-\exp(-\Theta(\Delta^{step3}_u))\geq 1-\exp(-\Theta(i^3C))$. Finally, take a union bound over $O(i)$ phases and another union bound over $O(n)$ nodes, the lemma is proved.
\end{proof}

The following lemma shows nodes cannot become \texttt{helper} in early super-epochs, as sending probabilities in these super-epochs are too high for nodes to hear enough silent slots.

\begin{lemma}\label{lemma-multicastadvadp-no-helper-in-early-epoch}
For each node $u$, the probability that it becomes \texttt{helper} in some phase where all nodes are active within some super-epoch $i\leq\lg(nC)$, is at most $n^{-\Omega(1)}$.
\end{lemma}

\begin{proof}
Fix a node $u$ and a phase in super-epoch $i$, we intend to upper bound $N^{c,step3}_u$. Therefore, assume Eve does no jamming during step three.

If $i\leq(\lg{n})/2$, then in a slot in which $u$ listens, the probability that $u$ hears silence is $(1-p_{step3}/C)^{n-1}\leq\exp(-(n-1)C/2^i)\leq \exp(-0.9C\sqrt{n})=n^{-\Omega(1)}$. Take a union bound over $\sum_{i=1}^{(\lg n)/2}bi\cdot a2^{i}i^3=O(\sqrt{n}\lg^4{n})$ slots, with high probability $u$ will never hear a silent slot during step three by the end of super-epoch $(\lg{n})/2$. Together with part (1) of Claim \ref{claim-multicastadvadp-pu-pv-bounded}, we know with high probability $\eta_u\leq 2.4$ always holds. As a result, $u$ will not be \texttt{helper} by the end of super-epoch $(\lg{n})/2$, with high probability.

Next, let us assume $(\lg{n})/2<i\leq\lg{(nC)}$. We know $\mathbb{E}[N^{c,step3}_u]= \Delta_u^{step3}\cdot(1-p_{step3}/C)^{n}\leq\Delta_u^{step3}\cdot e^{-np_{step3}/C}\leq\Delta_u^{step3}/e$ since $p_{step3}=C^2/2^i\geq C/n$. Thus by a Chernoff bound, with probability at least $1-\exp(-\Theta(i^3C^2))$, $N^{c,step3}_u$ will not reach $0.4\Delta_u^{step3}$. Now, together with part (1) of Claim \ref{claim-multicastadvadp-pu-pv-bounded}, and apply a union bound over all phases for super-epochs $(\lg{n})/2$ to $\lg(nC)$, the probability that $u$ becomes \texttt{helper} in these super-epochs is at most $\sum_{i=(\lg{n})/2}^{\lg(nC)}{bi\cdot\exp(-\Theta(i^3C))}=n^{-\Omega(1)}$.
\end{proof}

The following lemma gives an upper bound on the estimates of $n$ when nodes become \texttt{helper}. It also suggests when nodes' working probabilities in step two are too small, they will not become \texttt{helper}, as the number of messages heard is not enough.

\begin{lemma}\label{lemma-multicastadvadp-estimate-upper-bound}
For each node $u$, the probability that it becomes \texttt{helper} with $n_u>4n$ in some super-epoch $i>\lg n$ is at most $n^{-\Omega(1)}$.
\end{lemma}

\begin{proof}
Due to Corollary \ref{cor-multicastadvadp-pu-inc-bounded}, we know starting from super-epoch $\lg{n}+1$, for any two nodes $w$ and $v$, $1/2\leq p_w/p_v\leq 2$ holds throughout the entire execution, with probability at least $1-n^{-\Omega(1)}$. Assume $1/2\leq p_w/p_v\leq 2$ indeed always holds. Fix a phase in a super-epoch $i$. If $u$ becomes \texttt{helper} in this phase with $n_u>4n$, then its working probability in step two $p_u$ is less than $\sqrt{{C}/{(4n2^i)}}$. Now, let $N_u^{m,step2}$ denote the number of messages $u$ heard during step two, and assume Eve does no jamming during step two so as to maximize $N_u^{m,step2}$. We have $\mathbb{E}[N_u^{m,step2}]\leq R_ip_u\cdot p^u_m\leq R_ip_u\cdot P_M\leq ai^3/2$ since $P_M\leq P_V\leq 2np_u/C\leq \sqrt{{n}/{(2^iC)}}$. Apply a Chernoff bound, $u$ will hear message at least $ai^3$ times in step two with probability at most $\exp(-\Omega(i^3))$. Take a union bound over all super-epochs $i>\lg{n}$ and the phases within, we conclude $u$ becomes \texttt{helper} with $n_u>4n$ with probability at most $\sum_{i>\lg{n}}bi\cdot \exp(-\Theta(i^3))=n^{-\Omega(1)}$.
\end{proof}

\subsection{Omitted Proofs}\label{subsec-app-multicastadvadp-proof}

\begin{proof}[\underline{Proof of Claim \ref{claim-multicastadvadp-pu-pv-bounded}}]
We only need to prove part (3). Let $D_u:\mathcal{Z}\to\mathbb{R}$ be function $D_u(\bm{z})=\sqrt{\mathbb{E}[Y_u(\bm{z})]\cdot giC}$. Recall $\chi_u=\sqrt{{giC}/{(Rp_u)}}$, thus $D_u(\bm{z})/\Delta_u\leq\chi_u$. Let $\mathcal{Z}_2=\{\bm{z}\in\mathcal{Z}:L(\bm{z})\geq 0.05RC/\alpha\}$. As a result:
\begin{align*}
&\Pr\left((\lvert  X_u/\Delta_u-\mathbb{E}[X_u]/\Delta_u\rvert\geq\chi_u)\wedge(X_u/\Delta_u>0.1)\right)\\
\leq & \Pr\left((|X_u(\bm{Q})/\Delta_u-\alpha |\bm{Q}|/RC|>D_u(K(\bm{Q}))/\Delta_u)\wedge(X_u(\bm{Q})/\Delta_u>0.1)\right)\\
\leq & \sum_{\bm{z}\in\mathcal{Z}_2}\Pr(|Y_u(\bm{z})-\mathbb{E}[Y_u(\bm{z})]|>D_u(\bm{z})) + \sum_{\bm{z}\in\mathcal{Z}\setminus\mathcal{Z}_2}\Pr(Y_u(\bm{z}) >0.1\Delta_u )\\
\leq & \sum_{\bm{z}\in\mathcal{Z}_2}\Pr(|Y_u(\bm{z})-\mathbb{E}[Y_u(\bm{z})]|>D_u(\bm{z})) + \exp(-\Theta(i^3C))
\end{align*}

Let $\delta=D_u(\bm{z})/\mathbb{E}[Y_u(\bm{z})]$, then we have $\delta^2\cdot\mathbb{E}[Y_u(\bm{z})]=giC$, and $\delta<\sqrt{20giC/Rp_u}\leq\sqrt{20g/(ai^2)}<1$. Therefore, apply a Chernoff bound and we know:
$$\sum_{\bm{z}\in\mathcal{Z}_2}\Pr(|Y_u(\bm{z})-\mathbb{E}[Y_u(\bm{z})]|>D_u(\bm{z})) \leq (2R)^C\cdot 2\cdot\exp(-\Theta(giC))=\exp(-\Theta(iC))$$
implying $\Pr\left((\lvert  X_u/\Delta_u-\mathbb{E}[X_u]/\Delta_u\rvert\geq\chi_u)\wedge(X_u/\Delta_u>0.1)\right)\leq\exp(-\Theta(iC))$.

Similarly, $\Pr\left((\lvert  X_v/\Delta_v-\mathbb{E}[X_v]/\Delta_v\rvert\geq\chi_v)\wedge(X_v/\Delta_v>0.1)\right)$ $\leq\exp(-\Theta(iC))$.

Finally, notice that ``$|x-y|>z_1+z_2$'' implies ``$|x-w|>z_1$ or $|y-w|>z_2$'' for any $w,x,y,z_1,z_2$, we conclude:
\begin{align*}
&\Pr\left(
\begin{matrix}
(\lvert  X_u/\Delta_u-X_v/\Delta_v \rvert \geq \chi_u+\chi_v)\wedge\\
( X_u/\Delta_u>0.1)\wedge(X_v/\Delta_v>0.1)
\end{matrix}
\right)\\
\leq & \Pr\left((\lvert  X_u/\Delta_u-\mathbb{E}[X_u]/\Delta_u\rvert\geq\chi_u)\wedge(X_u/\Delta_u>0.1)\right) +\\
&\Pr\left((\lvert  X_v/\Delta_v-\mathbb{E}[X_v]/\Delta_v\rvert\geq\chi_v)\wedge(X_v/\Delta_v>0.1)\right)\\
\leq & \exp(-\Theta(iC))
\end{align*}

This completes the proof of the claim.
\end{proof}

\begin{proof}[\underline{Proof of Lemma \ref{lemma-multicastadvadp-pu-pv-bounded}}]
%Take a union bound over all $O(n^2)$ pairs of nodes, and another union bound over all $3bi$ steps, we know Claim \ref{claim-multicastadvadp-pu-pv-bounded} holds during the entire super-epoch $i$. Assume this is indeed the case.
To prove the lemma, we show for any phase $j$ in super-epoch $i$, for any two nodes $u$ and $v$, it is true that $p^{i,j+1}_u/p^{i,j+1}_v\leq p^{i,j}_u/p^{i,j}_v\cdot 2^{1/bi}$. For the ease of presentation, denote the length of each step as $R$, and denote the working probabilities (i.e., the sending/listening probabilities in step one and two) of the current phase and the next phase as $p$ and $p'$, respectively. Recall for any node $u$, $\eta_u=N^{step1,c}_u/\Delta^{step1}_u+N^{step2,c}_u/\Delta^{step2}_u+N^{step3,c}_u/\Delta^{step3}_u$, and $p'_u=p_u\cdot 2^{\max\{0,\eta_u-2.5\}}$.

If $\eta_u\leq 2.5$ and $\eta_v\leq 2.5$, then $p'_u/p'_v=p_u/p_v$ and we are done. So, without loss of generality, assume $\eta_u> 2.5$. Due to part (1) of Claim \ref{claim-multicastadvadp-pu-pv-bounded}, and the fact that $2.5-2\cdot 1>0.2$, $N^{c,step*}_u/\Delta^{step*}_u\geq 0.2$ holds for any step $*$ in $\{1,2,3\}$. Then, apply part (2) of Claim \ref{claim-multicastadvadp-pu-pv-bounded}, we have $N^{c,step*}_v/\Delta^{step*}_v\geq 0.1$ holds for any step $*$ in $\{1,2,3\}$. Lastly, apply part (3) of Claim \ref{claim-multicastadvadp-pu-pv-bounded}, we have $|N^{c,step*}_u/\Delta^{step*}_u-N^{c,step*}_v/\Delta^{step*}_v|\leq \sqrt{{giC}/{(Rp_u)}}+\sqrt{{giC}/{(Rp_v)}}$ for any step $*$ in $\{1,2\}$ and $|N^{c,step3}_u/\Delta^{step3}_u-N^{c,step3}_v/\Delta^{step3}_v|\leq 2\sqrt{{giC}/{(Rp_{step3})}}$. Now, we conclude:
$$\frac{p'_u}{p'_v}\leq\frac{p_u}{p_v}\cdot\frac{2^{\max\{0,\eta_u-2.5\}}}{2^{\max\{0,\eta_v-2.5\}}}\leq\frac{p_u}{p_v}\cdot 2^{\eta_u-\eta_v}\leq\frac{p_u}{p_v}\cdot 2^{1/bi}$$
where the last inequality is due to:
$$2\left(\sqrt{\frac{giC}{Rp_u}}+\sqrt{\frac{giC}{Rp_v}}\right)+2\sqrt{\frac{giC}{Rp_{step3}}} \leq  6\sqrt{\frac{giC}{ai^32^i\cdot(C/2^i)}}\leq\sqrt{\frac{36g}{ai^2}}\leq\frac{1}{bi}$$
when $g$ is a constant satisfying $a\geq 36gb^2$. This completes the proof of the lemma.
\end{proof}

\begin{proof}[\underline{Proof of Lemma \ref{lemma-multicastadvadp-pu-raise-and-msg}}]
Let $\mathcal{E}_{R}$ be the event that some node raises its working probability by the end of the phase. Let $\mathcal{E}_{M}$ be the event that some node hears the message less than  $ai^3$ times during step two. Let $\mathcal{E}_{un}$ be the event that some node is still uninformed by the end of step one. Moreover, let $\mathcal{E}_1$ (respectively, $\mathcal{E}_2$) be the event that $\mathcal{E}_{step1}^{\geq 0.25}(\geq 0.25)$ (respectively, $\mathcal{E}_{step2}^{\geq 0.25}(\geq 0.25)$) occurs during step one (respectively, step two) of the phase. We know:
\begin{align*}
\Pr(\mathcal{E}_M\mathcal{E}_R)\leq &\Pr(\mathcal{E}_M\wedge(\mathcal{E}_1\wedge\mathcal{E}_2))+\Pr(\mathcal{E}_R\wedge\overline{(\mathcal{E}_1\wedge\mathcal{E}_2)})\\
\leq &\Pr(\mathcal{E}_{un}\mathcal{E}_1)+\Pr(\overline{\mathcal{E}_{un}}\mathcal{E}_M\mathcal{E}_2)+\Pr(\mathcal{E}_R \wedge (\overline{\mathcal{E}_1}\vee \overline{\mathcal{E}_2}))
\end{align*}

The reminder of the proof bounds the three probabilities in the last line.

\begin{claim*}\label{claim-multicastadvadp-pu-raise-and-msg-claim1}
$\Pr(\mathcal{E}_{un}\mathcal{E}_1)\leq\exp(-\Theta(i^2))$.
\end{claim*}

\begin{claimproof}
Similar to the proof of Lemma \ref{lemma-multicastadp-fast-bcst}, by a coupling argument, we can assume in the first $R_i/4=2d\cdot 2^ii^3$ slots, $\{1,2,\cdots,C/4\}$ are unjammed; and all channels are jammed in the remaining slots. Here, $d$ is some sufficiently large constant. Let $\mathcal{R}_1$ denote the collection of the first half of these $2d\cdot 2^ii^3$ slots, and let $\mathcal{R}_2$ denote the second half. We further divide $\mathcal{R}_1$ into $2i$ segments. The number of informed node is denoted as $t$.

We first show $t\geq n/2$ at the end of $\mathcal{R}_1$. Consider a segment in which $t\leq n/2$ holds at the beginning, and let $u$ be a node that is informed at the beginning of the segment. For a slot in the segment, the probability that $u$ informs an uninformed node is at least
\begin{align*}
& p_u\cdot \prod_{v\in M\setminus\{u\}}\left(1-\frac{p_v}{C}\right)\cdot \left[\prod_{w\in V\setminus M}\left(1-\frac{p_w}{C}\right)-\prod_{w\in V\setminus M}\left(1-\frac{2p_w}{C}\right)\right] \cdot \frac{1}{4}\\
=\enspace & p_u\cdot \prod_{v\in V\setminus\{u\}}\left(1-\frac{p_v}{C}\right)\cdot \left[1-\prod_{w\in V\setminus M}\left(1-\frac{p_w}{C-p_w}\right)\right] \cdot \frac{1}{4}\\
\geq\enspace & p_u\cdot \prod_{v\in V\setminus\{u\}}\left(1-\frac{p_v}{C}\right)\cdot \left[1-\prod_{w\in V\setminus M}\left(1-\frac{p_w}{C}\right)\right] \cdot \frac{1}{4}\\
\geq\enspace & p_u\cdot e^{-2P_V}\cdot \left(1-e^{P_V-P_M}\right) \cdot (1/4)\\
\geq\enspace & p_u\cdot e^{-2P_V}\cdot \left((P_V-P_M)/2\right) \cdot (1/4)\\
\geq\enspace & p_u(P_V-P_M)/24\geq p_u(P_V/3)/24\\
=\enspace & p_uP_V/72
\end{align*}
where the last inequality is due to $P_M\leq 2/3P_V$, which in turn is due to $t\leq n/2$, the definition of $P_M$ and $P_V$, and the assumption in the lemma statement that nodes' working probabilities are within a factor of two.

As a result, the probability that no previously uninformed node is informed during the segment by $u$ is at most $(1-p_uP_V/72)^{d\cdot 2^ii^2}\leq\exp(-\Theta(i^2))$ since $p_uP_V\geq(8\cdot\sqrt{{C}/{(2^in)}})^2\cdot (n/C)\geq 64/2^i$. Take a union bound over the $O(n)$ informed nodes, we know after each segment the number of informed nodes is likely to at least double. Take another union bound over the $i$ segments, and together with the fact that there are at least $\lg{n}$ segments, we conclude that the number of informed nodes will reach $n/2$ by the end of $\mathcal{R}_1$, with probability at least $1-\exp(-\Theta(i^2))$.

Next, we focus on $\mathcal{R}_2$, and assume $t\geq n/2$ holds at the beginning of $\mathcal{R}_2$. Fix a node $u$ that is still uninformed at the beginning of $\mathcal{R}_2$. When $t\geq n/2$, due to lemma assumption, we have $P_M\geq 1/3P_V$. Therefore, for each slot in $\mathcal{R}_2$, the probability that $u$ is informed is at least $(1/4)\cdot p_u\cdot\sum_{v\in M}((p_v/C)\cdot \prod_{w\in V\setminus \{u,v\}}(1-p_w/C))\geq (1/4)\cdot p_u\cdot(P_V/3)\cdot e^{-2P_V}\geq p_uP_V/36$.

As a result, the probability that $u$ is not informed after $\mathcal{R}_2$ is at most $(1-p_uP_V/36)^{d\cdot 2^ii^3}\leq\exp(-\Theta(i^3))$. Take a union bound over the $O(n)$ uninformed nodes, we know all nodes will be informed by the end of step one with probability at least $1-\exp(-\Theta(i^3))$.
\end{claimproof}

\begin{claim*}\label{claim-multicastadvadp-pu-raise-and-msg-claim2}
$\Pr(\overline{\mathcal{E}_{un}}\mathcal{E}_{M}\mathcal{E}_2)\leq\Pr(\mathcal{E}_{M}\mathcal{E}_2|\overline{\mathcal{E}_{un}})\leq\exp(-\Theta(i^3))$.
\end{claim*}

\begin{claimproof}%[Proof sketch]
Fix a node $u$, and assume all nodes are alive and informed at the beginning of step two. By lemma assumption we have $P_M-p_u/C\geq (15/16)P_V$ when $n\geq 32$. Now, if $\mathcal{E}_2$ happens, then the expected number of message heard by $u$ in step two is at least $p^m_u/4 \cdot p_u\cdot (R_i/4)\geq ((15/16)P_V\cdot e^{-2P_V}/4)\cdot p_u\cdot(R_i/4)\geq 15ai^3/(4e)$. Similar to the proof of Lemma \ref{lemma-multicastadp-fast-halt} (except that we focus on message slots and apply the coupling argument accordingly), the probability that $u$ hears the message less than $ai^3$ times during a step two when $\mathcal{E}_2$ occurs is at most $\exp(-\Theta(i^3))$. Take a union over all nodes and the claim is proved.
\end{claimproof}

\begin{claim*}\label{claim-multicastadvadp-pu-raise-and-msg-claim3}
$\Pr(\mathcal{E}_{R}\wedge(\overline{\mathcal{E}_1}\vee \overline{\mathcal{E}_2}))\leq\exp(-\Theta(i^3C))$.
\end{claim*}
\begin{claimproof}%[Proof sketch]
Notice that $\Pr(\mathcal{E}_{R}\wedge(\overline{\mathcal{E}_1}\vee\overline{\mathcal{E}_2})) \leq \Pr(\mathcal{E}_{R}\overline{\mathcal{E}_1})+\Pr(\mathcal{E}_{R}\overline{\mathcal{E}_2}) \leq \sum_{u\in V}\big(\Pr(\mathcal{E}_{u,1}\overline{\mathcal{E}_1})+2\exp(-\Theta(i^3C))+\Pr(\mathcal{E}_{u,2}\overline{\mathcal{E}_{2}})+2\exp(-\Theta(i^3C))\big)$. Here, $\mathcal{E}_{u,1}$ (resp., $\mathcal{E}_{u,2}$) is the event that node $u$ hears silence more than $\Delta_{u}^{step1}/2$ (resp., $\Delta_{u}^{step2}/2$) times in step one (resp., step two) of the phase, and the last inequality is due to part (1) of Claim \ref{claim-multicastadvadp-pu-pv-bounded}. We only bound $\Pr(\mathcal{E}_{u,1}\overline{\mathcal{E}_1})$ as bounding the other one is similar.
When $\overline{\mathcal{E}_1}$ occurs, the expected number of silent slots heard by $u$ in step one is at most $p_u\cdot p_u^c\cdot (1/4\cdot 3/4 + 1\cdot 1/4)R_i\leq R_ip_u\cdot p_u^c\cdot 7/16\leq 7/16\Delta_{u}^{step1}$. Again via a coupling argument, we know $\Pr(\mathcal{E}_{u,1}\overline{\mathcal{E}_1})\leq \exp(-\Theta(i^3))$.
\end{claimproof}

The above three claims immediately lead to the lemma.
\end{proof}

As the final preparation before proving Lemma \ref{lemma-multicastadvadp-halt-imply-helper}, we introduce some notations that will be frequently used in the reminder of the analysis.

\begin{definition}\label{def-multicastadvadp-events}
Consider an $(i,j)$-phase, recall $p^{i,j}_u$ denotes the working probability of $u$ during step one and two. Following events are defined concerning nodes' working probabilities and status:
\begin{itemize}
	\item $\mathcal{E}^{i,j}_A$ is the event that every node is active at the beginning of the $(i,j)$-phase.
	\item $\mathcal{E}^{i,j}_B$ is the event that for all phases $k$ in super-epoch $i$ where $0\leq k\leq j$, and for all $u,v\in V$, it holds that $1/2\leq p^{i,k}_u/p^{i,k}_v\leq 2$.
	\item $\mathcal{E}^{i,j}_D$ is the event that for all phases $k$ in super-epoch $i$ where $0\leq k\leq j$, it holds that $P^{i,k}_V\leq 1/2$.
	\item $\mathcal{E}^{i,j}_{\eta'}$ is the event that for all phases $k$ in super-epoch $i$ where $0\leq k<j$, every active node has its $\eta^{i,k}\leq 3$ by the end of phase $k$.
	\item $\mathcal{E}^{i,j}_{u,sI}$ (resp., $\mathcal{E}^{i,j}_{u,sH}$) is the event that node $u$ is in \texttt{init} (resp., \texttt{helper}) status at the beginning of the $(i,j)$-phase.
	\item $\mathcal{E}^{i,j}_{R}$ is the event that some node raises its working probability by the end of the $(i,j)$-phase.
	\item For any node $u$, $\mathcal{E}^{i,j}_L$ is the event that $u$ halts by the end of the $(i,j)$-phase. Notice this implies $p^{i,j+1}_u\geq 64\sqrt{{C}/{(2^in_u)}}$.
	\item $\mathcal{E}^{i,j}_{M}$ is the event that some node hears the message less than  $ai^3$ times during step two of the $(i,j)$-phase.
	\item $\mathcal{E}^{i,j}_{\eta}$ is the event that some node has its $\eta^{i,j}<2.4$ by the end of the $(i,j)$-phase.
\end{itemize}
We abuse the notion to denote $u$'s working probability at the end of super-epoch $i$ as $p_u^{i,bi}$. We also often omit the indices $i$ and/or $j$ when they are clear from the context.
\end{definition}

\begin{proof}[\underline{Proof of Lemma \ref{lemma-multicastadvadp-halt-imply-helper}}]
Assume $(i,j)$ is the first phase in which some node halts, and $u$ is one of the nodes that halt. Assume $i>\lg{n}$. Denote the super-epoch number when $u$ becomes \texttt{helper} as $\hat{i}_u$, and let $\mathcal{E}_u$ be the event that ``$n_u\leq 4n$ and $\hat{i}_u>\lg (nC)$''.

Now, if we assume $\mathcal{E}^{i,j}_B$, $\mathcal{E}^{i,j}_{\eta'}$ and $\mathcal{E}_u$ all happen, then event $\mathcal{E}^{i,j}_{L}$ implies $p_u^{i,j}\geq p_u^{i,j+1}/\sqrt{2}\geq 64\sqrt{{C}/{(2^in_u)}}/\sqrt{2}\geq 16\sqrt{2}\sqrt{{C}/{(2^in)}}$, which further implies $p_v^{i,j}\geq 8\sqrt{2}\sqrt{{C}/{(2^in)}}$ for each node $v$. Moreover, if we assume $\mathcal{E}^{i,j}_B$, $\mathcal{E}^{i,j}_{\eta'}$ and $\mathcal{E}_u$ all happen, then event $\mathcal{E}^{i,j}_{L}$ also implies $\mathcal{E}^{i,j}_{R}$, since $p^{i,j+1}_u\geq 64\sqrt{{C}/{(2^in_u)}}$ must imply $p^{i,j}_u<64\sqrt{{C}/{(2^in_u)}}$ regardless of the value of $j$: (a) if $j=0$ then $p^{i,0}=C/2^i<64\sqrt{{C}/{(2^in_u)}}$ when $n_u\leq 4n$ and $i>\lg(nC)$; (b) if $u$ becomes \texttt{helper} during phase $j-1$ then $p^{i,j}\leq\sqrt{2}p^{i,j-1}=\sqrt{2}\sqrt{{C}/{(2^in_u)}}$; (c) otherwise we know $u$ is \texttt{helper} at the beginning of phase $j-1$ but did not halt by the end of phase $j-1$, thus $p^{i,j}<64\sqrt{{C}/{(2^in_u)}}$. At this point, apply Lemma \ref{lemma-multicastadvadp-pu-raise-and-msg}, we have $\Pr(\mathcal{E}^{i,j}_L\mathcal{E}^{i,j}_M|\mathcal{E}^{i,j}_B\mathcal{E}^{i,j}_D\mathcal{E}^{i,j}_{A}\mathcal{E}^{i,j}_{\eta'}\mathcal{E}_u)\leq\Pr(\mathcal{E}^{i,j}_R\mathcal{E}^{i,j}_M|\mathcal{E}^{i,j}_B\mathcal{E}^{i,j}_D\mathcal{E}^{i,j}_{A}\mathcal{E}^{i,j}_{\eta'}\mathcal{E}_u)\leq\exp(-\Theta(i^2))$. Combining this with Lemma \ref{lemma-multicastadvadp-pu-pv-bounded}, Corollary \ref{cor-multicastadvadp-pu-inc-bounded}, Lemma \ref{lemma-multicastadvadp-eta-diff-bounded}, and Lemma \ref{lemma-multicastadvadp-Pv-bounded}, we conclude:
\begin{align*}
& \Pr\left({\text{some node is not \texttt{helper} when }u\text{ halts firstly in super-epoch }i\geq\hat{i}_u}\wedge\mathcal{E}_u\right)\\
\leq & \Pr\left(\overline{\mathcal{E}^{bi}_B}\right)+\Pr\left(\overline{\mathcal{E}^{bi}_D}\right)+\Pr\left(\overline{\mathcal{E}^{bi}_{\eta'}}\right)+\sum_{j=0}^{bi-1}\Pr\left(\mathcal{E}^{j}_B\mathcal{E}^{j}_D\mathcal{E}^j_{A}\mathcal{E}^{j}_{\eta'}\mathcal{E}_u\mathcal{E}^j_L\left(\mathcal{E}^j_M \vee \mathcal{E}^j_{\eta}\right)\right)\\
\leq & \Pr\left(\overline{\mathcal{E}^{bi}_B}\right)+\Pr\left(\overline{\mathcal{E}^{bi}_D}\right)+\Pr\left(\overline{\mathcal{E}^{bi}_{\eta'}}\right)+\\
& \sum_{j=0}^{bi-1}\left(\Pr\left(\mathcal{E}^j_L\mathcal{E}^j_M \middle| \mathcal{E}^j_B\mathcal{E}^j_D\mathcal{E}^j_{A}\mathcal{E}^j_{\eta'}\mathcal{E}_u\right)
+\Pr\left(\mathcal{E}^j_L\mathcal{E}^j_{\eta} \middle| \mathcal{E}^j_B\mathcal{E}^j_D\mathcal{E}^j_{A}\mathcal{E}^j_{\eta'}\mathcal{E}_u\right)\right)\\
\leq & \exp(-\Theta(iC))+\exp(-\Theta(i^3C))+\exp(-\Theta(i^3C))+\\
& \sum_{j=0}^{bi-1}\left(\Pr\left(\mathcal{E}^j_R\mathcal{E}^j_M \middle| \mathcal{E}^j_B\mathcal{E}^j_D\mathcal{E}^j_{A}\mathcal{E}^j_{\eta'}\mathcal{E}_u\right) 
+\Pr\left(\mathcal{E}^j_R\mathcal{E}^j_{\eta} \middle| \mathcal{E}^j_B\mathcal{E}^j_D\mathcal{E}^j_{A}\mathcal{E}^j_{\eta'}\mathcal{E}_u\right)\right)\\
\leq & \exp(-\Theta(iC))+bi\cdot\exp(-\Theta(i^2))+bi\cdot\exp(-\Theta(i^3C))\\
\leq & \exp(-\Theta(i))
\end{align*}

Finally, apply Lemma \ref{lemma-multicastadvadp-no-helper-in-early-epoch} and Lemma \ref{lemma-multicastadvadp-estimate-upper-bound}:
\begin{align*}
& \Pr(\text{some node is not \texttt{helper} when some node halts})\\
\leq & \sum_u\sum_{i>\lg (nC)}\Pr\left({\text{some node is not \texttt{helper} when }u\text{ halts firstly in super-epoch }i\geq\hat{i}_u}\wedge\mathcal{E}_u\right)\\
\textcolor{white}{\leq} & + \sum_u\Pr\left(\hat{i}_u\leq\lg(nC)\right)
+ \sum_u\Pr\left(\left(n_u>4n\right)\wedge\left(\hat{i}_u>\lg(nC)\right)\right)\\
\leq & n\cdot \sum_{i>\lg (nC)}\exp(-\Theta(i))+n^{-\Omega(1)}+n^{-\Omega(1)}=n^{-\Omega(1)}
\end{align*}
which is exactly the lemma statement.
\end{proof}

\begin{proof}[\underline{Proof of Lemma \ref{lemma-multicastadvadp-good-estimate}}]
Since Lemma \ref{lemma-multicastadvadp-estimate-upper-bound} already implies the $n_u>4n$ part, here we focus on proving the $n_u<n/256$ part. We begin by proving the following claim which states if all nodes are alive and the working probability of $u$ is close to the ideal value $\Theta(\sqrt{C/(2^in)})$, then $u$ must have become \texttt{helper} already.

\begin{claim}\label{claim-multicastadvadp-pu-high-implies-helper}
Consider a super-epoch $i\geq\lg(nC)-7$. Assume at the beginning of the super-epoch all nodes are alive and there is a node $u$ in \texttt{init} status. Then by the end of each phase: the probability that at all nodes are alive and $u$ is still in \texttt{init} status with $p_u\geq 16\sqrt{{C}/{(2^in)}}$ is at most $\exp(-\Theta(i))$.

%Then for each phase: the probability that at all nodes are alive before the phase and $u$ is still in \texttt{init} status with $p_u\geq 16\sqrt{{C}/{(2^in)}}$ by the end the phase, is at most $\exp(-\Theta(i))$.
\end{claim}

\begin{claimproof}
In this proof, the ``some node'' in the definition of events $\mathcal{E}^{i,j}_{M},\mathcal{E}^{i,j}_{R}$ and $\mathcal{E}^{i,j}_{\eta}$ refers to node $u$. Let $\mathcal{E}^{i,j}_{P}$ be the event that $p^{i,j}_u\geq 16\sqrt{{C}/{(2^in)}}$, and $\mathcal{E}^{i,j}_{P'}$ be the event that $8\sqrt{2}\sqrt{{C}/{(2^in)}}\leq p^{i,j}_u<16\sqrt{{C}/{(2^in)}}$. Notice when $i\geq\lg(nC)-7$, we have $p_u^{i,0}<16\sqrt{{C}/{(2^in)}}$, which implies $p_u$ must have increased to reach $16\sqrt{{C}/{(2^in)}}$. As a result:
\begin{align*}
& \Pr\left(\mathcal{E}^{j+1}_{A}\wedge\mathcal{E}^{j+1}_{P}\wedge\mathcal{E}^{j+1}_{sI}\right)\\
\leq & \left(\Pr\left(\overline{\mathcal{E}^{bi}_B}\right)+\Pr\left(\overline{\mathcal{E}^{bi}_D}\right)+\Pr\left(\overline{\mathcal{E}^{bi}_{\eta'}}\right) \right)+
\Pr\left(\mathcal{E}^{j}_B\mathcal{E}^{j}_D\mathcal{E}_{A}^j\mathcal{E}_{\eta'}^j\wedge \left(\exists j'\leq j:\mathcal{E}^{j'}_{P'}\wedge\mathcal{E}^{j'}_R \wedge \mathcal{E}^{j'+1}_{sI} \right)\right)\\
\leq & \exp(-\Theta(iC))+\sum_{j'=0}^{j}\Pr\left(\mathcal{E}^{j'}_B\mathcal{E}^{j'}_D\mathcal{E}^{j'}_A\mathcal{E}^{j'}_{\eta'}\mathcal{E}^{j'}_{P'}\mathcal{E}^{j'}_R\wedge \mathcal{E}^{j'}_M\right)\\
\leq & \exp(-\Theta(iC))+
\sum_{j'=0}^{j}\Pr\left(\mathcal{E}^{j'}_R\mathcal{E}^{j'}_M\middle|\mathcal{E}^{j'}_B\mathcal{E}^{j'}_D\mathcal{E}^{j'}_{A}\mathcal{E}^{j'}_{\eta'}\mathcal{E}^{j'}_{P'}\right) \\
\leq & \exp(-\Theta(iC))+bi\cdot \exp(-\Theta(i^2))\\
\leq & \exp(-\Theta(i))
\end{align*}
where second to last inequality is due to Lemma \ref{lemma-multicastadvadp-pu-pv-bounded}, Corollary \ref{cor-multicastadvadp-pu-inc-bounded}, Lemma \ref{lemma-multicastadvadp-Pv-bounded}, and Lemma \ref{lemma-multicastadvadp-pu-raise-and-msg}.

%Similarly, we can derive $\Pr(\mathcal{E}^{bi}_{A}\wedge\mathcal{E}^{bi}_{P}\wedge\mathcal{E}^{bi}_{sI})=\exp(-\Theta(i))$.
\end{claimproof}

Fix a node $u$ and we often omit the subscript $u$ for simplicity. Recall event $\mathcal{E}^{i,j}_{P}$ means $p^{i,j}_u\geq 16\sqrt{{C}/{(2^in)}}$ as defined in the proof of Claim \ref{claim-multicastadvadp-pu-high-implies-helper}. Let $\mathcal{E}^{i,j}_H$ be the event that $u$ become \texttt{helper} during $(i,j)$-phase. Notice $\mathcal{E}^{i,j}_{H}$ implies $\mathcal{E}^{i,j}_{sI}$. Therefore:
\begin{align*}
& \Pr\left(u\text{ becomes \texttt{helper} with }n_u<n/256\right)\\
\leq & \Pr\left(u\text{ is \texttt{init} when some node halts firstly}\right)+\Pr\left(\exists i,j:\mathcal{E}^{i,j}_{A}\wedge\mathcal{E}^{i,j}_{H}\wedge\mathcal{E}^{i,j}_{P}\right)\\
\leq & \Pr\left(u\text{ is \texttt{init} when some node halts firstly}\right)+\Pr\left(\exists i\leq \lg(nC), j\in[0,bi):\mathcal{E}^{i,j}_{A}\wedge\mathcal{E}^{i,j}_{H}\right)\\
\phantom{\leq} & +\Pr\left(\exists i>\lg(nC), j\in[1,bi):\mathcal{E}^{i,j}_{A}\wedge\mathcal{E}^{i,j}_{sI}\wedge\mathcal{E}^{i,j}_{P}\right)+\Pr\left(\exists i>\lg(nC):\mathcal{E}^{i,0}_{P}\right)\\
\leq & n^{-\Omega(1)}+n^{-\Omega(1)}+\sum_{i=\lg(nC)+1}^{\infty}\sum_{j=1}^{bi-1}\exp(-\Theta(i))+0\leq n^{-\Omega(1)}
\end{align*}
where the second to last inequality is due to Lemma \ref{lemma-multicastadvadp-no-helper-in-early-epoch}, Lemma \ref{lemma-multicastadvadp-halt-imply-helper}, and Claim \ref{claim-multicastadvadp-pu-high-implies-helper}.
\end{proof}

\begin{proof}[\underline{Proof sketch of Lemma \ref{lemma-multicastadvadp-pu-inc-in-weak-jam-phase}}]
Due to Lemma \ref{lemma-multicastadvadp-pu-pv-bounded}, we know nodes' working probabilities will be within a factor of two, with probability at least $1-\exp(-\Omega(iC))$. Assume this indeed holds. When the phase is weakly jammed, $\mathbb{E}[N^{c,step1}_u]/\Delta^{step1}_u=\mathbb{E}[N^{c,step2}_u]/\Delta^{step2}_u\geq 0.95\cdot 0.95\cdot e^{-2P_V} \geq 0.874>2.6/3$ due to $P_V\leq n\cdot 2p^{i,j}_u/C\leq 1/64$. Similarly, $\mathbb{E}[N^{c,step3}_u]/\Delta^{step3}_u\geq 0.95\cdot 0.95\cdot e^{-2|V|\cdot p_{step3}/C}\geq 0.874$ due to $2|V|\cdot p_{step3}/C\leq 2nC/2^i\leq 1/32$. Finally, by a coupling argument and Chernoff bounds, we know the probability that $\eta_u\geq 2.6$ is at least $1-3 \cdot \exp(-\Theta(i^3C))=1-\exp(-\Theta(i^3C))$.
\end{proof}

\begin{proof}[\underline{Proof of Lemma \ref{lemma-multicastadvadp-pu-inc-in-weak-jam-epoch}}]
Define $\mathcal{E}_{j}$ be the event that the $j$\textsuperscript{th} phase in super-epoch $i$ is weakly jammed where $j\in\{0,1,\cdots,bi-1\}$. For any $J\subseteq\{0,1,\cdots,bi-1\}$, define $\mathcal{E}_{J}$ be the event that $J$ contains exactly the set of phases that are weakly jammed. Let $\mu=\log_{2^{(1/10)}}(\tilde{p_i}/(C/2^i))=100+5(i-\lg(nC))$ denote the maximum number of increments required for $p_u$ (i.e., the working probability of $u$) to reach $\tilde{p_i}$, if in each increment $p_u$ grows by a factor at least $2^{1/10}$. In the reminder of this proof, for simplicity, superscript $i$ and subscript $u$ are often omitted. Notice that $\mu<bi/2$ (since $i\geq 20$ and $b=20$), thus we have:
\begin{align*}
& \Pr\left(\left(\text{the super-epoch is weakly jammed}\right)\wedge\left(p^{bi}<\tilde{p_i}\right)\wedge\left(\mathcal{E}^{bi}_{sI}\vee\mathcal{E}^{bi}_{sH}\right)\right)\\
= & \sum_{J\subseteq\{0,1,\cdots,bi-1\}:|J|\geq bi/2}\Pr\left(\mathcal{E}_{J}\wedge\left(p^{bi}<\tilde{p_i}\right)\wedge\left(\mathcal{E}^{bi}_{sI}\vee\mathcal{E}^{bi}_{sH}\right)\right)\\
\leq & \sum_{J\subseteq\{0,1,\cdots,bi-1\}:|J|\geq bi/2}\sum_{j\in J}\Pr\left(\mathcal{E}_{j}\wedge\left(p^{j+1}/p^{j}<2^{(1/10)}\right)\wedge\left(\mathcal{E}^{j}_{sI}\vee\mathcal{E}^{j}_{sH}\right)\right)\\
\leq & 2^{bi}\cdot bi\cdot \exp(-\Theta(iC))=\exp(-\Theta(iC)).
\end{align*}
where the last inequality is due to Lemma \ref{lemma-multicastadvadp-pu-inc-in-weak-jam-phase}, and $\tilde{p_i}\leq C/(128n)$ when $i\geq 34+\lg(n/C)$.
\end{proof}

\begin{proof}[\underline{Proof of Lemma \ref{lemma-multicastadvadp-large-pu-halt}}]
We consider several potential scenarios:
\begin{itemize}
	\item If $u$ already halts at the beginning of the last phase of super-epoch $i$, the lemma holds.
	\item If $u$ is in \texttt{helper} status at the beginning of the last phase in super-epoch $i$, then $p^{i,bi}_u\geq\tilde{p_i}$ and $n_u\geq n/256$ imply $p^{i,bi}_u\geq\tilde{p_i}\geq64\sqrt{{C}/{(2^in_u)}}$, which in turn suggests $u$ must have stopped by the end of super-epoch $i$ according to the algorithm description.
	\item If $u$ is in \texttt{init} status at the beginning of the last phase in super-epoch $i$, but becomes \texttt{helper} during the last phase, then due to algorithm description and Corollary \ref{cor-multicastadvadp-pu-inc-bounded}, we know with probability at least $1-\exp(-\Theta(i^3C))$, $p^{i,bi}_u\leq\sqrt{2}\cdot p^{i,bi-1}_u=\sqrt{2}\sqrt{{C}/{(2^in_u)}}$, but this contradicts the assumption that $p^{i,bi}_u\geq\tilde{p_i}\geq64\sqrt{{C}/{(2^in_u)}}$.
	\item If $u$ is in \texttt{init} status at the end of super-epoch $i$, then the probability that $p^{i,bi}_u\geq\tilde{p_i}\geq 16\sqrt{{C}/{(2^in)}}$ also happens is at most $\exp(-\Theta(i))$. This is due to Claim \ref{claim-multicastadvadp-pu-high-implies-helper} and the assumptions of the lemma.
\end{itemize}

By now we have proved the lemma.
\end{proof}

\begin{proof}[\underline{Proof of Theorem \ref{thm-multicastadvadp}}]
Fix an arbitrary node $u$, we often omit the subscript $u$ and/or the superscript $i$ and/or $j$ when they are clear from the context. Also, recall the various definitions we introduced in Definition \ref{def-multicastadvadp-events}.

Our first step is to analyze how long $u$ remains active.

Let $\hat{I}=34+\lg C+\max\{\lg C,\lg n\}\geq \max\{34+\lg(nC),I_b\}$, let $r_i=bi\cdot 3R_i=3ab\cdot i^4\cdot 2^i$ be the number of slots in super-epoch $i$, and let $sr_i=\sum_{k=\hat{I}+1}^{i}r_k$ be the total number of slots from super-epoch $\hat{I}+1$ to super-epoch $i$ (both inclusive). We also set $r_{\hat{I}-1}=0$ and $sr_{\hat{I}}=0$ for the ease of presentation. It is easy to verify, for $i\geq \hat{I}+1$, $sr_{i}\leq \sum_{k=\hat{I}+1}^{i}3ab\cdot i^4\cdot 2^{k}\leq 3ab\cdot i^4\cdot 2^{i+1}\leq 15ab\cdot (i-1)^4\cdot 2^{i-1}=5r_{i-1}$; as for $i=\hat{I}$, we also have $sr_{i}=0\leq 5r_{i-1}$.

Let $\mathcal{E}_{sI \to A}$ be the event that $(\forall i,j: \overline{\mathcal{E}^{i,j}_{u,sI}}\vee\mathcal{E}^{i,j}_A)$, then $\Pr(\overline{\mathcal{E}_{sI \to A}})=\Pr(\exists i,j: \mathcal{E}^{i,j}_{u,sI}\wedge\overline{\mathcal{E}^{i,j}_A})=n^{-\Omega(1)}$ by Lemma \ref{lemma-multicastadvadp-halt-imply-helper}. Let $\mathcal{E}^{i,j}_H$ be the event that $u$ becomes \texttt{helper} during $(i,j)$-phase, $\mathcal{E}_{u}$ be the event that $n_u\geq n/256$, and $\mathcal{E}_{H\to u}$ be the event that $(\forall i,j: \overline{\mathcal{E}^{i,j}_{H}}\vee \mathcal{E}_u)$, then $\Pr(\overline{\mathcal{E}_{H \to u}})=\Pr(\exists i,j: \mathcal{E}^{i,j}_{H}\wedge \overline{\mathcal{E}_u})=n^{-\Omega(1)}$ by Lemma \ref{lemma-multicastadvadp-good-estimate}.

Let $\mathcal{E}_{\tilde{p}}$ be the event that $p^{i,bi}_u\geq\tilde{p_i}$, and let constant $\beta=2400$. Let random variable $L$ denote node $u$'s actual runtime starting from super-epoch $\hat{I}+1$. Then we have:
\begin{align*}
& \Pr\left(L>5\beta T/C\right)\\
\leq & \Pr\left(\overline{\mathcal{E}_{sI \to A}}\right) + \Pr\left(\overline{\mathcal{E}_{H \to u}}\right) + \sum_{i=\hat{I}+1}^{\infty}\Pr\left(\mathcal{E}_{sI \to A}\wedge\mathcal{E}_{H \to u}\wedge\left(L>\frac{5\beta T}{C}\right)\wedge\left(r_{i-1}\leq\frac{\beta T}{C}<r_{i}\right)\right)\\
\leq & n^{-\Omega(1)} + \sum_{i=\hat{I}+1}^{\infty}\Pr\left(\mathcal{E}_{sI \to A}\wedge\mathcal{E}_{H \to u}\wedge(L>sr_i)\wedge(\beta T/C <r_{i})\right)\\
\leq & n^{-\Omega(1)} + \sum_{i=\hat{I}+1}^{\infty}\Pr\left(\overline{\mathcal{E}_{\tilde{p}}}\wedge(L>sr_i)\wedge(\beta T/C<r_{i})\right) + \sum_{i=\hat{I}+1}^{\infty}\Pr\left(\mathcal{E}_{\tilde{p}}\wedge(L>sr_{i})\middle|\mathcal{E}_{sI \to A}\wedge\mathcal{E}_{H \to u}\right) \\
\leq & n^{-\Omega(1)} + \sum_{i=\hat{I}}^{\infty}\exp(-\Theta(iC))+\sum_{i=\hat{I}}^{\infty}\exp(-\Theta(i))\\
\leq & n^{-\Omega(1)}
\end{align*}
Notice, in the second to last inequality, the first part is due to Lemma \ref{lemma-multicastadvadp-pu-inc-in-weak-jam-epoch} and the fact that Eve spends less than $r_iC/\beta=bi/2\cdot 0.05^2R_iC$ in super-epoch $i$ implies super-epoch $i$ is weakly jammed; the second part is due to Lemma \ref{lemma-multicastadvadp-large-pu-halt}.

Take a union bound over all nodes, we know all nodes will terminate within $(\sum_{k=I_b}^{\hat{I}}bk\cdot 3R_k)+5\beta T/C=O(T/C+(nC+C^2)\cdot\lg^4(nC))$ slots, with high probability.

We continue to analyze the cost of node $u$.

Let $F_{step1,2}^{i,j}$ (resp., $F_{step1,2}^{i}$) be node $u$'s total actual cost during step one and step two in $(i,j)$-phase (resp., in all phases within super-epoch $i$).
Also let $F_{step1,2}$ be node $u$'s total actual cost during step one and step two in all phases starting from super-epoch $\hat{I}+1$. 

Let $\mathcal{E}^{step1,2}_F$ be the event that $F_{step1,2}^{i,j}\leq (2R_i)\cdot (1+1/2)\cdot (2\cdot 1024\sqrt{{C}/{(2^in)}})=3R_ip_i$ holds for each $i>\hat{I}$ and $0\leq j<bi$, where $p_i$ is set as $2048\sqrt{{C}/{(2^in)}}$. Then we have:
\begin{align*}
& \Pr\left(\overline{\mathcal{E}^{step1,2}_F} \wedge \mathcal{E}_{sI \to A}\wedge \mathcal{E}_{H \to u}\right)\\
\leq & \sum_{i=\hat{I}+1}^{\infty}\sum_{j=0}^{bi-1}\Pr\left(\mathcal{E}_{sI}^{i,j}\wedge\left(p^{i,j}> 16\sqrt{\frac{C}{2^in}}\right)\wedge \mathcal{E}_{sI \to A}\right) +\\
\phantom{\leq} & \sum_{i=\hat{I}+1}^{\infty}\sum_{j=0}^{bi-1}\Pr\left(\mathcal{E}_{sI}^{i,j}\wedge\left(p^{i,j}\leq 16\sqrt{\frac{C}{2^in}}\right)\wedge\left(F_{step1,2}^{i,j}>3R_ip_i\right)\right) \\
\phantom{\leq} & \sum_{i=\hat{I}+1}^{\infty}\sum_{j=0}^{bi-1}\Pr\left(\overline{\mathcal{E}_{sI}^{i,j}}\wedge\left(F_{step1,2}^{i,j}>3R_ip_i\right)\wedge\mathcal{E}_{H \to u}\right)+\\
\leq & \sum_{i=\hat{I}+1}^{\infty}\sum_{j=0}^{bi-1}\left(\exp(-\Theta(i)) + \exp\left(-\Theta\left(i^3\cdot\sqrt{\frac{2^iC}{n}}\right)\right) +\exp\left(-\Theta\left(i^3\cdot\sqrt{\frac{2^iC}{n}}\right)\right)\right)\\
\leq & n^{-\Omega(1)}
\end{align*}
Notice, in the second to last inequality: the first part is due to Claim \ref{claim-multicastadvadp-pu-high-implies-helper}; the second part is due to a Chernoff bound and $p^{i,j}\leq 16\sqrt{{C}/{(2^in)}}<1024\sqrt{{C}/{(2^in)}}$; and the last part is due to a Chernoff bound and the fact that $p^{i,j}<64\sqrt{{C}/{(2^in_u)}}\leq 1024\sqrt{{C}/{(2^in)}}$ when $n_u\geq n/256$. 

Notice that event $\mathcal{E}^{step1,2}_F$ will imply $F_{step1,2}^{i}\leq bi\cdot 3R_i\cdot p_i=r_ip_i$. We are now ready to bound $F_{step1,2}$. Define constant $\gamma=2^{28}\cdot 3ab$, then we have:
\begin{align*}
& \Pr\left(\left(F_{step1,2}\right)^2>\gamma\beta\cdot\lg^4(T)\cdot T/n\right)\\
\leq & \Pr\left(\overline{\mathcal{E}_{sI \to A}}\right) + \Pr\left(\overline{\mathcal{E}_{H \to u}}\right) + \Pr\left(\overline{\mathcal{E}^{step1,2}_F}\wedge\mathcal{E}_{sI \to A}\wedge\mathcal{E}_{H \to u}\right) + \\
\phantom{\leq} & \sum_{i=\hat{I}+1}^{\infty}\Pr\left(
\mathcal{E}^{step1,2}_F\wedge\left(\left(F_{step1,2}\right)^2>\gamma\beta\cdot\lg^4(T)\cdot T/n\right)\wedge\mathcal{E}_{sI \to A}\wedge\mathcal{E}_{H \to u}\wedge\left(r_{i-1}\leq\beta T/C<r_{i}\right)
\right)\\
\leq & n^{-\Omega(1)}+n^{-\Omega(1)}+n^{-\Omega(1)} + \\
\phantom{\leq} & \sum_{i=\hat{I}+1}^{\infty}\Pr\left(\left(L>sr_i\right)\wedge\mathcal{E}_{sI \to A}\wedge\mathcal{E}_{H \to u}\wedge\left(\beta T/C<r_{i}\right)\right) + \\
\phantom{\leq} & \sum_{i=\hat{I}+1}^{\infty}\Pr\left(\left(L\leq sr_i\right)\wedge\mathcal{E}^{step1,2}_F \wedge\left(\left(F_{step1,2}\right)^2>256r^2_{i-1}p^2_{i-1}\right)\right)\\
\leq & n^{-\Omega(1)}+n^{-\Omega(1)}+n^{-\Omega(1)}+n^{-\Omega(1)}+0
=n^{-\Omega(1)}
\end{align*}
Notice, in the second to last inequality, we consider two complement scenarios: either $u$ already halts by the end of super-epoch $i$, or not. Moreover, in case $u$ already halts by the end of super-epoch $i$, then it is easy to verify $\gamma\beta\cdot\lg^4(T)\cdot T/n\geq 256r^2_{i-1}p^2_{i-1}$ holds for all $i>\hat{I}$ ($a\cdot b$ should be sufficient large to ensure $\lg T\geq i-1$) or $i=\hat{I}$ (since $r_{\hat{I}-1}=0$) . In the last inequality, the first part is due to the fact that when super-epoch $i$ is weakly jammed, $u$ will halt by the end of the super-epoch with probability at least $1-\exp(-\Theta(i))$. (See our earlier analysis in this theorem proof regarding the runtime of node $u$.) The second part, on the other hand, is because: when $\mathcal{E}^{step1,2}_{F}$ happens, we have $F_{step1,2}\leq \sum_{k=\hat{I}+1}^{i}r_kp_k\leq \sum_{k=\hat{I}+1}^{i}2^{11}\cdot 3abi^4\sqrt{{2^kC}/{n}}\leq 2^{12}\cdot 3ab(i-1)^4\sqrt{{2^{i+1}C}/{n}}/(\sqrt{2}-1) \leq 2^{13}\cdot 3ab(i-1)^4\sqrt{{2^{i-1}C}/{n}}/(1/4)=16r_{i-1}p_{i-1}$, contradicting $(F_{step1,2})^2>256r^2_{i-1}p^2_{i-1}$.

Similarly, we can also bound $F_{step3}$, which is node $u$'s total actual cost during step step three, in all phases starting from super-epoch $\hat{I}+1$. Let $\mathcal{E}^{step3}_F$ be the event that ``$\forall i\geq\hat{I}, F^{step3}_i\leq(3/2)\cdot(r_i/3)\cdot(2\cdot C^2/2^i)=3abi^4C^2$''. Then by a Chernoff bound, $\Pr(\overline{\mathcal{E}^{step3}_F})\leq \sum_{i=\hat{I}}^{\infty} \exp(-\Theta(i^4C^2))\leq n^{-\Omega(1)}$. Define constant $\gamma'=6ab$, then we have:
\begin{align*}
&\Pr\left(F_{step3}>\gamma'C^2\cdot\left(\lg(T)+\hat{I}\right)^5\right)\\
\leq & \Pr\left(\overline{\mathcal{E}_{sI \to A}}\right) + \Pr\left(\overline{\mathcal{E}_{H \to u}}\right) + \Pr\left(\overline{\mathcal{E}^{step3}_F}\right) +\\
\phantom{\leq} & \sum_{i=\hat{I}+1}^{\infty}\Pr\left(
\mathcal{E}^{step3}_F\wedge\left(F^{step3}>\gamma'C^2\cdot\left(\lg(T)+\hat{I}\right)^5\right)%\\
\wedge\mathcal{E}_{sI \to A}\wedge\mathcal{E}_{H \to u}\wedge\left(r_{i-1}\leq\beta T/C<r_{i}\right)
\right)\\
\leq & n^{-\Omega(1)}+n^{-\Omega(1)}+n^{-\Omega(1)}+\\
\phantom{\leq} & \sum_{i=\hat{I}+1}^{\infty}\Pr\left(\left(L>sr_i\right)\wedge\mathcal{E}_{sI \to A}\wedge\mathcal{E}_{H \to u}\wedge\left(\beta T/C<r_{i}\right)\right) +\\
\phantom{\leq} & \sum_{i=\hat{I}+1}^{\infty}\Pr\left(\left(L\leq sr_i\right)\wedge\mathcal{E}^{step3}_F\wedge\left(F_{step3}>\gamma'C^2\cdot(i-1)^5\right)\right)\\
\leq & n^{-\Omega(1)}+n^{-\Omega(1)}+n^{-\Omega(1)}+n^{-\Omega(1)}+0=n^{-\Omega(1)}
\end{align*}
Again, in the second to last inequality, we consider two complement scenarios: either $u$ already halts by the end of super-epoch $i$, or not. Moreover, in case $u$ already halts by the end of super-epoch $i$, then it is easy to verify that $\lg(T)+\hat{I}\geq i-1$ holds when $i\geq\hat{I}$ and $r_{i-1}\leq\beta T/C<r_{i}$. In the last inequality, the first part is due to the fact that when super-epoch $i$ is weakly jammed, $u$ will halt by the end of the super-epoch with probability at least $1-\exp(-\Theta(i))$. The second part, on the other hand, is because: when $\mathcal{E}^{step3}_{F}$ happens, we have $F_{step3}\leq\sum_{k=\hat{I}+1}^{i}3abi^4C^2\leq 3abi^5C^2\leq 6ab(i-1)^5C^2 =\gamma'C^2 \cdot (i-1)^5$.

By now, we can conclude, with high probability, the total cost of $u$ is bounded by $O(\sqrt{T/n\cdot\lg^4{T}})+O(C^2\cdot \lg^5(nCT))+(\sum_{k=I_b}^{\hat{I}}bk\cdot 3R_k)=O(\sqrt{T/n\cdot\lg^4{T}}+C^2\cdot\lg^5(nCT)+(nC+C^2)\cdot\lg^4(nC))$. Take a union bound over all nodes, we know the cost of each node is $\tilde{O}(\sqrt{T/n}+nC+C^2)$, with high probability.

Finally, notice that the algorithm itself ensures each node must have learned the message when the node halts, and this completes the proof of the theorem.
\end{proof}

\section{Omitted Parts in the Lower Bounds Section }\label{sec-app-lower-bound-proof}

This section contains two parts: (a) some auxiliary material; and (b) proofs that are omitted in the main body of the paper.

\subsection{Auxiliary Material}\label{subsec-app-lower-bound-aux-material}

\begin{claim}\label{claim-keymath-simple}
For any $n\in\mathbb{N}^+$, for any $\vec{a}=(a_1,a_2,\dots,a_n)\in(\mathbb{R}^{\geq0})^n, \vec{x}=(x_1,x_2,\dots,x_{n})\in(\mathbb{R}^{\geq0})^n,\vec{y}=(y_1,y_2,\dots,y_{n})\in(\mathbb{R}^{\geq0})^n$, let $\vec{x}^i=(x_i,\dots,x_i)$ and $\vec{y}^i=(y_i,\dots,y_i)$. Then $\min_i\{F_i\}\leq F$, where $F=(\vec{a}\cdot \vec{x})\cdot(\vec{a}\cdot \vec{y})$ and $F_i=(\vec{a}\cdot \vec{x}^i)(\vec{a}\cdot \vec{y}^i)$.
\end{claim}

\begin{claimproof}
Define $F'=\sum_{i=1}^{n}\eta_i\cdot F_i$, where $\sum_{i=1}^{n}\eta_i=1$ and $0\leq\eta_i\leq 1$ for all $i\in[n]$. We intend to prove $F\geq F'$.

Denote $\left\|\vec{a}\right\|_1 =\sum_{i=1}^{n} |a_i|$. By definition, $F_i=\left\|\vec{a}\right\|_1^2 \cdot x_i\cdot y_i$. Therefore, our goal is to construct $\vec{\eta}=(\eta_1,\dots,\eta_n)$ where $F'=\sum_{i=1}^{n} \eta_i\cdot \left\|\vec{a}\right\|_1^2 \cdot x_iy_i$, and prove $F\geq F'$.

Let $S_{i,j}=(a_i x_i)\cdot(a_jy_j)$, by definition $F=\sum_{i=1}^{n}\sum_{j=1}^{n}S_{i,j}$. We set $\vec{\eta}=(a_1^2,\dots,a_n^2)$ initially. For every $(i,j)$ pair where $i\neq j$, we will lower bound $S_{i,j}+S_{j,i}$ with some linear combination of $x_iy_i$ and $x_jy_j$, and adjust $\vec{\eta}$ along the way. Specifically, there are four cases:
\begin{itemize}
	\item Case 1: if $x_i\geq x_j$ and $y_i\leq y_j$, by the sorting inequality $x_iy_j+ x_jy_i \geq x_iy_i+ x_jy_j$, thus $S_{i,j}+S_{j,i}=a_i\cdot a_j\cdot (x_iy_j+x_jy_i)\geq a_i\cdot a_j\cdot (x_iy_i+ x_jy_j)$. We set $\eta_i \leftarrow \eta_i + a_i\cdot a_j$	and $\eta_j \leftarrow \eta_j + a_i\cdot a_j$.
	\item Case 2: if $x_i\leq x_j$ and $y_i\geq y_j$, then this is similar to case 1 as $x_iy_j+ x_jy_i\geq x_iy_i+ x_jy_j$.
	\item Case 3: if $x_i\geq x_j$ and $y_i\geq y_j$, then $S_{i,j}\geq a_i\cdot a_j\cdot x_jy_j$ and $S_{j,i}\geq a_i\cdot a_j\cdot x_jy_j$,  thus we set $\eta_j \leftarrow \eta_j + 2\cdot a_i\cdot a_j$.
	\item Case 4: if $x_i\leq x_j$ and $y_i\leq y_j$, then $S_{i,j}\geq a_i\cdot a_j\cdot x_iy_i$ and $S_{j,i}\geq a_i\cdot a_j\cdot x_iy_i$, thus we set $\eta_i \leftarrow \eta_i + 2\cdot a_i\cdot a_j$.
\end{itemize}

At this point, let $\vec{\eta}=\frac{1}{\left\|\vec{a}\right\|^2_1}\vec{\eta}$, and we have $F'=\sum_{i=1}^{n}\eta_i\cdot F_i$.
\end{claimproof}

\begin{claim}\label{claim-keymath-complex}
Let $\vec{a},\vec{x},\vec{y},\vec{x}^i,\vec{y}^i$ be similarly defined as in Claim~\ref{claim-keymath-simple}. For any $n\in\mathbb{N}^+$, for any $b\in\mathbb{R}$ and $c\in\mathbb{R}$, it holds that $\min_i\{F^{new}_i\}\leq F^{new}$, where $F^{new}=(\vec{a}\cdot \vec{x}+b)\cdot(\vec{a}\cdot \vec{y}+c)$, and $F^{new}_i=(\vec{a}\cdot \vec{x}^i+b)(\vec{a}\cdot \vec{y}^i+c)$.
\end{claim}

\begin{claimproof}
Similar to the proof of Claim~\ref{claim-keymath-simple}, we show $F^{new}\geq \sum_{i=1}^{n} \eta_i\cdot F^{new}_i$, where $\vec{\eta}$ is constructed as in the proof of Claim~\ref{claim-keymath-simple}. We only need to take care of the terms related to $b$ and $c$. Let $F^{b,c}=(\vec{a}\cdot \vec{x})c+(\vec{a}\cdot \vec{y})b+bc$ and $F_i^{b,c}=(\vec{a}\cdot \vec{x}^i)c+(\vec{a}\cdot \vec{y}^i)b+bc$.

If we assume all $(i,j)$ pairs belong to either case 1 or case 2, then $\eta_i=a_i/\left\|\vec{a}\right\|_1$. By simple calculation we know $F^{b,c}=\sum_{i=1}^{n} \eta_i\cdot F_i^{b,c}$, therefore $F^{new}=F+F^{b,c}\geq F'+F^{b,c}=\sum_{i=1}^{n} \eta_i\cdot (F_i+F_i^{b,c})=\sum_{i=1}^{n} \eta_i\cdot F^{new}_i$, as desired.

Otherwise, for each $(i,j)$ pair belonging to case 3, let $F^{b,c}\leftarrow F^{b,c}-\frac{a_ia_j}{\left\|a\right\|_1}((x_i -x_j)\cdot c+(y_i -y_j)\cdot b)$; and for each pair of $i,j$ belonging to case 4, let $F^{b,c}\leftarrow F^{b,c}-\frac{a_ia_j}{\left\|a\right\|_1}((x_j -x_i)\cdot c+(y_j -y_i)\cdot b)$. Notice that each modification makes $F^{b,c}$ smaller, and after all modifications $F^{b,c}=\sum_{i=1}^{n} \eta_i\cdot F_i^{b,c}$ is satisfied.
\end{claimproof}

\subsection{Omitted Proofs}\label{subsec-app-lower-bound-proof}

\begin{proof}[\underline{Proof of Theorem \ref{thm-cost-lower-bound-1-to-1}}]
Throughout the proof, we assume Eve uses jamming strategy $\mathcal{S}$: in each slot, for each channel, if the probability that Alice successfully transmits $m$ to Bob over this channel exceeds $1/T$, then Eve jams this channel so long as her energy is not depleted.

In general, any multi-channel 1-to-1 communication algorithm $\mathcal{A}$ can be viewed as a decision tree where each non-leaf node (i.e., internal node) $u$ at depth $i$ (assuming root has depth 1) represents one potential configuration of the states of Alice and Bob at the beginning of slot $i$ (this configuration specifies the distribution of Alice's behavior and Bob's behavior in slot $i$), and an edge from a depth $i$ node $u$ to a depth $i+1$ node $v$ represents the actual behavior of Alice, Bob, and Eve in slot $i$ (so that Alice's and Bob's states change from $u$ to $v$). Recall we assume Alice and Bob automatically stops once Bob is informed, so there are two kinds of leaf nodes. In particular, for any leaf node $v$, if the unique edge directed to $v$ indicates a successful transmission from Alice to Bob then $v$ is labeled ``success'', otherwise labeled ``fail''.
Denote the set of non-leaf nodes as $\mathcal{I}$, and the set of fail-leaf nodes as $\mathcal{F}$. We say a node is a ``boundary'' node if it has fail-leaf children, and denote the set of boundary nodes as $\mathcal{B}\subseteq \mathcal{I}$. For any internal node $u$, let $a_{u,k}$ denote the probability that Alice broadcasts $m$ on channel $k$, and $b_{u,k}$ denote the probability that Bob does not broadcast but listens on channel $k$. We further define $\hat{a}_{u,k}$ and $\hat{b}_{u,k}$: if $a_{u,k}b_{u,k}>1/T$ then $\hat{a}_{u,k}=\hat{b}_{u,k}=0$, otherwise $\hat{a}_{u,k}=a_{u,k}$ and $\hat{b}_{u,k}=b_{u,k}$. (Notice $a_{u,k}$ and $b_{u,k}$ are helpful when computing Alice's and Bob's cost, while $\hat{a}_{u,k}$ and $\hat{b}_{u,k}$ are helpful when computing the algorithm's success probability.) Denote $p_u$ as the probability that the algorithm reaches $u$ starting from the root. Lastly, for each node $u$, the set of its ancestors and $u$ itself (i.e., all nodes in the path from root to $u$) is denoted as $\mathcal{H}_u$, the set of $u$'s children is denoted as $\mathcal{C}_u$, and the set of $u$'s descendants and $u$ itself (i.e., all nodes in the subtree rotted at $u$) is denoted as $\mathcal{D}_u$.

If the multi-channel 1-to-1 communication algorithm $\mathcal{A}$ in consideration is \emph{adaptive}, then for each boundary node $u$ in $\mathcal{A}$, we can construct an oblivious algorithm $\mathcal{A}_u$. In particular, the length of $\mathcal{A}_u$ is the number of nodes in the path from root to $u$ (including $u$) in the decision tree of $\mathcal{A}$. Let $v_i$ be the $i$\textsuperscript{th} node in the path from root to $u$ in the decision tree of $\mathcal{A}$, then in slot $i$ in $\mathcal{A}_u$, the broadcasting probability of Alice and listening probability of Bob are $a_{v_i,k}$ and $b_{v_i,k}$, respectively.

Now, we claim, any multi-channel 1-to-1 communication algorithm $\mathcal{A}$ can be transformed into another algorithm $\mathcal{A}'$ that is a ``convex combination'' of a set $\Sigma$ of oblivious algorithms, without changing the success probability or the product of Alice's and Bob's expected cost. This is trivially true if $\mathcal{A}$ itself is already an oblivious algorithm, so we assume $\mathcal{A}$ is an adaptive algorithm. In such case, $\Sigma$ is the set consisting of oblivious algorithms $\mathcal{A}_w$ for each boundary node $w\in \mathcal{B}$. At the beginning of $\mathcal{A}'$, Alice and Bob will toss a (shared) coin and pick an algorithm in $\Sigma$ to execute.\footnote{Allowing for a shared coin is fine here, since we only use it to pick an algorithm from $\Sigma$ and do not exploit it in any other ways.} (Eve will also know the result of this coin toss.) The probability that $\mathcal{A}'$ chooses $\mathcal{A}_w$ is $p'_w=(\sum_{x\in\mathcal{F}\cap\mathcal{C}_w}p_x)/(\prod_{v\in \mathcal{H}_w}(1-\sum_{k=1}^{C}\hat{a}_{v,k}\hat{b}_{v,k}))$. To prove the claim, we will show: (i) $\sum_{w\in \mathcal{B}}p'_w=1$; (ii) $\mathbb{E}_{\mathcal{A}}[A]=\mathbb{E}_{\mathcal{A}'}[A]$, $\mathbb{E}_{\mathcal{A}}[B]=\mathbb{E}_{\mathcal{A}'}[B]$; and (iii) $\sum_{x\in \mathcal{F}}p_x= \sum_{w\in \mathcal{B}}p'_w\cdot(\prod_{v\in \mathcal{H}_w}(1-\sum_{k=1}^{C}\hat{a}_{v,k}\hat{b}_{v,k}))$.

To prove (i), we prove the following claim by induction on the max distance between (internal node) $u$ and set $\mathcal{D}_u$: $\forall u\in\mathcal{I}\textrm{, }\sum_{w: w\in \mathcal{B} \cap \mathcal{D}_u}p'_w=p_u/\prod_{v\in \mathcal{H}_u\setminus\{u\}}(1-\sum_{k=1}^{C}\hat{a}_{v,k}\hat{b}_{v,k})$. The base case is $u\in \mathcal{B}$ and every child of $u$ is a leaf, which holds trivially by the definition of $p'_w$ and the fact $\mathcal{B} \cap \mathcal{D}_u=\{u\}$. For the inductive step, assume $u\in \mathcal{I}$ has non-leaf children $v_1,\cdots,v_m\in \mathcal{I}$, and the claim holds for all $v_1,\cdots,v_m$. Then we know:
\begin{align*}
\sum_{w: w\in \mathcal{B} \cap \mathcal{D}_u}p'_w
&=\mathbb{I}[u\in\mathcal{B}]\cdot p'_u+\sum_{j=1}^{m} \sum_{w: w\in \mathcal{B} \cap \mathcal{D}_{v_j}}p'_w\\
&=\mathbb{I}[u\in\mathcal{B}]\cdot p'_u+\sum_{j=1}^{m}\left(p_{v_j}\middle/\prod_{v\in \mathcal{H}_{v_j}\setminus \{v_j\}}\left(1-\sum_{k=1}^{C}\hat{a}_{v,k}\hat{b}_{v,k}\right)\right)\\
&=\mathbb{I}[u\in\mathcal{B}]\cdot\left(\sum_{x\in \mathcal{F} \cap \mathcal{C}_u}p_x\middle/\prod_{v\in \mathcal{H}_u}\left(1-\sum_{k=1}^{C}\hat{a}_{v,k}\hat{b}_{v,k}\right)\right)+\left(\sum_{j=1}^{m}p_{v_j}\middle/\prod_{v\in \mathcal{H}_{u}}\left(1-\sum_{k=1}^{C}\hat{a}_{v,k}\hat{b}_{v,k}\right)\right)\\
&=p_u\cdot\left(1-\sum_{k=1}^{C}\hat{a}_{u,k}\hat{b}_{u,k}\right)/\prod_{v\in \mathcal{H}_{u}}\left(1-\sum_{k=1}^{C}\hat{a}_{v,k}\hat{b}_{v,k}\right)\\
&=p_u/\prod_{v\in \mathcal{H}_{u}\setminus\{u\}}\left(1-\sum_{k=1}^{C}\hat{a}_{v,k}\hat{b}_{v,k}\right)
\end{align*}
which proves the inductive step.

Therefore, $\sum_{w\in \mathcal{B}}p'_w=\sum_{w:w\in \mathcal{B}\cap \mathcal{D}_{root}}p'_w=p_{root}/\prod_{v\in \mathcal{H}_{root}\setminus \{root\}}(1-\sum_{k=1}^{C}\hat{a}_{v,k}\hat{b}_{v,k})=p_{root}=1$, since $\mathcal{H}_{root}\setminus \{root\}=\emptyset$. Thus (i) is proved.

To prove (ii), we only prove $\mathbb{E}_{\mathcal{A}}[A]=\mathbb{E}_{\mathcal{A}'}[A]$ since the other one is similar. Notice that $\mathbb{E}_{\mathcal{A}}[A]=\sum_{u\in \mathcal{I}}p_u\cdot \sum_{k=1}^{C}a_{u,k}$. Moreover, we have:
\begin{align*}
\mathbb{E}_{\mathcal{A}'}[A]&=\sum_{w\in \mathcal{B}}p'_w\cdot \mathbb{E}_{\mathcal{A}_w}[A]\\
&=\sum_{w\in \mathcal{B}}p'_w\cdot\left(\sum_{u\in\mathcal{H}_w}\left(\left(\sum_{k=1}^{C}a_{u,k}\right)\cdot \prod_{v\in \mathcal{H}_{u}\setminus \{u\}}\left(1-\sum_{k=1}^{C}\hat{a}_{v,k}\hat{b}_{v,k}\right)\right)\right)\\
&=\sum_{u\in \mathcal{I}}\left(\sum_{k=1}^{C}a_{u,k}\right)\cdot\left(\sum_{w:w\in \mathcal{B}\cap \mathcal{D}_u} p'_w \cdot \prod_{v\in \mathcal{H}_{u}\setminus \{u\}}\left(1-\sum_{k=1}^{C}\hat{a}_{v,k}\hat{b}_{v,k}\right)\right)\\
&=\sum_{u\in \mathcal{I}}\left(\sum_{k=1}^{C}a_{u,k}\right)\cdot p_u
\end{align*}
where the last equality is due to the claim mentioned when proving (i), thus (ii) is proved.

Lastly, to prove (iii), we start from the right-hand side:
\begin{align*}
& \quad \sum_{w\in \mathcal{B}}p'_w\cdot \prod_{v\in \mathcal{H}_w}\left(1-\sum_{k=1}^{C}\hat{a}_{v,k}\hat{b}_{v,k}\right)\\
&=\sum_{w\in \mathcal{B}} \left(\sum_{x\in \mathcal{F} \cap \mathcal{C}_w}p_x\middle/\left(\prod_{v\in \mathcal{H}_w}\left(1-\sum_{k=1}^{C}\hat{a}_{v,k}\hat{b}_{v,k}\right)\right)\right)\cdot \left(\prod_{v\in \mathcal{H}_w}\left(1-\sum_{k=1}^{C}\hat{a}_{v,k}\hat{b}_{v,k}\right)\right)\\
&=\sum_{x\in \mathcal{F}}p_x
\end{align*}

At this point, we have proved any multi-channel 1-to-1 communication algorithm $\mathcal{A}$ can be transformed into another algorithm $\mathcal{A}'$ that is a ``convex combination'' of a set $\Sigma$ of oblivious algorithms, without changing the success probability or the product of Alice's and Bob's expected cost. Assume $\lambda$ is the failure probability of $\mathcal{A}$, let $\Sigma_{2\lambda}$ be the subset of $\Sigma$ so that each oblivious algorithm in $\Sigma_{2\lambda}$ fails with probability at most $2\lambda$. By Markov's inequality, $\sum_{w\in \Sigma_{2\lambda}}p'_w\geq 1/2$ holds. Besides, $\mathbb{E}_{\mathcal{A}'}[A]=\sum_{w\in \Sigma}p'_w \cdot \mathbb{E}_{\mathcal{A}_w}[A]\geq \sum_{w\in \Sigma_{2\lambda}}p'_w \cdot \mathbb{E}_{\mathcal{A}_w}[A]$, and similarly $\mathbb{E}_{\mathcal{A}'}[B]\geq \sum_{w\in \Sigma_{2\lambda}}p'_w \cdot \mathbb{E}_{\mathcal{A}_w}[B]$. Let $\hat{w}= \mathop{\arg\min}_
{w\in \Sigma_{2\lambda}}\mathbb{E}_{\mathcal{A}_w}[A]\cdot \mathbb{E}_{\mathcal{A}_w}[B]$. We know:
\begin{align*}
\mathbb{E}_{\mathcal{A}}[A]\cdot \mathbb{E}_{\mathcal{A}}[B]& \geq  \left(\sum_{w\in \Sigma_{2\lambda}}p'_w \cdot \mathbb{E}_{\mathcal{A}_w}[A]\right)\left(\sum_{w\in \Sigma_{2\lambda}}p'_w \cdot \mathbb{E}_{\mathcal{A}_w}[B]\right)\\
& \geq \left(\sum_{w\in \Sigma_{2\lambda}}p'_w \cdot \mathbb{E}_{\mathcal{A}_{\hat{w}}}[A]\right)\left(\sum_{w\in \Sigma_{2\lambda}}p'_w \cdot \mathbb{E}_{\mathcal{A}_{\hat{w}}}[B]\right)\\
& \geq \Theta(1) \cdot \mathbb{E}_{\mathcal{A}_{\hat{w}}}[A]\cdot \mathbb{E}_{\mathcal{A}_{\hat{w}}}[B]
\end{align*}
where the second inequality is due to Claim \ref{claim-keymath-complex}.

Notice that $\mathcal{A}_{\hat{w}}$ is an oblivious algorithm that succeeds with at least constant probability. Therefore, in the reminder of the proof, we focus on an algorithm $\mathcal{A}$ that is oblivious and succeeds with at least constant probability. We will show if Eve uses strategy $\mathcal{S}$ then it must be the case that $\mathbb{E}[A]\cdot\mathbb{E}[B]\in\Omega(T)$.

For each slot $i\geq 1$ and each channel $k\in[C]$, let $a_{i,k}$ denote the probability that Alice broadcasts $m$ on channel $k$ in slot $i$, and $b_{i,k}$ denote the probability that Bob does not broadcast but listens on channel $k$ in slot $i$. By definition of $\mathcal{S}$, Eve jams channel $k$ in slot $i$ iff $a_{i,k}\cdot b_{i,k}>1/T$ and she has not depleted her energy.
For the ease of presentation, if in a slot $i$, Eve's remaining energy cannot jam all channels satisfying $a_{i,k}\cdot b_{i,k}>1/T$, then she is allowed to jam all such channels and then stops jamming for good.
Further define $J_{i,k}$ to be an indicator random variable taking value one iff $a_{i,k}\cdot b_{i,k}>1/T$, and let $J=\{(i,k):J_{i,k}=1\}$. We consider two complement scenarios: $|J|<T$ and $|J|\geq T$.

\bigskip\maltese\xspace\textbf{\textsc{Scenario I}:} $|J|<T$.

By definition, for each $(\hat{i},\hat{k})\in J$, Eve will jam channel $\hat{k}$ in slot $\hat{i}$ if she has not depleted her energy at that point. Since $\mathcal{A}$ is oblivious, setting $a_{\hat{i},\hat{k}}=b_{\hat{i},\hat{k}}=0$ for each $(\hat{i},\hat{k})\in J$ will not affect the outcome of any execution or increase $\mathbb{E}[A]\cdot\mathbb{E}[B]$. Thus, we can assume $a_{i,k}\cdot b_{i,k}\leq 1/T$ always holds in $\mathcal{A}$.

For every non-negative integer $q$, define $M_q=\{(i,k):a_{i,k}\cdot b_{i,k}\in(1/(2^{q+1}T),1/(2^qT)]\}$, and $M=\bigcup_{q\in\mathbb{N}}M_q$. Let $\alpha^{(i)}$ be the probability that Alice and Bob are still active in slot $i$, thus $\alpha^{(i)}=\prod_{j=1}^{i-1}(1-\sum_{k=1}^{C}a_{j,k}b_{j,k})\geq\exp(-2\cdot\sum_{j=1}^{i-1}\sum_{k=1}^{C}a_{j,k}b_{j,k})$. Since algorithm $\mathcal{A}$ succeeds with at least constant probability, there must exist a minimum positive integer $t$ such that $\sum_{j=1}^{t}\sum_{k=1}^{C}a_{j,k}b_{j,k}\geq d$ for some constant $d>0$. (Otherwise $1-\alpha_{i}$ would always be $o(1)$, contradicting $\mathcal{A}$ succeeds with at least constant probability.) We claim:
\begin{align}
\mathbb{E}[A]\cdot\mathbb{E}[B]\geq\left(\sum_{j=1}^{t}\sum_{k=1}^{C}a_{j,k}b_{j,k}\right)^2\cdot\Theta(T)\label{eqn-1}
\end{align}

To prove Equation \ref{eqn-1}, define $M_q^t=\{(i,k):(i,k)\in M_q\textrm{ and }i\leq t\}$ for every $q\in\mathbb{N}$, also notice that $\alpha^{(t)}\geq\exp(-2\cdot\sum_{j=1}^{t-1}\sum_{k=1}^{C}a_{j,k}b_{j,k})>e^{-2d}$, thus:
\begin{align*}
\mathbb{E}[A]\cdot\mathbb{E}[B] & \geq \left(\sum_{i\in[t],k\in[C]}\alpha^{(i)}\cdot a_{i,k}\right)\left(\sum_{i\in[t],k\in[C]}\alpha^{(i)}\cdot b_{i,k}\right)\\
& \geq \left(\alpha^{(t)}\right)^2 \cdot \sum_{q\in\mathbb{N}}\left(\sum_{(i,k)\in M_q^t}\sum_{(i',k')\in M_q^t} a_{i,k}\cdot b_{i',k'}\right)\\
& \geq \Theta(1) \cdot \sum_{q\in\mathbb{N}}\left(|M_q^t|^2\cdot\left(\prod_{(i,k)\in M_q^t}(a_{i,k}\cdot b_{i,k})^{|M_q^t|}\right)^{1/|M_q^t|^2}\right)\\
& \geq \Theta(1) \cdot \sum_{q\in\mathbb{N}}\left(|M_q^t|^2\cdot\left(\prod_{(i,k)\in M_q^t}\left(\frac{1}{2^{q+1}T}\right)^{|M_q^t|}\right)^{1/|M_q^t|^2}\right) = \Theta(1) \cdot \sum_{q\in\mathbb{N}}\frac{|M_q^t|^2}{2^{q+1}T}
\end{align*}
where the third inequality is due to the AM-GM inequality.

Now, observe that $\sum_{q\in\mathbb{N}}|M_q^t|^2/(2^{q+1}T)=\sum_{q\in\mathbb{N}}|M_q^t|\cdot(|M_q^t|/2^{q+1}T)\geq\sum_{q\in\mathbb{N}}(|M_q^t|/2\cdot\sum_{(i,k)\in M_q^t}a_{i,k}b_{i,k})$. To minimize $\sum_{q\in\mathbb{N}}(|M_q^t|/2\cdot\sum_{(i,k)\in M_q^t}a_{i,k}b_{i,k})$, we should let $M_q^t$ be an empty set for every $q\in\mathbb{N}^+$ and let each $(\hat{i},\hat{k})$ pair in $M_0^t$ be $a_{\hat{i},\hat{k}}\cdot b_{\hat{i},\hat{k}}=1/T$. (Intuitively, if we remove some $(i,k)$ pairs from $M_q^t$ such that $\sum_{(i,k)\in M_q^t}a_{i,k}b_{i,k}$ is decreased by $\delta$, and then add some $(i,k)$ pairs to $M_{q'}^t$ where $q'<q$ such that $\sum_{(i,k)\in M_{q'}^t}a_{i,k}b_{i,k}$ is increased by $\delta$, then the number of $(i,k)$ pairs we remove from $M_q^t$ must be larger than the number of $(i,k)$ pairs we add to $M_{q'}^t$. Thus, value of $\sum_{q\in\mathbb{N}}(|M_q^t|/2\cdot\sum_{(i,k)\in M_q^t}a_{i,k}b_{i,k})$ is decreased after this adjustment.) Therefore:
$$\sum_{q\in\mathbb{N}}\left(\frac{|M_q^t|}{2}\cdot\sum_{(i,k)\in M_q^t}a_{i,k}b_{i,k}\right)\geq\frac{1}{2}\cdot\frac{\sum_{j=1}^{t}\sum_{k=1}^{C}a_{j,k}b_{j,k}}{1/T}\cdot\left(\sum_{j=1}^{t}\sum_{k=1}^{C}a_{j,k}b_{j,k}\right)$$
which immediately leads to Equation \ref{eqn-1}.

Since Equation \ref{eqn-1} is true, it must be the case $\mathbb{E}[A]\cdot\mathbb{E}[B]\in\Omega(T)$, as $\sum_{j=1}^{t}\sum_{k=1}^{C}a_{j,k}b_{j,k}\in\Theta(1)$. This completes the proof for the first scenario.

\bigskip\maltese\xspace\textbf{\textsc{Scenario II}:} $|J|\geq T$.

This scenario is more complicated, and we begin with some notations. Define $h(i)\triangleq\sum_{k=1}^{C} \mathbb{I}[(i,k)\in J]$ (in particular $h(0)\triangleq 0$) and $sh(i)\triangleq \sum_{j=0}^{i} h(j)$. Let $l$ be the minimum integer satisfying $sh(l)\geq T$. Clearly, such an integer $l$ must exist as $|J|\geq T$.

We use the oblivious algorithm $\mathcal{A}$ to construct another oblivious algorithm $\mathcal{A}'$. Let $\hat{a}_{i,k}$ denote the probability that Alice broadcasts $m$ on channel $k$ in slot $i$ in $\mathcal{A}'$, and $\hat{b}_{i,k}$ denote the probability that Bob does not broadcast but listens on channel $k$ in slot $i$ in $\mathcal{A}'$. We copy the first $l$ slots of $\mathcal{A}$ as the first $l$ slots of $\mathcal{A}'$ and make the following adjustments: for any $1\leq i\leq l, k\in[C]$ in which ``$J_{i,k}=1$'', set $\hat{a}_{i,k}=\hat{b}_{i,k}=0$. We then make another copy of the first $l$ slots of $\mathcal{A}$, use them as the $(l+1)$-st to $2l$-th slots of $\mathcal{A}'$, and make the following adjustments: for any $l+1\leq i\leq 2l, k\in[C]$ in which ``$J_{i-l,k}=0$'', set $\hat{a}_{i,k}=\hat{b}_{i,k}=0$. For every $i>2l$, slot $i$ in $\mathcal{A}'$ is a copy of slot $i-l$ in $\mathcal{A}$.

We now prove $\mathbb{E}_\mathcal{A}[A]\geq\mathbb{E}_{\mathcal{A}'}[A]$. Let $\alpha^{(i)}_{\mathcal{A}}$ (resp., $\alpha^{(i)}_{\mathcal{A}'}$) denote the probability that Alice and Bob are still active in slot $i$ in $\mathcal{A}$ (resp., $\mathcal{A}'$). By construction, for every $1\leq i\leq l+1$, $\alpha^{(i)}_{\mathcal{A}}=\alpha^{(i)}_{\mathcal{A}'}$; and for every $l+1\leq i\leq 2l$, $\alpha^{(i)}_{\mathcal{A}'}=\alpha^{(i+1)}_{\mathcal{A}'}$. Thus $\alpha^{(l+1)}_{\mathcal{A}}=\alpha^{(2l+1)}_{\mathcal{A}'}$. Moreover, for every $i>2l$, slot $i$ in $\mathcal{A}'$ is identical to slot $i-l$ in $\mathcal{A}$. Hence, to compare $\mathbb{E}_\mathcal{A}[A]$ and $\mathbb{E}_{\mathcal{A}'}[A]$, it suffices to compare Alice's expected cost during the first $l$ slots in $\mathcal{A}$ and Alice's expected cost during the first $2l$ slots in $\mathcal{A}'$. Let $\mathbb{E}^{[1,l]}_\mathcal{A}[A]$ (resp., $\mathbb{E}^{[1,2l]}_{\mathcal{A}'}[A]$) denote Alice's expected cost in the first $l$ (resp., $2l$) slots in $\mathcal{A}$ (resp., $\mathcal{A}'$). We know:
\begin{align*}
\mathbb{E}^{[1,l]}_\mathcal{A}[A] & = \sum_{i=1}^{l}\sum_{k=1}^{C}\alpha_{\mathcal{A}}^{(i)}\cdot a_{i,k} = \left( \sum_{i=1}^{l}\sum_{(i,k)\notin J}\alpha_{\mathcal{A}}^{(i)}\cdot a_{i,k} \right) + \left( \sum_{i=1}^{l}\sum_{(i,k)\in J}\alpha_{\mathcal{A}}^{(i)}\cdot a_{i,k} \right) \\
& \geq \left( \sum_{i=1}^{l}\sum_{(i,k)\notin J}\alpha_{\mathcal{A}'}^{(i)}\cdot \hat{a}_{i,k} \right) + \left( \sum_{i=1}^{l}\sum_{(i,k)\in J}\alpha_{\mathcal{A}}^{(l+1)}\cdot a_{i,k} \right) \\
& = \left( \sum_{i=1}^{l}\sum_{(i,k)\notin J}\alpha_{\mathcal{A}'}^{(i)}\cdot \hat{a}_{i,k} \right) + \left( \sum_{i=1}^{l}\sum_{(i,k)\in J}\alpha_{\mathcal{A}'}^{(l+1)}\cdot a_{i,k} \right) \\
& = \left( \sum_{i=1}^{l}\sum_{(i,k)\notin J}\alpha_{\mathcal{A}'}^{(i)}\cdot \hat{a}_{i,k} \right) + \left( \sum_{i=l+1}^{2l}\sum_{(i-l,k)\in J}\alpha_{\mathcal{A}'}^{(i)}\cdot \hat{a}_{i,k} \right) = \mathbb{E}^{[1,2l]}_{\mathcal{A}'}[A]
\end{align*}
Thus $\mathbb{E}_\mathcal{A}[A]\geq\mathbb{E}_{\mathcal{A}'}[A]$, and proof of $\mathbb{E}_\mathcal{A}[B]\geq\mathbb{E}_{\mathcal{A}'}[B]$ is similar.

We now focus on proving $\mathbb{E}_{\mathcal{A}'}[A]\cdot\mathbb{E}_{\mathcal{A}'}[B]\in\Omega(T)$. To prove this, we first modify $\mathcal{A}'$ so that for every $1\leq i\leq l$ and $k\in[C]$, either $\hat{a}_{i,k}\cdot \hat{b}_{i,k}=1/T$ or $\hat{a}_{i,k}=\hat{b}_{i,k}=0$; and these modifications will not increase $\mathbb{E}_{\mathcal{A}'}[A]\cdot\mathbb{E}_{\mathcal{A}'}[B]$. More specifically, examine $\mathcal{A}'$ slot by slot, and let slot $\hat{i}$ be the first slot in which ``$\hat{a}_{\hat{i},k}\cdot \hat{b}_{\hat{i},k}=1/T$ or $\hat{a}_{\hat{i},k}=\hat{b}_{\hat{i},k}=0$'' does not hold for some $k\in[C]$. Now, let $P_A$ (respectively, $P_B$) be the expected cost of Alice (respectively, Bob) in the first $\hat{i}-1$ slots; let $\alpha$ be the probability that Alice and Bob are still active in slot $\hat{i}$; and let $S_A$ (respectively, $S_B$) be the expected cost of Alice (respectively, Bob) after slot $\hat{i}$ conditioned on the algorithm does not halt by the end of slot $\hat{i}$. We know $\mathbb{E}_{\mathcal{A}'}[A]\geq P_A+\alpha(\sum_{k=1}^{C}{\hat{a}_{\hat{i},k}}+S_A\cdot(1-\sum_{k=1}^{C}{\hat{a}_{\hat{i},k}\cdot \hat{b}_{\hat{i},k}}))=(P_A+\alpha S_A)+\alpha\cdot\sum_{k=1}^{C}{\hat{a}_{\hat{i},k}(1-\hat{b}_{\hat{i},k}\cdot S_A)}$. Similarly, $\mathbb{E}_{\mathcal{A}'}[B]\geq(P_B+\alpha S_B)+\alpha\cdot\sum_{k=1}^{C}{\hat{b}_{\hat{i},k}(1-\hat{a}_{\hat{i},k}\cdot S_B)}$. Now, if ``$\hat{a}_{\hat{i},k}\cdot\hat{b}_{\hat{i},k}=1/T$ or $\hat{a}_{\hat{i},k}=\hat{b}_{\hat{i},k}=0$'' does not hold for some $\hat{k}\in[C]$, then we can always decrease $\mathbb{E}_{\mathcal{A}'}[A]\cdot\mathbb{E}_{\mathcal{A}'}[B]$ via one of the following adjustments:
\begin{itemize}
	\item If $\hat{b}_{\hat{i},\hat{k}}\cdot S_A\geq 1$, then increasing the value of $\hat{a}_{\hat{i},\hat{k}}$ to $1/(T\cdot\hat{b}_{\hat{i},\hat{k}})$ decreases $\mathbb{E}_{\mathcal{A}'}[A]$. Besides, $\mathbb{E}_{\mathcal{A}'}[B]$ also decreases with the increase of $\hat{a}_{\hat{i},\hat{k}}$.\footnote{There is a technical subtlety worth clarifying. It might be the case that $\hat{a}_{\hat{i},\hat{k}}$ cannot reach $1/(T\cdot\hat{b}_{\hat{i},\hat{k}})$ since $\hat{a}_{\hat{i},\hat{k}}$ is a probability. Nonetheless, the value of $\mathbb{E}_{\mathcal{A}'}[A]\cdot\mathbb{E}_{\mathcal{A}'}[B]$ when $\hat{a}_{\hat{i},\hat{k}}=1/(T\cdot\hat{b}_{\hat{i},\hat{k}})$ always provides a lower bound for the actual $\mathbb{E}_{\mathcal{A}'}[A]\cdot\mathbb{E}_{\mathcal{A}'}[B]$.}
	\item If $\hat{a}_{\hat{i},\hat{k}}\cdot S_B\geq 1$, then similar to above, increasing the value of $\hat{b}_{\hat{i},\hat{k}}$ to $1/(T\cdot\hat{a}_{\hat{i},\hat{k}})$ will decrease both $\mathbb{E}_{\mathcal{A}'}[A]$ and $\mathbb{E}_{\mathcal{A}'}[B]$.
	\item If $\hat{b}_{\hat{i},\hat{k}}\cdot S_A<1$ and $\hat{a}_{\hat{i},\hat{k}}\cdot S_B<1$, then setting both $\hat{a}_{\hat{i},\hat{k}}$ and $\hat{b}_{\hat{i},\hat{k}}$ to zero will decrease both $\mathbb{E}_{\mathcal{A}'}[A]$ and $\mathbb{E}_{\mathcal{A}'}[B]$.
\end{itemize}

In the reminder of the proof, assume for every $1\leq i\leq l$ and $k\in[C]$, either $\hat{a}_{i,k}\cdot \hat{b}_{i,k}=1/T$ or $\hat{a}_{i,k}=\hat{b}_{i,k}=0$. We continue to prove $\mathbb{E}_{\mathcal{A}'}[A]\cdot\mathbb{E}_{\mathcal{A}'}[B]\in\Omega(T)$.

Let $\tilde{\mathbb{E}}^{[j,j']}_{\mathcal{A}'}[A]$ (resp., $\tilde{\mathbb{E}}^{[j,j']}_{\mathcal{A}'}[B]$) denote Alice's (resp., Bob's) expected cost in slots $j$ to $j'$ in $\mathcal{A}'$, conditioned on Alice and Bob are still active at the beginning of slot $j$ in $\mathcal{A}'$. Notice, in $\mathcal{A}'$, if Alice and Bob are still active at the beginning of slot $l+1$, then by construction they cannot succeed in slots $l+1$ to $2l$. Let $J^l=\{(i,k)\in J:i\leq l\}$, we have:
\begin{align*}
\tilde{\mathbb{E}}^{[l+1,2l]}_{\mathcal{A}'}[A]\cdot\tilde{\mathbb{E}}^{[l+1,2l]}_{\mathcal{A}'}[B] & \geq \left( \sum_{(i,k)\in J^l}a_{i,k} \right) \left( \sum_{(i,k)\in J^l}b_{i,k} \right) = \sum_{(i,k)\in J^l}\sum_{(i',k')\in J^l}a_{i,k}\cdot b_{i',k'} \\
& \geq \left|J^l\right|^2 \cdot \left(\prod_{(i,k)\in J^l}(a_{i,k}\cdot b_{i,k})^{|J^l|}\right)^{1/|J^l|^2} \\
& > T^2\cdot(1/T)=T
\end{align*}
where the second to last inequality is due to the AM-GM inequality.

In general, for any $1\leq j\leq l$, if we assume $\tilde{\mathbb{E}}^{[j+1,2l]}_{\mathcal{A}'}[A]\cdot\tilde{\mathbb{E}}^{[j+1,2l]}_{\mathcal{A}'}[B]>T$, then we also have $\tilde{\mathbb{E}}^{[j,2l]}_{\mathcal{A}'}[A]\cdot\tilde{\mathbb{E}}^{[j,2l]}_{\mathcal{A}'}[B]>T$. To see this, for every $1\leq j\leq l$, define $M_1^j=\{(j,k):\hat{a}_{j,k}\cdot\hat{b}_{j,k}=1/T\}$, and let $p_j=1-|M_1^j|/T$, then we know:

\begin{small}\begin{align*}
\tilde{\mathbb{E}}^{[j,2l]}_{\mathcal{A}'}[A]\cdot\tilde{\mathbb{E}}^{[j,2l]}_{\mathcal{A}'}[B] & \geq
\left(\left(\sum_{(j,k)\in M_1^j}\hat{a}_{j,k}\right)+p_j\cdot\tilde{\mathbb{E}}^{[j+1,2l]}_{\mathcal{A}'}[A]\right) \left(\left(\sum_{(j,k)\in M_1^j}\hat{b}_{j,k}\right)+p_j\cdot\tilde{\mathbb{E}}^{[j+1,2l]}_{\mathcal{A}'}[B]\right)\\
& > p_j^2\cdot T + \left(\sum_{(j,k)\in M_1^j}\hat{a}_{j,k}\right)\left(\sum_{(j,k)\in M_1^j}\hat{b}_{j,k}\right)\\
& \phantom{>} + p_j\cdot\left(\left(\sum_{(j,k)\in M_1^j}\hat{b}_{j,k}\right)\tilde{\mathbb{E}}^{[j+1,2l]}_{\mathcal{A}'}[A] + \left(\sum_{(j,k)\in M_1^j}\hat{a}_{j,k}\right)\tilde{\mathbb{E}}^{[j+1,2l]}_{\mathcal{A}'}[B]\right) \\
& \geq p_j^2\cdot T+|M_1^j|^2/T\\
& \phantom{\geq} + p_j \cdot (2|M_1^j|) \cdot \left(\left(\prod_{(j,k)\in M_1^j}\hat{a}_{j,k}\cdot\hat{b}_{j,k}\right) \cdot \left(\tilde{\mathbb{E}}^{[j+1,2l]}_{\mathcal{A}'}[A]\cdot\tilde{\mathbb{E}}^{[j+1,2l]}_{\mathcal{A}'}[B]\right)^{|M_1^j|}\right)^{1/(2|M_1^j|)}\\
& > T\left(p_j^2+|M_1^j|^2/\left(T^2\right)+2p_j\cdot |M_1^j|/T\right)\\
& = T\cdot(p_j+|M_1^j|/T)^2=T
\end{align*}\end{small}
where the second to last inequality is due to the AM-GM inequality.

Finally, via induction we can conclude $\tilde{\mathbb{E}}^{[1,2l]}_{\mathcal{A}'}[A]\cdot\tilde{\mathbb{E}}^{[1,2l]}_{\mathcal{A}'}[B]>T$. Since $\mathbb{E}_{\mathcal{A}'}[A]\cdot\mathbb{E}_{\mathcal{A}'}[B] \geq \tilde{\mathbb{E}}^{[1,2l]}_{\mathcal{A}'}[A]\cdot\tilde{\mathbb{E}}^{[1,2l]}_{\mathcal{A}'}[B]$, the proof for the second scenario is completed.
\end{proof}

\end{document}